\documentclass[12pt]{article}
\usepackage[margin=1in]{geometry}
\usepackage[round,compress, sort, authoryear]{natbib}
\usepackage{setspace}
\usepackage{xcolor}
\usepackage[hidelinks]{hyperref}
\usepackage{amsmath,amsthm,amssymb}
\usepackage{enumitem}
\usepackage{cleveref}
\newtheorem{theorem}{Theorem}

\usepackage{rotating,multirow,makecell,array}
\usepackage {tikz}
\usepackage{tikz}
\usetikzlibrary{arrows.meta}

\usepackage{float}
\usepackage{pgfplots}
\usepackage{mathtools}
\usepgfplotslibrary{groupplots}
\usepackage{textcomp}
\usetikzlibrary{calc} 
\usepackage{subcaption}
\usepackage{algorithm}
\usepackage{algpseudocode}

\usetikzlibrary {positioning}
  \newtheorem{cor}{Corollary}
    \newtheorem{lem}{Lemma}
  \newtheorem{prop}{Proposition}
\theoremstyle{definition}
 \newtheorem{defn}{Definition}
  \newtheorem{ass}{Assumption}
 
  \usepackage{multibib}
  
\usepackage[utf8]{inputenc}   
\usepackage[T1]{fontenc}      
\usepackage{graphicx}         
\usepackage{booktabs}         
\usepackage{epstopdf} 
\usepackage{siunitx}
\usepackage{booktabs}   
\usepackage{multirow}   
\usepackage{booktabs}
\usepackage{multirow}      
\usepackage{makecell}      
 \newtheorem{exmp}{Example}
 \newtheorem{rem}{Remark}
 
\usepackage{booktabs}   
\usepackage{amsmath}    

\usepackage{caption}              
\ifxetex
  \usepackage{mathspec}
\fi

\usepackage{titlesec}
\titlespacing*{\paragraph}{2pt}{10pt}{1em}

\DeclareMathOperator*\argmin{arg\,min}

\title{When is the statistical evidence strong enough? \\ Using hypothesis tests to value data collection\footnote{Davide Viviano acknowledges support by the Harvard Griffin Fund in economics. We thank Kevin Chen, Lim Cheaheon, Ed Glaeser, Toru Kitagawa, Yiqi Liu, Craig Mcintosh, Paul Niehaus, David Ritzwoller, Jesse Shapiro and Tomasz Strzalecki for helpful comments. We used Chat-GPT for research assistance in some mathematical derivations. All mistakes are our own. A companion R-package is available at \url{https://github.com/dviviano/avalue}. }
}
\author{Aristotelis Epanomeritakis%
\footnote{\raggedright Department of Economics, Harvard University. Email address: \texttt{aristotle\_epanomeritakis@fas.harvard.edu}.}
\and Davide Viviano%
\footnote{Department of Economics, Harvard University.  Email address: \texttt{dviviano@fas.harvard.edu}.}
}
\date{October, 2026}

\begin{document}

\maketitle

\vspace{-24pt}

\begin{abstract}
Policy decisions often hinge on conventional \(p\)-value thresholds, which ignore economic costs and benefits of further data collection. 
This paper recasts statistical significance as a choice between making an immediate policy recommendation and deferring it until further evidence is collected. The welfare-optimal decision corresponds, under minimax regret, to a statistical test whose level depends on the cost and precision of additional evidence. Inverting this rule, we introduce and recommend reporting the abstention-value (\(A\)-value) to determine where additional data collection is most needed. The $A$-value defines the break-even welfare cost of abstaining and recommending further experimentation after observing the initial evidence. Computing the $A$-value only requires a point estimate and its standard error, imposes no prior assumptions on policy effects, and can be  converted to a monetary research budget using inputs already frequently used by regulators and funding agencies. When experimentation capacity is limited, we show that prioritizing additional data collection where \(A\)-values are the largest yields strong finite-sample guarantees. Applications to economic program evaluation illustrate the benefits of using $A$-values.
\end{abstract}

\newpage 

\onehalfspacing

\section{Introduction}

Statistical significance plays a prominent role in evidence-based policy.
From the regulatory approval of new drugs \citep{guidance2017multiple}
to the adoption of economic programs \citep{gertleretal2016},
researchers routinely compare \(p\)-values with conventional thresholds
to determine whether evidence is ``strong enough'' to implement a
new policy. Yet the economic rationale for this practice remains
unclear and actively debated
\citep[e.g.][]{wassersteinlazar2016,benjaminetal2018,amrheinetal2019,
abadie_nonsignificance,imbens2021,frankel2022findings}. A common approach is to frame standard hypothesis testing as a
comparison between implementing the policy
and retaining the status quo. However, in this approach, conventional significance levels are difficult to rationalize for policies with different welfare implications  \citep{tetenov2016testing, manski2019inference}, and may not capture key economic trade-offs when ``failure to reject'' reflects insufficient evidence to support an immediate policy recommendation, rather than evidence in favor of the status quo \citep{fisher1935design}. 

For instance, consider an experiment to evaluate a cash-transfer program returning a large positive estimated effect (net of implementation costs) and a confidence interval that includes zero  at a standard nominal level. A mistaken recommendation can forgo the benefits of an effective program or waste resources on an ineffective one. Researchers may therefore recommend collecting additional evidence to improve the recommendation, but obtaining it is costly. Whether the existing evidence is sufficient for an immediate decision therefore depends not only on the welfare consequences of a mistake, but also on the cost of additional experimentation. Conventional $p$-values may fail to reflect these economic trade-offs. 

This paper re-interprets statistical significance through the lens of the cost incurred by declaring ``we do not know yet'', defined as the cost of deferring a policy recommendation and collecting further evidence. 
We characterize the trade-off between an immediate recommendation
and further experimentation by embedding the classical binary
testing problem \citep{neymanpearson1933} in a two-stage sequential
decision as in \citet{wald1947sequential} (and multi-stage in Appendix \ref{sec:multiple_periods}). Our main contribution is to introduce and establish the properties
of the abstention value ($A$-value), the break-even cost
of further experimentation given the existing evidence. 
Researchers can report this statistic without specifying
experimentation costs, allowing readers to compare it with their
own cost assessments. 
We show that the $A$-value guides both researchers choosing
between an immediate recommendation or experimentation
and funding agencies allocating a research budget across multiple policies. 




Consider first a researcher who studies a \textit{single policy} with unknown welfare effect \(\Delta\) relative to the status quo, net of implementation costs.  The researcher observes an unbiased (asymptotically) Gaussian estimate \(X_1\), with standard error \(s_1\), and chooses either to recommend  implementation, the status quo, or deferring the recommendation. Deferral entails a prespecified follow-up study that produces an independent estimate of the same effect, with standard error $s_2$, after which a final recommendation is made using the accumulated evidence. The follow-up design carries a fixed cost \(c\) measured in the same units as $\Delta$.\footnote{Appendix \ref{rem:endogenous_precision} studies settings with $s_2$ optimally chosen ex-ante, when funding agencies or grant-makers ex-ante commit to a given budget for follow-up studies. Our results continue to hold in this setting.}

The researcher has limited information about $\Delta$, as prior evidence about the same program is limited or not available. For instance, evidence from other regions or programs may potentially return very different effects from the current program.
We therefore evaluate decision rules by their worst-case expected
regret \citep{savage1951},  relative to a decision-maker
who knows \(\Delta\) and chooses immediately.
It equals the expected welfare loss from an incorrect policy recommendation plus any
experimentation costs incurred, worst-case over $\Delta$.

The first step is to 
characterize the optimal decision as a (unique) threshold rule familiar from the classical testing literature \citep[e.g.][]{wald1950statistical}. 
Sufficiently positive evidence leads to implementation, sufficiently
negative evidence to retaining the status quo, and intermediate
evidence to experimentation.
We construct the \(A\)-value by inverting this threshold: holding the observed estimate and the standard errors fixed, we find the smallest cost $c$ at which the optimal rule makes an immediate recommendation after observing $X_1$.
We establish the monotonicity of the threshold in cost so that any minimax-regret rule recommends further experimentation \textit{if and only if} its cost is below the \(A\)-value. 
Researchers can report the statistic for readers to compare with their own cost assessments. We propose an equally precise follow-up, \(s_2=s_1\), as a simple benchmark when no particular design has been specified (practitioners may refer to Section \ref{sec:conclusion} for an explicit guide to practice). 

We establish two main properties of the \(A\)-value. First, holding experimental precision fixed, it increases with the \(p\)-value, i.e., weaker statistical evidence makes additional information more valuable. Second, holding the \(p\)-value fixed, the \(A\)-value is smaller when the experimental design is more precise. The intuition is that with higher precision, the policy effects that remain difficult to distinguish from zero are smaller in welfare units. The worst-case welfare of an incorrect immediate decision, and hence the value of collecting further evidence, are therefore smaller. Thus, unlike the \(p\)-value, the \(A\)-value reflects both the strength of the statistical evidence and the welfare implication of the policy decision.

The $A$-value requires no prior assumptions on policy effects and is only a function of the standard error and point estimate. Section \ref{sec:examples} shows how to convert it to a monetary research budget using inputs already frequently used by regulators and funding agencies.

Our next results concern the allocation of a fixed research budget across \textit{$J$ many policies}. For example, a funding agency may support follow-up studies for
a subset of programs with initial pilot studies. 
We evaluate allocations by (weighted) worst-case average regret, with experimentation costs entering through the budget constraint. We propose to fund additional data collection for the highest \(A\)-values per unit of research expenditure. 
We show that its excess worst-case regret relative to the best rule facing the same budget is of order \(J^{-1/2}\), allowing policy effects to differ arbitrarily
across studies. In the relevant finite $J$ regime, we bound excess regret relative to the oracle that can collect additional data at \textit{no cost} in terms of $A$-values. Thus, the same statistic guides
research allocation and provides an upper bound on the welfare cost
of limited research capacity. For fixed $J$, our bounds attain the
minimax-optimal rate as within-study sample sizes grow proportionally.

We demonstrate the applicability of our proposed solution 
we reanalyze 42 evaluations of unconditional cash-transfer programs \citep{crosta2024unconditional},
measuring welfare by their effects on consumption.
We consider follow-up studies as precise as their initial
evaluations. The distribution of $A$-values is strongly right-skewed:
for an illustrative \$10 million rollout intervention  \citep[matching][]{egger2022general, kahura2022givedirectly} with consumption gains
over one year, the median break-even cost is about
\$3,500, compared with a mean of \$310,000 in consumption-equivalent
dollars. At a follow-up cost of \$500,000, the rule recommends 
experimentation for three programs, suggesting strong evidence from most programs.

When funding ten follow-up experiments, the disagreement between the $A$-value rule and the largest $p$-values or  standard errors ranking illustrate the importance of accounting jointly for statistical inconclusiveness and the scale of uncertainty, with some disagreements implying up to \$168,000 in break-even experimentation costs. 
 Similar differences also arise when comparing with an empirical Bayes (EB) benchmark. While useful in the presence of many studies drawn from the same distribution \citep[e.g.][]{abadie2023estimating}, EB can be sensitive to the prior when a small number of comparable studies make non-parametric approaches infeasible, or when structural priors are misspecified \citep[e.g.][]{schorfheide2025uncertainty}. Here, key differences arise because, for programs with imprecise estimates, empirical Bayes places most of the weight on the fitted prior; additional data collection is unlikely to change the EB recommendations. Without imposing a common prior,  
the \(A\)-value instead recommends experimentation especially in settings where both $p$-values and standard errors are the largest, due to the uncertainty on welfare effects. Calibrated simulations to the empirical estimates illustrate corresponding welfare gains of using the $A$-value.



\paragraph{Related literature} This paper builds on the classical literature on hypothesis testing  \citep{arrowblackwellgirshick1949,waldwolfowitz1950bayes,wald1950statistical, kaiser1960}, which develop the foundations for our problem. 
We connect hypothesis testing to the classical literature on sequential experimentation \citep[e.g.][]{fisher1952sequential,robbins1952aspects,sobel1953,maurice1957,degroot1960,chernoff1961normal,johnsonmaurice1963}, the recent literature on best-arm identification  \citep[e.g.][]{kasy2021adaptive, manskitetenov2016, garivier2016optimal,qin2017improving}, and to the generalized Wald in continuous-time studied by \citet{adusumilli2026wald}.\footnote{Note however that continuous time solutions differ from discrete batch solutions since simply discretizing continuous-time rules is not minimax optimal \citet[Section 6.1]{adusumilli2026wald}.} 
The three-action structure also connects to the literature on classification with a rejection option, where ``rejection'' means declining to classify \citep{chow1970optimum,cortes2016learning,bartlettwegkamp2008,elyanivwiener2010,ricekrakauer2023, mozannarsontag2020, trapeznikov2013supervised}.

Our main contribution relative to all these references is to introduce and establish properties of the $A$-value, as opposed to characterizing optimal (stopping) rules for a single decision. Using the $A$-value we derive asymptotically optimal experimentation rules when choosing multiple non-exclusive policies under a common budget constraint.
A key technical challenge is that the least-favorable distribution changes as the cost changes, making the posterior least-favorable distribution and the corresponding expected welfare gains change with the cost.
As a result, monotonicity of the minimax threshold in the
cost does not follow from its posterior characterization. We  establish monotonicity through a serie of lemmas after providing the exact expression of the threshold for our problem. We then use these properties to establish regret guarantees for choosing among multiple non-exclusive decisions.

Statistical testing from a decision-theoretic perspective has a long tradition in econometrics. 
The difficulty of motivating hypothesis testing in binary treatment-choice problems \citep{manski2004, stoye2009, manski2021econometrics, manski2019inference, kitagawatetenov2018}, has motivated the literature to (i) either take Type I error control as an exogenous part of researcher's preferences, and study approximately optimal decisions up-to the probability of making Type I error mistake \citep{andrews2025certified}; (ii) or study hypothesis testing in communication models involving multiple (often misaligned) stakeholders \citep[e.g.][]{mccloskey2026critical, tetenov2016testing,frankel2022findings,vivianoetal2026, jagadeesan2025publication, bates2026principal}; (iii) or focus on single-agent binary-decisions with highly asymmetric losses \citep[e.g.][]{das1994decision, kimchoi2021, tetenov2012}. All such references focus on a terminal (policy) decision. We complement this literature 
by studying statistical testing as a choice between immediate decision and future data collection. This allows us to characterize the significance threshold endogenously as a function of the economic cost of future data. With prohibitively costly data collection, we recover \cite{manski2004}'s empirical success rule.

The \(A\)-value is conceptually distinct from other statistics in the testing literature. For instance, the \(q\)-value  \citep{storey2002} measures the minimum positive false discovery rate associated with treating a new discovery. Unlike the \(A\)-value, it does not account for the cost of future evidence.
 An \(e\)-value is a nonnegative statistic whose expectation under the null is at most one. Sequences of \(e\)-values are used to preserve Type I error control when collecting additional observations based on previous evidence \citep{vovkwang2021,ramdas2025hypothesis}. $e$-values do not determine whether the value of additional information justifies its cost, and therefore they differ in their scope and construction from the \(A\)-value.

The $A$-value is instead related to  notions of economic value of experimentation studied under risk criteria that average with respect to a chosen (prior) distribution. Classical contributions include \citet{grundyreeshealy1954,grundyhealyrees1956,raiffaschlaifer1961,batherwalker1962,scott1968,schleifer1969}, while recent (empirical) Bayes analyses include \citet{brown2025valuing, abadie2023estimating, hendren2020unified, henry2019research, fudenberg2018speed,liang2022dynamically,morris2019wald,azevedoetal2023,azevedo2020b}. Here, we evaluate experimentation with minimax regret, connecting to classical frequentist hypotheses testing. The minimax perspective is particularly useful when researchers do not rely on or have access to a prior over studies. The \(A\)-value serves this purpose as a frequentist analog to Bayesian values of information; instead of computing expected
welfare under a fixed prior, the \(A\)-value inverts optimal decision rules under \textit{a family} of least-favorable priors indexed by the cost
of experimentation to find the break-even cost. 
(Appendix \ref{sec:prior_information} also studies $A$-values that balance minimax regret and Bayes risk.)


Finally, recent work in econometrics uses abstention to address other decision problems. \citet{horowitzlee2025} study partial identification in a maximum-score model and develop a computationally tractable linear-programming procedure. \citet{brezaetal2025} use an option to admit ignorance for heterogeneous causal effects. We complement this literature with an explicit characterization of optimal abstention threshold with sequential decisions which is key to study its connection to hypothesis testing and the $A$-value introduced in this paper.

\section{From statistical testing to economic decisions}


For a given policy question, a planner (or researcher) may reject or implement the given policy. If implemented, the policy generates a welfare effect $\Delta$, ex-ante unknown to the planner. We standardize welfare at status quo to be equal to zero.

\begin{ass}[Unbiased experiment] \label{ass:experiment} Researchers collect a Gaussian unbiased estimate of the welfare effect
\begin{equation*}
    X_1 \sim \mathcal{N}\left(\Delta,s_1^2\right),
\end{equation*}
where $s_1^2 \in (0,\infty)$ is assumed to be known. Let $\Delta \in [-\bar{\Delta}, \bar{\Delta}], \bar{\Delta}/s_1 \ge 0.76$, $\bar{\Delta} \in (0,\infty)$ with $\bar{\Delta}$ denoting a finite upper bound on welfare effects potentially \textit{unknown} to researchers. 
\end{ass} 

Assumption \ref{ass:experiment} states that the researcher conducts an experiment with an unbiased signal of welfare. Here,  $\Delta$ encodes overall welfare effects of the program, net of costs of implementation \citep[e.g.][]{crosta2024unconditional, hendren2020unified}.
Known $s_1^2$ can be replaced by a consistent estimate for $s_1^2$ for an asymptotic analysis. Gaussianity can be thought of as an asymptotic approximation under a central limit theorem. The upper bound $\bar{\Delta}$ guarantees a compact parameter space whose boundaries can remain unknown to researchers. 

An extension with non-Gaussian $X_1$ is in Appendix \ref{sec:non_gaussian_experiments}.

\begin{exmp}[Gaussian approximation] \label{exmp:asymptotic}
In a two-arm experiment, $X_1=\bar Y_t-\bar Y_c$ compares
average outcomes among $n_t$ treated and $n_c$ control units.
Assume independent observations, identically distributed within
each arm, with common outcome variance $\sigma^2\in(0,\infty)$
and difference in means $\Delta$.
Writing $n_1=n_t+n_c$ and $\pi=n_t/n_1$, we have
$ 
s_1^2
=\frac{\sigma^2}{\pi(1-\pi)n_1}.
$ 
If $\pi$ remains bounded away from zero and one as
$n_1\to\infty$, the central limit theorem gives
$ 
\frac{X_1-\Delta}{s_1}
\xrightarrow{d}\mathcal N(0,1),
\frac{\bar\Delta}{s_1}
=\frac{\sqrt{n_1\pi(1-\pi)}\,\bar\Delta}{\sigma}.
$ 
Thus, the Gaussian model provides an asymptotic approximation,
and, for fixed $\bar\Delta>0$, the condition
$\bar\Delta/s_1 \geq0.76$ holds for $n_1 \ge 0.76^2 \sigma^2 \Big/ \Big(\pi(1- \pi) \bar{\Delta}^2\Big)$.
\qed
\end{exmp}

While the standard Neyman-Pearson framework assumes a two-point decision (either reject or accept a null hypothesis), here we study a sequential testing procedure in the spirit of \cite{wald1947sequential}, where researchers may recommend to accept, reject or \textit{collect more evidence}. The third option connects the idea of rejection option in \cite{chow1970optimum} to classical literature on sequential experimentation \citep[e.g.][]{maurice1957, sobel1953}.

\begin{defn}[Rejection option, \cite{chow1970optimum}] Denote $\delta_1$ a decision as a function of evidence $X_1$ with 
$$
\small 
\begin{aligned}
\delta_1(X_1) & = \begin{cases}   1 & \text{ if accept policy} \\ 
 0 & \text{ if reject policy} \\ 
 \texttt{NA} & \text{ if abstain} 
\end{cases} 
\end{aligned}
$$
\end{defn} 
If abstention carried no cost, no immediate decision would be made. The core question is therefore how to model its consequences. Here, we give abstention an explicit sequential interpretation and find optimal decisions under such interpretation.

\begin{ass}[Abstention as a sequential sampling rule] \label{ass:2} Upon admitting ignorance $(\delta_1 = \texttt{NA})$, researchers recommend collecting more evidence. Additional evidence will be collected in the form of $X_2 \sim \mathcal{N}(\Delta, s_2^2), X_2 \perp X_1$ \text{ at cost } $c(s_2^2)$, potentially a function of $s_2^2 \in (0,\infty)$.\footnote{Independence is without loss because with jointly Gaussian and correlated $X_1,X_2$, we can re-define the second experiment signal as $X_2'$ after residualizing $X_1$. Specifically, for $X_1,X_2$ with correlation $|\rho| < 1$ and for $\rho s_2/s_1 \neq 1$, we can define $a = \rho \frac{s_2}{s_1}$ and $X_2' = (X_2 - a X_1)/(1 - a)$, independent of $X_1$ and centered around $\Delta$. The same analysis then continues where we replace $X_2$ with $X_2'$.}  

If $\delta_1 = \texttt{NA}$, the terminal decision takes the form: 
$$ 
\small 
\begin{aligned} 
\delta_2(X_1, X_2)  = \begin{cases}   1 & \text{ if accept policy} \\ 
 0 & \text{ if reject policy} 
\end{cases}
\end{aligned} 
$$ 
where no further experimentation is possible. 
 \qed 
\end{ass}

Assumption \ref{ass:2} gives abstention a welfare interpretation where abstention leads to additional evidence but entails the cost of collecting it. We treat the design of the follow-up experiment, and hence its variance $s_2^2$, as given. This matches common practices of fixed budget constraints for replication studies (e.g., funding agencies provide fixed funding amounts). Appendix  \ref{rem:endogenous_precision} studies optimal ex-ante choice of the precision of the follow-up study. 

The cost $c$ can be a function of $s_2^2$.\footnote{The cost $c$ may capture direct and opportunity cost as long as these are not a function of $\Delta$. Therefore we abstract from benefits forgone while awaiting additional evidence, which would generally depend on \(\Delta\). This approximation is appropriate when the delay is short relative to the horizon over which the policy affects welfare, so that forgone benefits from waiting are small. Examples include anti-poverty programs implemented at scale over a long horizon, or typical FDA decisions for drug approvals.} For expositional convenience, throughout we will refer to $c$ as $c(s_2^2)$ omitting its dependence with $s_2^2$ whenever clear from the context.

\begin{figure}[!ht]
\centering
\begin{tikzpicture}[
    >=latex,
    timeline/.style={thick},
    event/.style={
      rectangle, draw, rounded corners,
      align=left, inner sep=4pt, font=\footnotesize,
      text width=3.35cm
    },
    note/.style={
      align=center, font=\scriptsize,
      text width=3.35cm
    }
]

\draw[timeline,->] (0,0) -- (10.7,0);

\foreach \x/\lab in {0/$t=0$,5/$t=1$,10/$t=2$} {
  \draw[timeline] (\x,0.1) -- (\x,-0.1);
  \node[below=0.22cm] at (\x,0) {\lab};
}

\node[event, above=0.55cm] (design) at (0,0) {%
    \textbf{Data collection}\\
    Researcher observe $X_1 \sim \mathcal{N}(\Delta,s_1^2)$
};

\node[event, above=0.55cm] (report) at (5,0) {%
    \textbf{Decision}\\
    a policy or ask for more evidence
};

\node[event, above=0.55cm] (audience) at (10,0) {%
    \textbf{Loss}\\
    The cost of deferral plus the expected loss using accumulated evidence
};

\draw[->,timeline] (design.east) -- (report.west);
\draw[->,timeline] (report.east) -- (audience.west);

\end{tikzpicture}
\caption{\footnotesize Timeline of the decision process.}
\label{fig:communication_model}
\end{figure}

 \subsection{Objective function}
 
Given the decisions $\delta_1, \delta_2$, the corresponding loss function depends on the welfare effect of adopting the policy $\Delta$ net of cost of future experimentation. This is defined as follows: 
$$
\ell(\delta_1, \delta_2; c, \Delta) = - \Big(\underbrace{\Delta 1\{\delta_1 = 1\}}_{\text{Welfare effect today}} + \underbrace{\Delta1\{\delta_2 = 1\} 1\{\delta_1 = \texttt{NA}\}}_{\text{Welfare effect from future decision}} - \underbrace{c 1\{\delta_1 = \texttt{NA}\}}_{\text{Cost from future experimentation}} \Big). 
$$
Because $\Delta$ is unknown, and 
following \cite{savage1951}, we study optimal decisions from a minimax regret perspective. Specifically, we compare the loss function to the highest achievable welfare corresponding to $\Delta 1\{\Delta \ge 0\}$. Expected regret therefore takes the following form: 
\begin{equation} \label{eqn:reg1}
\small 
\begin{aligned}
\mathrm{Reg}(\delta_1, \delta_2; \Delta,c) & = \underbrace{\mathbb{E}_{\Delta}\Big[\ell(\delta_1(X_1), \delta_2(X_1, X_2);c, \Delta) \Big]}_{\text{Expected researcher's loss}} + \underbrace{\Delta 1\{\Delta \ge 0\}}_{\text{highest achievable welfare}}
\end{aligned} 
\end{equation} 
It follows that under the regret objective, a correct policy decision incurs no loss while an incorrect decision incurs a loss equal to the magnitude of the treatment effect; $c$ measures the cost (in welfare units) of additional data. This is summarized in Table \ref{tab:additive_regret}. 
 \begin{table}[!ht]
\centering
\caption{Realized regret with immediate decisions and continuation}
\label{tab:additive_regret}
\begin{tabular}{lcccc}
\toprule
& \multicolumn{2}{c}{Without waiting}
& \multicolumn{2}{c}{After waiting} \\
\cmidrule(lr){2-3}\cmidrule(lr){4-5}
True welfare effect & Implement & Reject & Implement & Reject \\
\midrule
Beneficial policy: $\Delta>0$ & $0$ & $\Delta$ & $c$ & $c+\Delta$ \\
Harmful policy: $\Delta<0$ & $-\Delta$ & $0$ & $c-\Delta$ & $c$ \\
\bottomrule
\end{tabular}

\vspace{0.15cm}
\begin{minipage}{0.95\linewidth}
\footnotesize
\textit{Notes:} Entries report realized regret for the indicated
decisions. Waiting entails collecting additional evidence before making
a terminal decision at a fixed cost $c$. The benchmark is the welfare of an oracle that
knows $\Delta$ and decides immediately without incurring experimentation
costs. 
\end{minipage}
\end{table}

Following \cite{wald1947sequential}, we define the minimax
regret optimal decisions over all measurable decision rules
$\mathcal{D}$ as
\[
(\delta_1^\star,\delta_2^\star)
\in \operatorname*{arg\,min}_{(\delta_1,\delta_2)\in\mathcal{D}}
\sup_{\Delta:|\Delta|\leq\bar{\Delta}}
\mathrm{Reg}(\delta_1,\delta_2;\Delta,c).
\]
$\delta^\star$ does not require prior information about $\Delta$, which may be unavailable at the moment of the study. The ex-ante evaluation and commitment to $\delta^\star$ (which can be a function of the realized data in each period) follows similar to optimal sequential decisions studied in \cite{wald1947sequential}. Ex-ante commitment captures settings with pre-registration or pre-specified approval standards, common for funding agencies and regulatory agency.  
Examples include prespecified  rules in clinical
trials \citep[Section III.C]{food2019guidance}, or pre-specified approval decisions from funding agencies  
\citep{givewell2020preanalysis}. 
See \cite{stoye2012new} for a discussion about rules with commitment, and Section \ref{sec:Bayes_rule} for its connection to Bayesian solutions.

\begin{rem}[Alternative criteria]
Minimax regret is a common criterion for evaluating statistical decision rules;
see, for example, \cite{manski2021econometrics} for a review. Other objectives
are possible. A seemingly natural alternative is to minimize the worst-case
expected loss $\ell$, instead of the loss corresponding to regret. In our setting, however, this criterion admits a trivial
solution. In particular, nature can make the welfare for rules that approve the policy arbitrary small, making the rule that never approves minimax optimal. Minimax
regret avoids this asymmetry and assigns a cost both to approving a harmful policy and to rejecting a
beneficial one proportional to the welfare effects. Another alternative is a Bayesian
objective that averages $\ell$ with respect to a prior $\Delta \sim \pi$. This criterion is natural when the
researcher can elicit a credible distribution of policy effects (or as in Empirical Bayes can estimate it under exchangeability assumptions across studies). It may be
less attractive when a credible prior is unavailable or when the decision is sensitive to its misspecification.\footnote{Even when multiple studies evaluate similar interventions across different contexts, using their distribution to construct a prior for a new site require an exchangeability assumptions that may fail when studies selected for future implementation are not drawn from the same distribution.}

Our approach complements Bayes procedures by accommodating
decisions made before a sufficiently large body of comparable evidence has accumulated. 
\qed 
\end{rem}

\subsection{Optimal decisions}

We will leverage the following characterization of minimax regret optimal rules.

\begin{lem}[Confidence interval representation and uniqueness of minimax solution] \label{thm:dead_zone}
Let Assumptions \ref{ass:experiment}, \ref{ass:2} hold. 
    There exist minimax-regret optimal decisions $(\delta_1^\star,\delta_2^\star)$, 
$$  
\begin{aligned} 
\delta_{1}^\star(X_1) = \begin{cases} 
1 & \text{ if } \frac{X_1}{s_1} \ge \tau^\star  \\ 
0 & \text{ if } \frac{X_1}{s_1} \le -\tau^\star \\ 
\texttt{NA} & \text{ if } \frac{X_1}{s_1} \in (-\tau^\star,\tau^\star)
\end{cases} 
\end{aligned}, \qquad \delta_2^\star(X_1,X_2) = 1\Big\{ \frac{s_2^2}{s_1^2} X_1 + X_2 \ge 0\Big\}
 $$    
 where for $\Phi$ and $\phi$ denoting the Gaussian CDF and PDF respectively, and $w = \frac{s_2}{s_1}$ 
\begin{equation} \label{eqn:tau_star}
 \small 
 \begin{aligned}
 \tau^\star \in \mathrm{arg} \min_{\tau: \tau \ge 0} \max_{d \ge 0}&
    \Big\{d\underbrace{\Phi\!\left(-\tau-d\right)}_{
        \substack{\text{probability of an incorrect}\\
                  \text{first-period decision}}}
    +d\underbrace{\int_{-\tau}^{\tau}
        \Phi\!\left(-wx-d/w\right)
        \phi\!\left(x-d\right)\,dx}_{
        \substack{\text{probability of continuation and}\\
                  \text{an incorrect second-period decision}}}
\\ & \quad
    +\frac{c}{s_1} \underbrace{\left[
        \Phi\!\left(\tau-d\right)
        -\Phi\!\left(-\tau-d\right)
    \right]}_{\text{probability of continuation}}\Big\}.
\end{aligned}
\end{equation} 
In addition, any minimax-regret optimal rule equals almost surely
$\delta_1^\star$ in the first period and $\delta_2^\star$
on the event of continuation, under every
$\Delta\in[-\bar\Delta,\bar\Delta]$.\footnote{Formally, for
every minimax-regret optimal rule
$\delta'=(\delta_1',\delta_2')$ and every
$\Delta\in[-\bar\Delta,\bar\Delta]$,
$ 
\mathbb P_\Delta\!\left(
\delta_1^\star(X_1)\neq\delta_1'(X_1)
\right) =0, 
\mathbb P_\Delta\!\left(
\delta_1^\star(X_1)=\texttt{NA},\
\delta_2^\star(X_1,X_2)\neq\delta_2'(X_1,X_2)
\right)=0.
$ }
\end{lem}
\begin{proof}The threshold structure follows from \cite{wald1947sequential, wald1950statistical}'s minimax representation through least-favorable priors with an adaptation of \cite{waldwolfowitz1950bayes, arrowblackwellgirshick1949}'s characterization of Bayes-optimal sequential decisions. The explicit threshold $\tau^\star$ in Equation \eqref{eqn:tau_star}, and its independence from $\bar{\Delta}$ follows from properties specific to our problem derived in Appendix \ref{proof:dead_zone}, that includes a complete proof. 
\end{proof}

Lemma \ref{thm:dead_zone} shows that in the original pilot study, the optimal decision takes the form of a standard interval-based decision. For sufficiently large positive (resp. negative) estimated effects, the policy is accepted (resp. rejected), while sufficiently small estimated effects lead to deferral. The confidence interval  for which a policy recommendation is deferred depends on how costly a follow-up experiment is. A particularly costly experiment, such as large-scale experimentations discussed in \cite{muralidharan2017experimentation}, policy recommendation may demand more lenient (i.e., smaller) thresholds $\tau^\star$ because of the cost of additional data collection. Easily replicable experiments may require more stringent (i.e., larger) thresholds.

Our main technical step is to characterize the properties  of $\tau^\star$ below.

\begin{lem} \label{cor:1} Let Assumptions \ref{ass:experiment}, \ref{ass:2} hold.  The following hold: 
\begin{itemize} 
\item[(A)] 
$\tau^\star$ is constant in $\bar{\Delta}$ and is only a function of $(s_2/s_1,c/s_1)$.   Denote this as $\tau^\star(w,\gamma)$, with $w = s_2/s_1, \gamma = c/s_1$. 
\item[(B)] For every fixed $w>0$, $\tau^\star(w,\gamma)$
is unique, continuous and nonincreasing in $\gamma\geq0$.
Moreover, $\tau^\star(w,0)=\infty, \lim_{\gamma\downarrow0}\tau^\star(w,\gamma)=\infty$, and
$\tau^\star(w,\gamma)<\infty$ for every $\gamma>0$.
\end{itemize} 
\end{lem}

\begin{proof} See Appendix \ref{proof:cor:1}.  
\end{proof} 

$\tau^\star$ is only a function of the variance of the follow up study relative to the initial experiment $s_2^2/s_1^2$, and of the cost of deferral per standard deviation unit $c/s_1$ (the worst-case $\Delta$ is in the interior of the parameter space since a larger $\Delta$ makes worst-case welfare larger, but also increases the chance to detect the optimal decision). We will therefore write throughout $\tau^\star(w,\gamma)$ where $w = \frac{s_2}{s_1},\gamma = \frac{c}{s_1}$. The threshold is monotonic in the welfare cost of future data collection, for given precision of the first and follow up study.  Lemma \ref{cor:1} is a core technical step to establish Theorem \ref{thm:a1} below.  It involves a series of auxiliary lemmas in the Appendix.


Appendix Table \ref{tab:optimal_thresholds} tabulates the values of $\tau^\star$ as a function of $(w,\gamma)$.
As $c/s_1 \rightarrow 0$, the threshold $\tau^\star \rightarrow \infty$. Similar behavior for $\tau^\star$ occurs as $w$ decreases,  so that follow-up data collection is more precise. Interestingly, to motivate the standard critical value $\tau^\star = 1.96$, we need either a perfect follow up study and $c= 7\% s_1$ or with an equally precise follow up study (arguably the most credible case), we require a much smaller cost $c = 0.15\% s_1$.

We map these costs to explicit economic values in Section \ref{sec:application}. 

\begin{rem}[Computation of $\tau^\star$] The proof of Lemma~\ref{cor:1} establishes strict convexity of
worst-case regret in $2\Phi(\tau)-1$, ensuring a unique $\tau^\star$ computed by scalar optimization.  \qed 
\end{rem}



\subsection{$A$-value: an economic notion of significance value}

While the exact cost \(c\) may be difficult to elicit in practice,
our main idea is for researchers to report the \(A\)-value alongside
conventional \(p\)-values. For a given follow-up design, the \(A\)-value
is the switching welfare cost at which the minimax-optimal regret rule
changes from collecting additional evidence to making an immediate
policy recommendation.

\begin{defn}[$A$-value] The \(A\)-value is the smallest welfare cost of follow-up data
collection at which the minimax-optimal regret rule prescribes an
immediate policy recommendation (``results are deemed significant''). Specifically,  
$$
A_{\mathrm{value}}(X_1;s_1, w) = s_1 \inf\{\gamma: |X_1|/s_1 \ge \tau^\star(w,\gamma)\},  \qquad w = \frac{s_2}{s_1}
$$
with $\tau^\star$ in Lemma \ref{thm:dead_zone}. 
We consider as the benchmark  $A_{\mathrm{value}}(X_1; s_1, w = 1)$. 
\end{defn}

The $A$-value relates statistical significance to the immediate policy decision based on welfare considerations. When $s_2^2$ is unknown ex-ante, we recommend reporting $A$-values for $w = 1$, i.e., corresponding to an equally precise follow-up study.

The following theorem establishes the main decision-theoretic interpretation of $A$-values.

\begin{theorem} \label{thm:a1} Let Assumptions \ref{ass:experiment}, \ref{ass:2} hold.  Then every minimax-regret optimal rule $\delta_1^\star$ satisfies
$ 
\delta_1^\star(X_1)=\texttt{NA}$ if and only if $c<A_{\mathrm{value}}(X_1;s_1,w)$ 
almost surely. In words: 

\begin{center}
A smaller $A$-value indicates more decisive current evidence, as it is optimal to defer the policy decision if and only if additional data can be collected at a lower cost.
\end{center}
\end{theorem}

\begin{proof}  See Appendix \ref{proof:thm:a1}. 
\end{proof} 

Theorem \ref{thm:a1} establishes that for \textit{every} minimax optimal rule, abstention is optimal if and only if the cost of future data collection is below the $A$-value. A key technical challenge is that the least-favorable distribution changes as the cost changes, making the posterior least-favorable distribution and the corresponding expected welfare gains change with the cost. Therefore, the only if result must establish monotonicity and continuity property of $\tau^\star$ in the cost, which we derived in Lemma  \ref{cor:1}.

To our knowledge our definition of $A$-value (and its properties) are new to the literature.

\begin{prop}[Properties of the $A$-value]
\label{prop:c_value_properties}
Denote  $p(|X_1|/s_1) = 2\left\{1 - \Phi\Big(\frac{|X_1|}{s_1}\Big)\right\}$ the two-sided $p$-value and $  \Gamma_w(t) = \inf\{\gamma \ge 0: t\ge \tau^\star(w,\gamma)\}$.  Let Assumptions \ref{ass:experiment}, \ref{ass:2} hold. 
The following properties hold.

\begin{enumerate}
\item
The $A$-value depends on the initial estimate only through its
two-sided $p$-value:
\begin{equation} \label{eqn:characterization}
A_{\mathrm{value}}(X_1;s_1, w)
=
s_1\Gamma_w\left(
\Phi^{-1}\left(1-\frac{p(|X_1|/s_1)}{2}\right)
\right).
\end{equation}
\item For fixed $(s_1,w)$, the $A$-value is weakly increasing in $p(|X_1|/s_1)$.
\item For fixed $p$-value $p(|X_1|/s_1)$ and $w$, the $A$-value is increasing in the variance $s_1^2$. 
\item
The $A$-value is scale equivariant: for every $a>0$,
$ 
A_{\mathrm{value}}(aX_1;as_1, w)
=
aA_{\mathrm{value}}(X_1;s_1, w).
$ 
\end{enumerate}
\end{prop}

\begin{proof} See Appendix \ref{proof:prop:c_value_properties}.  
\end{proof} 

First, after normalization by $s_1$, the $A$-value depends on the initial estimate $X_1$ only through its $p$-value and on the design of the follow-up experiment through $w=s_2/s_1$. For fixed $(s_1,w)$, the $A$-value is increasing in the $p$-value. Thus, as with the $p$-value, larger values indicate weaker initial evidence and a greater value of collecting additional information.

Second, unlike the $p$-value, the $A$-value also depends on the uncertainty through $s_1$ separately. Holding the $p$-value and the relative precision $s_2/s_1$ of the follow-up experiment fixed, the $A$-value is proportional to $s_1$, so that more precise initial experiments lead to smaller $A$-values (and therefore more conclusive evidence). Intuitively, fixing the $p$-value, a smaller $s_1$ implies a smaller estimated effect. The welfare consequences of an incorrect immediate decision, and hence the value of collecting additional evidence, are therefore smaller. 
Finally, the $A$-value changes proportionally with the units in which $X_1$ and $s_1$ are measured. Comparisons across studies therefore require estimated effects to be expressed in common welfare units. Below is a standardized version.

\begin{cor}[Scale-invariant $A$-value]
\label{rem:scale_invariant_a_value} Let Assumptions \ref{ass:experiment}, \ref{ass:2} hold. Consider Example \ref{exmp:asymptotic} where
$n_1$ denotes the total sample size of the initial experiment and $\pi$ the fraction of treated units, and $\sigma^2$ the homoskedastic outcome variance.\footnote{With heteroskedastic variances, we recommend replacing $\sigma^2 = (1-\pi)\sigma_1^2 + \sigma_0^2 \pi$ with $\sigma_1^2, \sigma_0^2$ denoting the variance of treated and control units respectively.} We define the standardized $A$-value 
\[
A_{\mathrm{value}}^\star\left(\frac{|X_1|}{s_1};n_1, \pi\right) \equiv 
\frac{1}{\sqrt{n_1 \pi(1-\pi)}}
\Gamma_1\left(\frac{|X_1|}{s_1}\right).
\] 
Denote $c' \equiv c/\sigma$ the cost per standard deviation $\sigma$. 
The following holds: 
\begin{itemize} 
\item It is optimal to defer a policy decision for an equally-precise follow up study if and only if $c' < A_{\mathrm{value}}^\star\left(\frac{|X_1|}{s_1};n_1, \pi\right)$ almost surely for any admissible $\Delta$; 
\item $A_{\mathrm{value}}^\star\left(\frac{|X_1|}{s_1};n_1, \pi\right)$ is increasing in the $p$-value for fixed $n_1,\pi$, decreasing in $n_1$ for fixed $p$-value and $\pi$ and scale invariant. 
\end{itemize}
\end{cor}

$A^\star$ defines a natural benchmark. It retains the same interpretation of the $A$-value, with  experimentation costs measured per standard deviation $\sigma$. This benchmark highlights an important welfare interpretation.
Holding the $p$-value fixed, the standardized $A$-value decreases
with sample size. The reason is that effects that remain difficult
to distinguish from zero are smaller in welfare units. The welfare consequences of an incorrect
recommendation, and hence the value of collecting further evidence,
are therefore smaller.



\subsection{Examples and $A$-value report}
\label{sec:examples}

\paragraph{Funding anti-poverty programs} First, we map our framework to funding decisions at GiveWell, a leading philanthropic funding agency.
GiveWell compares programs by their welfare return per dollar relative to
a benchmark. Let
$ 
b
$ denote the welfare per dollar generated by the program 
and let $\widehat b_1$ denote its estimate. GiveWell constructs this
estimate using empirical evidence and its own valuation assumptions, and
sets a funding per dollar opportunity cost $\kappa>0$ explicitly calculated by the agency.\footnote{See
\href{https://www.givewell.org/how-we-work/our-criteria/cost-effectiveness/cost-effectiveness-models/december-2024-version}{GiveWell's explanation of its benchmark and funding threshold}.} Making a decision solely based on whether $\widehat{b}_1 \ge \kappa$ may fail to capture the role of uncertainty on downstream welfare implications and value of further experimentation. The A-value accounts for both without requiring prior information about the program's effectiveness, unavailable without comparable evidence.

Let $K>0$ denote the implementation budget considered for the program, and let $C$ denote GiveWell's
contribution to a follow-up study. Funding the program then gains $K(b-\kappa)$, while spending $C$ on research forgoes in opportunity costs $\kappa C$.
Suppose an initial study with $n_1$ participants yields a Gaussian estimate
$\widehat b_1$ with mean $b$ and standard error $s_1$. Dividing welfare by $K$ maps
this decision into our framework:
\[
\small 
\begin{aligned} 
\Delta=b-\kappa,\qquad
X_1=\widehat b_1-\kappa,\qquad
c=\frac{\kappa C}{K}.
\end{aligned} 
\]
Economic inputs for $K$, $\kappa$, and $C$ can be constructed from
GiveWell's valuation models.\footnote{For an  example, see 
GiveWell's \href{https://docs.google.com/spreadsheets/d/1q9WkdBjh6ZhGieXMFgfvD9yg2ALmzifCz9IO2S1OGhs/edit?gid=1171143230}{2024 eyeglasses research valuation model}. } 

\paragraph{Medical study evaluation}
Consider a health authority deciding whether to recommend a treatment targeted to 
$N$ patients. Let $b$ denote the
average effect on the probability of survival at a given horizon, $V>0$ the value of a statistical life, and $\kappa$ the
incremental treatment cost per patient. Adopting the treatment generates
net monetary welfare $N(Vb-\kappa)$. Let $C$ denote the cost of a follow-up
trial.
Suppose an initial study with $n_1$ participants yields an approximately
Gaussian estimate $\widehat b_1$ with mean $b$ and standard error $s_1$.
Mapping this into our framework
\[
\small
\begin{aligned}
\Delta=b-\frac{\kappa}{V},\qquad
X_1=\widehat b_1-\frac{\kappa}{V},\qquad
c=\frac{C}{NV}.
\end{aligned}
\]
Here, $V$ is obtained from \cite{usdot2026vsl}, while $N$ is the size of the target population, specified ex ante. Appendix \ref{sec:application_ist} presents an empirical illustration. 

\paragraph{$A$-value report}
In either example, consider an independent follow-up study estimating the
same welfare effect, with $n_2$ participants and the treatment probability $\pi$ as in Example \ref{exmp:asymptotic}, so
$w=\sqrt{n_1/n_2}$ where $n_1$ is the number of participants in the first study. A natural benchmark is to consider $n_1 = n_2$ although other choices are possible. 
By Theorem \ref{thm:a1}, the minimax-optimal rule defers the decision and recommends the follow-up study if and only if $c<A_{\mathrm{value}}(X_1;s_1,w)$. 
Researchers can
report the $A$-value and, when the required calibration inputs are available,
its monetary break-even cost, see  
Table~\ref{tab:a_value_report}.

\subsection{Connection to Bayesian  solutions}
\label{sec:Bayes_rule}

We pause here and discuss the connection of the minimax regret solution to Bayes procedures and how it is possible to extend our notion of $A$-value by also incorporating prior information.

\paragraph{Bayesian formulation as a game against nature} For each experimentation cost, the proof of
Lemma~\ref{thm:dead_zone}, using
\citet{wald1950statistical}, shows that the canonical minimax rule
is Bayes optimal under every ex ante least-favorable prior (Appendix Equation \eqref{eqn:Bayes_minimax}).
These priors are determined by the minimax problem and do
not require to specify a distribution of $\Delta$.\footnote{The least favorable priors include those  symmetric around zero and non-degenerate, see Appendix \ref{app:least_favorable}.} In the associated game against Nature, a least-favorable
prior is Nature's optimal ex ante mixed strategy. When
Nature commits to such a prior, $\delta^\star$ remains optimal at every reached history
under Bayesian updating of that prior and also minimizes worst-case posterior expected regret over
the set obtained by updating all ex ante least-favorable
priors.\footnote{This argument connects to discussion around dynamic consistency. Here we consider rules with pre-commitment, which may not (need to) satisfy dynamic consistency, see \cite{stoye2012new}. Because
it is also possible to interpret the minimax regret solution as a Bayes rule optimal under any least-favorable prior defined by the game against nature, Assumption \ref{ass:experiment} implies sequential optimality against every posterior defined by the class of ex-ante least-favorable priors. For connections of minimax and Bayes rules, see also  \cite{wald1947sequential}, \cite{manski2021econometrics}.} 
The \(A\)-value admits a posterior interpretation under the least favorable prior as the cost for which the posterior expected welfare gain from the
follow-up equal that cost.

\paragraph{An extension when incorporating a prior} In some applications, researchers may want to weigh Bayes risk under a \textit{given} prior $\pi$ against the minimax regret solution. We study this question in Appendix \ref{sec:prior_information}. For a given prior $\pi$ centered around zero and symmetric, it follows that the Bayes optimal rule takes the same form as $\delta^\star$ with $\tau^\star$ replaced by a prior dependent threshold $\tau_\pi(c; s_1, s_2)$. To weigh this decision with the minimax regret optimal decision, choose a 
weight
$\lambda\in[0,1]$ and define
\[
\small 
\begin{aligned}
A_{\lambda,\pi}
&=\inf\left\{c\geq0:
\frac{|X_1|}{s_1}\geq\tau_\lambda(c)\right\}, \qquad \tau_\lambda(c)
&=\Phi^{-1}\Big(
(1-\lambda)\Phi(\underbrace{\tau^\star(w,c/s_1)}_{\text{minimax optimal}})
+\lambda\Phi(\underbrace{\tau_\pi(c; s_1, s_2)}_{\text{Bayes optimal}})
\Big).
\end{aligned}
\]
Thus, $\lambda=0$ recovers the minimax rule and $\lambda=1$ recovers
the Bayes rule. 

Theorem \ref{thm:prior_a_value} shows that the rule choosing 
$$
\delta_1^{\lambda}(X_1) = \texttt{NA} \text{ if and only if } c < A_{\lambda,\pi}(X_1;s_1,w)
$$
has two main properties. 
First, it achieves
at least a fraction $\lambda$ of the Bayes rule's improvement in
prior-average regret. Second, it incurs at most the same fraction
of its excess worst-case regret. The weight $\lambda$ therefore allows users to weigh prior information in the construction of the $A$-value against robustness, while retaining its intuitive characterization.

\section{Using $A$-values to allocate (scarce) resources} \label{sec:multiple_comparisons}

The \(A\)-value is particularly useful when researchers face several policy questions and must decide which follow-up experiments to conduct subject to a budget constraint. 
Here, we will be focusing on the leading case where $w = 1$, i.e., follow up data collection corresponds to an experiment with equal precision to the current study. Under this formulation the cost of an equally precise future experiment can be calibrated using the current experimentation cost and provides a natural benchmark interpretation. For $w \neq 1$, see Remark \ref{rem:w}. 

\begin{ass}[Multiple policy comparisons] \label{ass:multiple} Consider $J$ policies. For each policy \(j\), suppose that, by paying a cost $c_j > 0$ each period, researchers observe
$$
X_{1j} \sim\mathcal N(\Delta_j,s_{j}^2),
\qquad
X_{2j} \sim\mathcal N(\Delta_j,s_{j}^2), \qquad s_j \in (0,\infty)
$$
where the $X_{tj}$ are mutually independent across policies and periods. We standardize $\max_j c_j = 1$ and assume $\min_j c_j$ is bounded away from zero. Let
$
\boldsymbol\Delta
=
(\Delta_1,\ldots,\Delta_J)
\in
\prod_{j=1}^J[-\bar\Delta_j,\bar\Delta_j]$, with the upper bound satisfying $\bar{\Delta}_j < \infty, \bar{\Delta}_j/s_j \ge 0.76
$ for all $j$.  

Write
$
A_j
=
\frac{A_{\mathrm{value}}(X_{1j};s_{j}, w = 1)}{c_j},  
$
the cost-adjusted $A$-value. 
\end{ass} 

Independence is natural when the policy-specific signals are constructed from separate samples; this is common when studying interventions in different experimental sites or populations. Effects between different policies can be arbitrary heterogeneous, and each policy may have specific experimentation costs $c_j$. Importantly, costs only need to reflect relative costs across studies, and can be expressed in monetary units even when the $A$-value is not.  

\begin{defn}[Strategy with $q$-experimentation constraints] A strategy
\(\boldsymbol\delta=(\boldsymbol\delta_1,\boldsymbol\delta_2)\)
consists of a measurable first-period rule
$
\boldsymbol\delta_1:
\mathbb R^J
\rightarrow
\{0,1,\texttt{NA}\}^J,
$ where \(\delta_{1j}(\mathbf X_1)=1\) denotes immediate implementation,
\(\delta_{1j}(\mathbf X_1)=0\) denotes immediate rejection, and
\(\delta_{1j}(\mathbf X_1)=\texttt{NA}\) denotes collecting follow-up evidence for policy \(j\). A budget constraint imposes that 
$
\sum_{j=1}^J
c_j 1\{\delta_{1j}(\mathbf X_1)=\texttt{NA}\}
\le  
q
$
for given $0 \le q \le \sum_j c_j$, almost surely.  

After the first-period decisions, researchers observe \(X_{2j}\) only for policies satisfying
\(\delta_{1j}(\mathbf X_1)=\texttt{NA}\). For each such policy,
\(\delta_{2j}\in\{0,1\}\) denotes the terminal implementation decision and may depend on \(\mathbf X_1\) and all the follow-up evidence collected under the first-period rule. Let \(\mathcal D_q\) denote the class of all measurable strategies satisfying these constraints. \qed 
\end{defn} 

The budget constraint captures the allocation of a predetermined research budget across follow-up studies. Fixed start-up and coordination costs make the extensive margin, i.e., \textit{which} experiments to conduct, particularly important \citep{list2024optimally}. In economic experiments, these costs include establishing research teams and survey infrastructure, as well as training research teams. In clinical trials, they include establishing collaborations with hospitals, preparing procedures, and approvals \citep[e.g.][]{nevens2019budgeting}.

We evaluate a strategy \(\boldsymbol\delta\in\mathcal D_q\) using average regret. We write the regret as 
$$
\small 
\begin{aligned}
\overline{\mathrm{Reg}}(\boldsymbol\delta;\boldsymbol\Delta)
=
\frac{1}{J} \sum_j \mathrm{Reg}(\delta_{1j}, \delta_{2j}; \Delta_j,c = 0), \qquad \text{ such that } \sum_j c_j 1\{\delta_{1j}(\mathbf{X}_1) = \texttt{NA}\} \le q
\end{aligned}
$$
where we set the cost in the objective function equal to zero and re-interpret the cost of experimentation as a budget constraint. We propose the following intuitive strategy. 

\begin{defn}[Top $A$-value selection] \label{defn:top_A}
After observing \(\mathbf X_1\), order policies such that their cost-adjusted $A$-values are (breaking ties arbitrary) 
$ 
A_{(1)}
\geq \cdots \geq
A_{(J)},
$ 
and let
\[
\small 
\begin{aligned} 
\mathcal J_q(\mathbf X_1)
=
\left\{
k:
k\in\{1,\ldots,J\},\
\sum_{\ell=1}^{k}c_{\ell}\leq q
\right\}.
\end{aligned} 
\]
Thus, policies are selected in decreasing order of \(A_j\),
stopping at the first policy whose inclusion would exceed the budget. The top \(A\)-value procedure
\(\boldsymbol\delta^{A,q}\in\mathcal D_q\) is defined by
\[
\small 
\begin{aligned} 
\delta_{1j}^{A,q}(\mathbf X_1)
=
\begin{cases}
\texttt{NA},
& j\in\mathcal J_q(\mathbf X_1),\\
1,
& j\notin\mathcal J_q(\mathbf X_1)
\text{ and }X_{1j}\geq0,\\
0,
& j\notin\mathcal J_q(\mathbf X_1)
\text{ and }X_{1j}<0,
\end{cases}
\end{aligned} 
\]
and, for every \(j\in\mathcal J_q(\mathbf X_1)\),
$ 
\delta_{2j}^{A,q}(X_{1j},X_{2j})
=
1\left\{X_{1j}+X_{2j}\geq0\right\}.
$  \qed 
\end{defn}

Thus, the procedure collects follow-up evidence for the policies with largest $A$-value adjusted by the cost of the corresponding experiment, makes an immediate decision based on the sign of \(X_{1j}\) for every other policy (which by construction have a lower $A_j$), and uses the combined first- and second-period evidence for the selected policies.  This procedure provides an intuitive approach to making recommendation with finite sample guarantees. 
Whenever $c_j = 1$ for all $j$, the top $A$-value selects the $q$ studies with the largest $A$-values.  

\begin{rem}[Heterogeneous follow-up] \label{rem:w} When $X_{2j} \sim \mathcal{N}(\Delta_j, \tilde{s}_j^{2})$ with $\tilde{s}_j^{2} \neq s_j^2$, we can show our results continue to hold, with $A_j$ replaced by $A_{\mathrm{value}}(X_{1j};s_j, \tilde{s}_j/s_j)/c_j$, provided $w_j = \tilde{s}_j^2/s_j^2$ is bounded away from zero and infinity, and $\delta_{2j}^{A,q} = 1\{w_j X_{1,j} + X_{2,j} \ge 0\}$. \qed  
\end{rem} 

\begin{rem}[\(A\)-values and \(p\)-value rankings]
Suppose that all policies have the same initial and follow-up precisions, so that \(s_{j}=s\) for every \(j\). The \(A\)-value is then increasing in the two-sided \(p\)-value. Consequently, with \textit{both} homogeneous costs and variances, the top-\(q\) \(A\)-values procedure collects additional evidence for the \(q\) policies with the largest \(p\)-values. Whenever either or both differ across policies, the top $A$-ranking may differ from the ranking under the $p$-value as the $A$-value accounts in addition for the precision of the estimates. \qed
\end{rem}

\begin{rem}[Weighted objective]
In some cases, researchers may assign a fixed welfare weight
$\omega_j>0$ to each policy, for example to reflect the size of the
population affected. Write the corresponding regret as
$ 
\frac{1}{J} \sum_j \omega_j \mathrm{Reg}(\delta_{1j}, \delta_{2j}; \Delta_j,c = 0),  
$ 
under the same $q$-budget constraint $\sum_j c_j 1\{\delta_{1j}(\mathbf{X}_1) = \texttt{NA}\} \le q$. 
This problem is equivalent to the unweighted framework after
replacing
$(X_{tj},\Delta_j,s_j,\bar\Delta_j)$ by
$(\omega_jX_{tj},\omega_j\Delta_j,\omega_js_j,
\omega_j\bar\Delta_j)$ for each $j$ and $t\in\{1,2\}$.
The costs $c_j$ and the budget $q$ retain their original values
in the constraint defining $\mathcal D_q$.
One can then apply our same procedure (and inherit the same guarantee below) replacing $A_j$ with its weighted version $\omega_j A_j$
by the scale equivariance in Proposition
\ref{prop:c_value_properties}.  \qed
\end{rem}

\subsection{Optimality guarantees}

Our next goal is to study how the top $A$-value solution compares to two target objectives
 $$
R_{q'}^\star = \min_{\boldsymbol{\delta} \in \mathcal{D}_{q = q'}} \sup_{\boldsymbol{\Delta} \in \prod_{j=1}^J[-\bar\Delta_j,\bar\Delta_j]} \overline{\mathrm{Reg}}(\boldsymbol\delta;\boldsymbol\Delta), \qquad q' \in \{q,\infty\}. 
 $$
 Here $R_q^\star$ denotes the smallest objective under the feasible budget constraint. Comparison with $R_q^\star$ establishes whether the proposed procedure is (approximately) optimal.

\begin{theorem}[Approximate optimality as a function of $J$ and study's precision] \label{thm:optimality1}
Suppose that Assumption \ref{ass:multiple} holds and let $\bar{s} = \max_j s_j$. 
There exists a universal constant $C<\infty$ such that,
for every $J$ and every $0\leq q<\sum_{j=1}^J c_j$,
\[
\small 
\begin{aligned} 
\sup_{\boldsymbol\Delta}
\overline{\mathrm{Reg}}
\bigl(\boldsymbol\delta^{A,q};\boldsymbol\Delta\bigr)
-
R_q^\star
\leq
C\frac{\bar{s}}{\underline{c}} \min\left\{\frac{1}{\sqrt J}, \frac{q}{J}\right\}.
\end{aligned} 
\]
In particular, suppose that $J\rightarrow\infty$, with $q$ potentially growing with $J$ at arbitrary rate, and
$ 
\frac{\max_j s_{j}}
     {J^{-1}\sum_{j=1}^J s_{j}}
=
\mathcal O(1).
$ 
Then
$
\frac{
\sup_{\boldsymbol\Delta}
\overline{\mathrm{Reg}}
\bigl(\boldsymbol\delta^{A,q};\boldsymbol\Delta\bigr)
}{
R_q^\star
}
=
1+\mathcal O\!\left(J^{-1/2}\right).
$ 
\end{theorem}

\begin{proof} See Appendix \ref{app:proof:optimality1}. 
\end{proof}

Theorem \ref{thm:optimality1} shows that the proposed procedure is approximately minimax regret optimal up-to an error of order $\bar{s}/\sqrt{J}$ even when 
$q/J$ is potentially constant, i.e., the experimentation budget
is proportional to the number of policies. The regime where $q/J = O(1)$ is the most relevant and harder to study regime that captures settings where funding agencies may decide to fund follow-up data collection of a fraction of existing studies.\footnote{A second but less interesting regime is where $q$ is negligible relative to $J$, i.e. $q/J=o(1)$. In this case even the
procedure that conducts no follow-up experiments has relative
worst-case regret $1+\mathcal O(q/J)$ compared with $R_q^\star$.
Thus, approximate optimality in that regime is less interesting as it cannot distinguish
between alternative experimentation procedures and remains uninformative, motivating our focus on $q/J \rightarrow \rho \in (0,1)$. For example, this captures settings 
where the number of starting grants that would fund the first round experiment are roughly of the same order of consolidating grants \citep[e.g.][]{erc_work_programme_2026}. 
} Importantly, the result does not require assumptions on how each study effect relates to each other, as in settings where effects may correspond to different countries or even different program interventions.


Minimax solutions are also interesting in small $J$ regimes. We therefore establish finite-sample upper  bounds in settings with small $J$ and compare regret against the unconstrained optimum $R_{\infty}^\star$, allowing the oracle to use more data at \textit{no cost}. Comparing against  $R_{\infty}^\star$ provides a measure of the welfare cost of the constraint $q$.

\begin{theorem}[Comparison with unrestricted rule with small $J$]
\label{prop:top_q_c_value}
Suppose that Assumption \ref{ass:multiple} holds. Let
$A_{(1)}\geq\cdots\geq A_{(J)}$ denote the ordered
cost-adjusted $A$-values. Then, for every
$q\in\{1,\ldots,J-1\}$ satisfying $q\leq\sum_{j=1}^Jc_j$
and every
$\boldsymbol\Delta\in\prod_{j=1}^J[-\bar\Delta_j,\bar\Delta_j]$,
\[
\overline{\mathrm{Reg}}
\bigl(\boldsymbol\delta^{A,q};\boldsymbol\Delta\bigr)
-R_\infty^\star
\leq
\frac{J-q}{J}
\mathbb E_{\boldsymbol\Delta}[A_{(q)}].
\]
Moreover,
$ 
\sup_{\boldsymbol\Delta\in
\prod_{j=1}^J[-\bar\Delta_j,\bar\Delta_j]}
\mathbb E_{\boldsymbol\Delta}[A_{(q)}]
=
\mathbb E_{\boldsymbol\Delta=\mathbf0}[A_{(q)}].
$ 
\end{theorem}
\begin{proof}
See Appendix \ref{app:proof:top_q_c_value}. 
\end{proof}

Theorem \ref{prop:top_q_c_value} provides an approximation guarantee for the top \(A\)-values procedure relative to the infeasible oracle $R_\infty^\star$ that holds for finite $J$.

The bound depends on the expected $q$-th largest cost-adjusted $A$-value. Under unit costs, exactly $q$ policies receive follow-up evidence, so $A_{(q)}$ is the smallest funded score and represents the experimentation margin.
This result further highlights the interpretation of the \(A\)-value as a measure of the value of additional evidence, where policies with larger \(A\)-values are those for which an immediate decision is potentially most costly. The $A$-value at the margin $q$ determines the welfare cost from imposing a research budget constraint.

While the bound holds for every treatment-effect vector \(\boldsymbol\Delta\), because the worst-case bound $\mathbb E_{\boldsymbol{\Delta} = \boldsymbol 0}
\left[
A_{(q)}
\right]$ can be computed as a function of the precision of each study in analytic form, researchers may ex-ante commit to a given $q$ to guarantee control of the excess risk.

We conclude this discussion by studying how regret decreases with the per-study sample size as $s_j \rightarrow 0$ uniformly for all $j$ in a regime where $J$ is fixed. This is the least interesting regime here as many rules achieve minimax optimal rates in this regime. Nevertheless, it is  reassuring to see that our proposed top $A$ achieves a fast rate also under such regime. 

\begin{cor}[Minimax rate in $n$ with fixed $J$]
\label{prop:minimax}
Maintain Assumption \ref{ass:multiple}. Fix $J$, positive costs
$c_1,\ldots,c_J$, a budget $0\leq q<\sum_{j=1}^Jc_j$, and
constants $a_1,\ldots,a_J>0$. Suppose that
$ 
s_j=\frac{a_j}{\sqrt n},
 n\rightarrow\infty,
$ 
with $\bar\Delta_j/s_j\geq0.76$ for every $j$ along the sequence.
There exist constants $C,C'\in(0,\infty)$, depending only on
$J$, $q$, and $(a_j,c_j)_{j=1}^J$, such that, for every $n$
along the sequence and every $\boldsymbol\delta\in\mathcal D_q$,
$ 
\sup_{\boldsymbol\Delta}
\overline{\mathrm{Reg}}(\boldsymbol\delta;\boldsymbol\Delta)
-R_\infty^\star
\geq\frac{C}{\sqrt n}.
$ 
Moreover,
$ 
\sup_{\boldsymbol\Delta}
\mathrm{Reg}(\boldsymbol\delta^{A,q};\boldsymbol\Delta)
-R_\infty^\star
\leq\frac{C'}{\sqrt n}.
$ 
\end{cor}

\begin{proof}
See Appendix \ref{proof:prop:minimax}.
\end{proof}


\begin{rem}[Comparison with alternative procedures for large $J$] \label{rem:comparisons}
Appendix \ref{sec:comparisons} shows that experimentation rules based on $p$-value or standard error only ranking, Bayes rules under a
common prior, and the unrestricted conditional likelihood
procedure need not attain asymptotic minimax optimality as
$J\to\infty$.
\qed
\end{rem}

\section{Extensions}

Appendix \ref{sec:extensions} presents six extensions. First, we allow researchers to collect additional data sequentially over multiple rounds. The optimal rule retains a threshold representation, with thresholds depending on the round and the costs and precision of future data collection.

Second, we consider asymmetric regret in the spirit of \cite{tetenov2012}, allowing incorrect implementation and incorrect retention of the status quo to entail different losses. The threshold representation continues to hold, with thresholds reflecting the asymmetry.

Third, we study the optimal follow-up sample size when researchers or funding agencies must commit to the experimentation design ex ante. Conditional on the chosen sample size, our characterization applies with the corresponding cost and precision.

Fourth, we formalize the intuition in Remark \ref{rem:comparisons}, illustrating how alternative intuitive rules are not maximin optimal even for large $J$.  

Fifth, we consider the case where $X_1$ is drawn from an exponential-family distribution. 

Sixth, we present a version of the $A$-value that guarantees control of Bayes risk when researchers want to use a (potentially misspecified) priorr and minimax regret. 

\section{Empirical application} \label{sec:application}

In this section, we use data from \cite{crosta2024unconditional} to analyze unconditional cash transfer (UCT) programs in low- and middle-income countries. \cite{crosta2024unconditional} cover 115 studies of 72 randomized evaluations in 34 countries. We focus on monthly household consumption, which is available and has been harmonized across 45 programs. Treatment effects and standard errors are expressed as monthly household consumption per \$100 of total transfer in 2010 USD PPP. For each study, we use a single point estimate selected as follows. First, we keep estimates which use the overall sample and omit those based on smaller subgroups. Second, within each program, we retain estimates from the longest available follow-up horizon. Third, we select the arm with the largest total transfer delivered.\footnote{For lump-sum transfers, we use the reported total transfer. For stream transfers, we use cumulative transfers delivered
through the selected follow-up, accounting for the payment schedule.} The longest-follow-up horizon focuses  on the most persistent effects, while the largest-transfer restriction focuses on the highest-intensity arm. This procedure yields 42 program-level estimates out of the original 45 studies for which household consumption is available.\footnote{Three programs are excluded because they report consumption effects only for particular subgroups.}

For each of the 42 programs, we construct an estimate of average monthly consumption, expressed per \$100 of total transfers.  
Define \(\widehat b_j\) and \(s_{j}\) the resulting monthly consumption estimate and standard error at the selected follow-up, expressed per \$100 of total transfers.\footnote{Lump-sum programs measure monthly household consumption effects per \$100 of total transfers, whereas those for stream programs measure monthly effects per \$100 of the monthly tranche. Because the monthly tranche equals cumulative transfers divided by the months since the first transfer to follow-up $m_j$, we divide both the reported estimate and its standard error by \(m_j\) for stream programs and leave them unchanged for lump-sum programs.  
}

We consider two scenarios. In the first scenario we set 
$
X_{1j}
=\widehat b_j,
$
so that in the baseline specification we assume that treatment effects measure directly the researcher's  welfare criterion, abstracting from the opportunity and administrative costs of transfer expenditure. As a second specification, we target monthly consumption at follow-up net of amortized administrative costs. Following \cite{crosta2024unconditional}, we assume that administrative costs equal 24 percent of the transfer amount. We value each dollar of administrative expenditure as one dollar of forgone consumption and spread these costs evenly over  a common $H = 12$ months window of intervention. The resulting estimate is
$
X_{1j}^{\mathrm{admin}}
=\widehat b_j-\frac{24}{H}.
$



Our framework applies to two complementary setups: 
\begin{itemize} 
\item  \textit{Individual researcher decision}: We take the perspective of a researcher reporting the results of an individual policy evaluation.  After observing the evidence from a given program, the research team can recommend scaling-up the program, reject it, or defer the decision and commission a follow-up experiment with the same precision as the initial experiment (and different precisions in Appendix Figure \ref{fig:different_precision}).

The research team may lack a collection of comparable studies from which to estimate a distribution of policy effects when evaluating the program, as when studies in different countries can be arbitrarily heterogeneous. 
\item \textit{Funding agency:} A funding agency observes a collection of multiple heterogeneous studies and must decide which ones to prioritize for follow-up data collection.
\end{itemize}

\subsection{Analysis}

Figure \ref{fig:distribution} and
Table \ref{tab:a-value-distribution} summarize the distribution
of $A$-values across the 42 programs. The distribution is strongly
right-skewed, with many values close to zero and a few substantially
larger values. Selecting the ten programs
with the largest $A$-values yields a cutoff of approximately
0.070 without administrative costs and 0.057 with them.
Administrative costs therefore change both the selection cutoff
and the programs selected. They can bring estimated welfare effects
closer to zero, 
or move them farther from zero, reducing that value.

\begin{table}[!ht]
\centering
\setstretch{1}
\small
\setlength{\tabcolsep}{3pt}
\caption{Distribution of A-values in consumption units and monetary equivalent}
\label{tab:a-value-distribution}
\begin{tabular}{@{}lrr@{\hspace{0.45cm}}lrr@{}}
\toprule
\multicolumn{3}{c}{No administrative costs}
& \multicolumn{3}{c}{Administrative costs (24\%)} \\
\cmidrule(lr){1-3}\cmidrule(lr){4-6}
Statistic & A-value & \shortstack{Monetary\\equivalent (\$)}
& Statistic & A-value & \shortstack{Monetary\\equivalent (\$)} \\
\midrule
Studies & 42 & --- & Studies & 42 & --- \\
Mean & 0.258 & 309,487 & Mean & 0.230 & 276,364 \\
Standard deviation & 0.977 & 1,172,850
& Standard deviation & 0.828 & 993,472 \\
25th percentile & $1.25\times 10^{-5}$ & 15
& 25th percentile & $6.76\times 10^{-4}$ & 811 \\
Median & 0.003 & 3,468 & Median & 0.013 & 15,669 \\
75th percentile & 0.059 & 70,644
& 75th percentile & 0.044 & 52,637 \\
\bottomrule
\end{tabular}
\par\vspace{0.6em}
\begin{minipage}{\linewidth}
\footnotesize\raggedright
Each panel summarizes all 42 program-level A-values at $w=1$,
in monthly consumption units per \$100 of total transfers
(2010 USD PPP). The administrative-cost specification subtracts
$24/12=2$ from each program's gross consumption estimate before
computing A-values. Monetary equivalents assume 10,000 recipient
households, \$1,000 transferred per household, and 12 months of
average monthly welfare gains. They equal $1{,}200{,}000$ times
the corresponding unrounded consumption-unit statistics, rounded
to the nearest dollar; study counts have no monetary equivalent.
\end{minipage}
\end{table}

\begin{figure}[!ht]
    \centering
    \includegraphics[scale = 0.55]{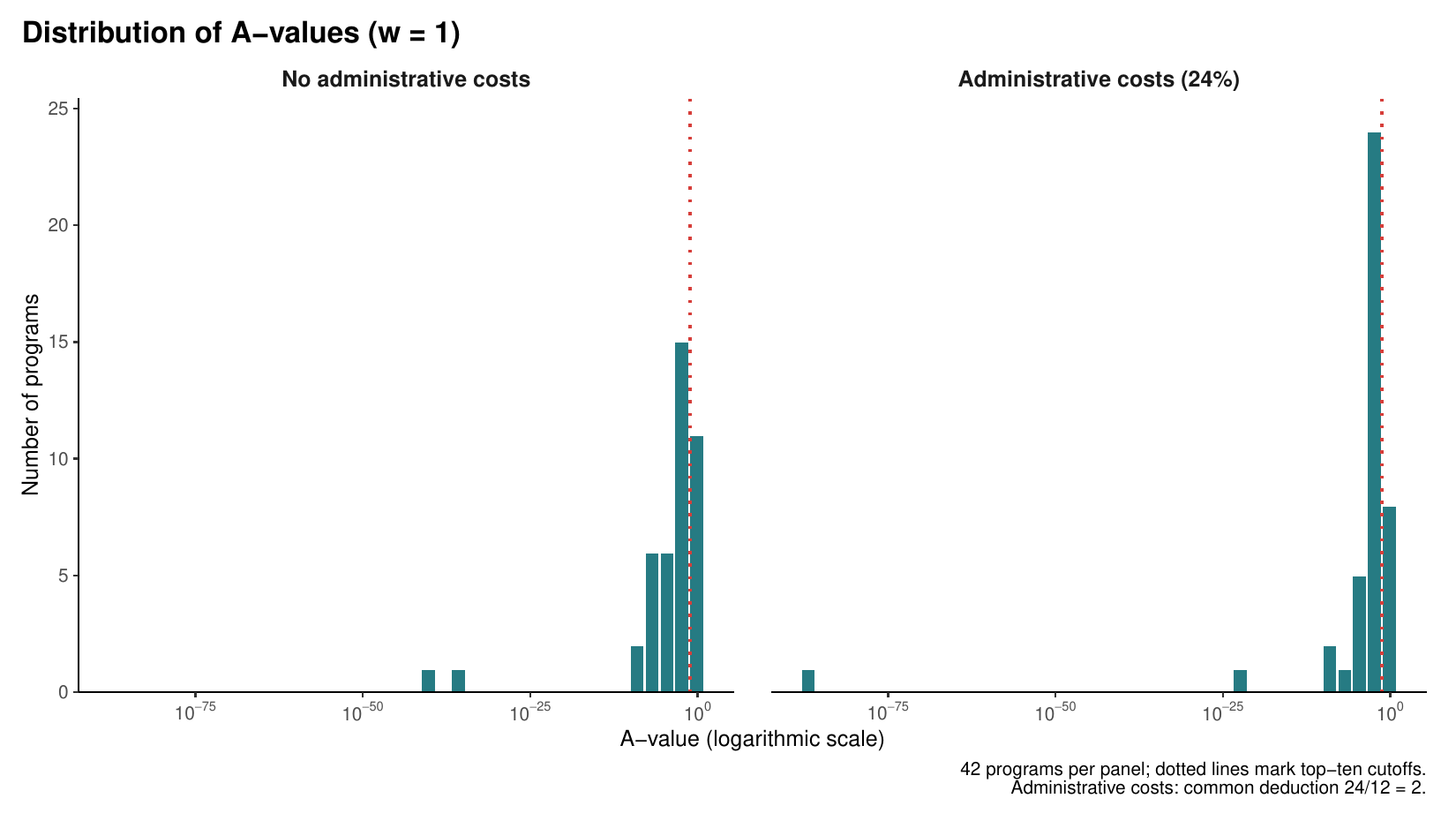}
    \caption{\footnotesize Histogram in logarithmic scale of the distribution of $A$-values across experiments.}
    \label{fig:distribution}
\end{figure}

\paragraph{Monetary standardization and individual researcher's action}
In practice, we recommend reporting each $A$-value alongside a simple monetary standardization. This allows individual researchers to assess whether the cost of further data collection falls below the implied break-even threshold. (Under a common monetary calibration across programs, this conversion does not affect the selection of a fixed number of follow-up experiments.)

 Proposition \ref{prop:monetary} assumes the average monthly consumption
effect at follow-up equals average monthly
consumption effect over \(12\) months, abstracting from dynamics and discounting.

\begin{prop}[Monetary calibration] \label{prop:monetary}
Maintain the conditions of Theorem \ref{thm:a1}, and let
$A=A_{\mathrm{value}}(X_1;s_1, w= 1)$ denote the $A$-value
in monthly household welfare units per \$100 of total transfers.
Suppose the evaluated intervention could be implemented at scale for $N$
households, each receiving the evaluated total transfer of $T$
dollars, and that its effect represents average monthly welfare
gains for these households over $H$ months.
Assume that the total welfare (opportunity) cost of conducting an independent, equally precise follow-up experiment
equals its monetary cost $C_{\mathrm{exp}}$, with each dollar
of expenditure valued as one dollar of forgone consumption.
Then, almost surely, the minimax-regret rule recommends further
experimentation if and only if
$ 
C_{\mathrm{exp}}<\frac{NTH}{100}A.
$ 
\end{prop}

The proof follows directly from Theorem \ref{thm:a1} and scale-equivariance of the $A$-value.\footnote{Multiplying $X_1$ and $s_1$ by $NTH/100$ expresses
the experiment in total consumption-equivalent welfare units.
Scale equivariance multiplies the $A$-value by the same factor.
Theorem \ref{thm:a1} then compares this monetary threshold
directly with $C_{\mathrm{exp}}$.}



For illustrative purposes here, consider as the leading example the case where the total budget of the at-scale intervention $N T = 10,000,000 \$$ which matches large-scale transfer interventions \citep[e.g.][]{egger2022general, kahura2022givedirectly}, and fix the horizon $H = 12$ with $N$ denoting the eligible population at the terminal decision.
For example, an $A$-value of $0.07$, evaluated over twelve months for a program transferring \$1,000 to each of 10,000 households amounts to \$84,000 in consumption-equivalent dollars. This corresponds to a research budget of \$84,000 if each dollar of research expenditure entails one dollar of forgone consumption, as for \cite{crosta2024unconditional}'s cost-benefit calculations for lump sum transfers.

We report the monetary equivalent under the calibration $N_j T_j = 10$ \text{ million } in Table \ref{tab:a-value-distribution}. We find that most of the studies have low $A$-value or equivalently, follow up data collection can be justified only by a very small cost of the study. For some studies above the $75^{th}$ percentile, the monetary equivalent is much higher (median $3,000\$$ and mean of $300,000\$$).

These cutoffs help individual researchers decide whether to make an
immediate policy recommendation or recommend collecting further
evidence. Appendix Table~\ref{tab:uct_scale_cost_sensitivity} extends
our numerical illustration by reporting $A$-value cutoffs for
alternative transfer allocations $NT$ and follow-up
experimentation costs $C_{\mathrm{exp}}$, both of which can be
specified ex ante.

Within-country comparisons can also help researchers contextualize how decisive their evidence is for a policy recommendation. A practical reporting exercise would place a study's \(A\)-value within the distribution of previous comparable evaluations, even when only a few are available. For example, without administrative costs, Kenya's mean \(A\)-value is \(0.4\), and six of its seven studies have values below \(0.07\), corresponding to $ 84,000\$$ cost equivalence.

\paragraph{Funding agency problem: selecting ten follow-up studies}
We consider a funding agency choosing ten studies to sponsor
at equal costs. Table \ref{tab:country-a-values-selected-panels}
reports the countries represented among the ten largest $A$-values.
Without administrative costs, the rule selects one of seven
programs in Kenya, one of six in Uganda, one of four in Malawi,
one of two in Liberia, and both programs in Mexico.
With administrative costs, it selects two Kenyan programs
and one Mexican program. Malawi, Uganda, and Liberia are no
longer represented, while Niger, Togo, and Mali enter.
Under both specifications, the mean monetary break-even value
among selected studies exceeds \$1 million in Nepal, Ghana,
and Kenya.

\paragraph{Robustness to follow-up precision} It is interesting to study settings where follow up data collection are significantly more precise instead of equally precise. Changing the follow-up standard-error ratio from $w=1$ to
$w=0.1$ preserves nine of the ten selected programs in each
administrative-cost specification
(Appendix Figure \ref{fig:different_precision}).

Appendix Figure \ref{fig:different_precision2} presents instead settings where the sample size in each follow-up study is set to the average baseline sample size $n \approx 3,000$, instead of assuming a follow-up study equally precise to the original study. Our results continue to hold with $A_j$ replaced by the $A$-value evaluated at $w_j = \sqrt{n_j/3,000}$ where $n_j$ is the sample size in experiment $j$.  We provide numerical results consistent with our main analysis below. 

\begin{table}[!ht]
\centering
\setstretch{1}
\footnotesize
\setlength{\tabcolsep}{3pt}
\caption{Countries with top-ten A-value selections}
\label{tab:country-a-values-selected-panels}
\begin{tabular}{lrrr|rrrr}
\toprule
\multicolumn{8}{l}{\textbf{Panel A: No administrative costs}} \\
\addlinespace[0.4em]
Country
& \shortstack{Mean\\A-value\\(selected)}
& \shortstack{Monetary\\equivalent (\$)\\(selected)}
& \shortstack{Within-country \\ share selected}
& \shortstack{Mean\\$X_{1j}/s_j$\\(selected)}
& \shortstack{Mean\\variance\\(selected)}
& \shortstack{Mean\\A-value\\(all country)}
& \shortstack{Mean\\variance\\(all country)} \\
\midrule
Nepal      & 5.680 & 6,815,621 & 100\% (1)  & -0.104 & 1393.156 & 5.680 & 1393.156 \\
Kenya      & 2.845 & 3,414,505 & 14.3\% (7) &  0.020 & 216.090  & 0.413 & 37.524 \\
Ghana      & 1.208 & 1,449,510 & 100\% (1)  &  0.111 & 65.610   & 1.208 & 65.610 \\
Malawi     & 0.241 & 288,903   & 25\% (4)   & -0.258 & 5.621    & 0.064 & 2.378 \\
Bangladesh & 0.165 & 198,200   & 100\% (1)  &  1.087 & 111.831  & 0.165 & 111.831 \\
Liberia    & 0.139 & 167,141   & 50\% (2)   & -0.718 & 15.210   & 0.070 & 7.650 \\
Mexico     & 0.106 & 127,406   & 100\% (2)  &  0.556 & 34.817   & 0.106 & 34.817 \\
Uganda     & 0.072 & 86,884    & 16.7\% (6) &  0.400 & 1.000    & 0.024 & 9.825 \\
Tanzania   & 0.070 & 83,892    & 100\% (1)  &  0.897 & 8.410    & 0.070 & 8.410 \\
\midrule
\addlinespace[0.8em]
\multicolumn{8}{l}{\textbf{Panel B: Administrative costs (24\%)}} \\
\addlinespace[0.4em]
Country
& \shortstack{Mean\\A-value\\(selected)}
& \shortstack{Monetary\\equivalent (\$)\\(selected)}
& \shortstack{Within-country \\ share selected}
& \shortstack{Mean\\$X_{1j}/s_j$\\(selected)}
& \shortstack{Mean\\variance\\(selected)}
& \shortstack{Mean\\A-value\\(all country)}
& \shortstack{Mean\\variance\\(all country)} \\
\midrule
Nepal      & 4.910 & 5,892,394 & 100\% (1)  & -0.157 & 1393.156 & 4.910 & 1393.156 \\
Ghana      & 1.129 & 1,354,962 & 100\% (1)  & -0.136 & 65.610   & 1.129 & 65.610 \\
Kenya      & 1.123 & 1,348,142 & 28.6\% (7) &  0.258 & 109.850  & 0.330 & 37.524 \\
Tanzania   & 0.335 & 402,328   & 100\% (1)  &  0.207 & 8.410    & 0.335 & 8.410 \\
Bangladesh & 0.254 & 304,698   & 100\% (1)  &  0.898 & 111.831  & 0.254 & 111.831 \\
Mexico     & 0.216 & 259,138   & 50\% (2)   &  0.864 & 69.444   & 0.108 & 34.817 \\
Niger      & 0.090 & 107,930   & 100\% (1)  &  0.392 & 1.490    & 0.090 & 1.490 \\
Togo       & 0.059 & 71,042    & 100\% (1)  &  1.136 & 18.062   & 0.059 & 18.062 \\
Mali       & 0.057 & 68,638    & 50\% (2)   &  1.333 & 43.891   & 0.035 & 36.525 \\
\bottomrule
\end{tabular}

\par\vspace{0.6em}
\begin{minipage}{\linewidth}
\footnotesize\raggedright
Only countries with at least one study selected under the respective
global top-ten rule are shown.
Each study is one of the 42 program-level estimates.
Selected-study summaries use only studies selected under the global
top-ten A-value rule; all-country means use all studies in the country.
The selection share is the number selected divided by the total number
of studies in that country. Numbers in parentheses are these total
study counts. Mean variance is the average $s_j^{2}$ in the respective
group. Mean $X_{1j}/s_j$ is the average signed estimate-to-standard-error
ratio among selected studies, using the specification-specific $X_{1j}$.
All A-values are evaluated at $w_j=1$.
In Panel B, the administrative-cost adjustment subtracts
$24/H=24/12=2$ from each study's normalized gross estimate,
using the common horizon $H=12$ months.
A-values are in monthly household consumption per \$100 of total
transfers (2010 USD PPP). The monetary equivalent equals
$1{,}200{,}000$ times the unrounded mean A-value among selected studies,
assuming 10,000 recipient households, \$1,000 transferred per household,
and 12 months of average monthly welfare gains.
\end{minipage}
\end{table}

 \subsection{Comparisons with $p$-values, variance and empirical Bayes}

To understand what drives the \(A\)-value ranking, we compare it with two simple benchmarks. The first selects the ten programs with the largest two-sided \(p\)-values, prioritizing estimates closest to zero relative to their standard errors. The second selects the ten programs with the largest standard errors \(s_j\), prioritizing the least precise estimates. 

Figure \ref{fig:rank-no-admin} compares these rankings with and without administrative costs. The \(A\)-value ranking generally tracks the \(p\)-value ranking more closely than the standard-error ranking. Without administrative costs, the \(A\)-value and \(p\)-value rules share eight of their ten selections; with administrative costs, they share seven. The corresponding overlaps with the standard-error rule are five and six programs.

Table \ref{tab:selection-disagreements-p_value} reports the programs on which these rules disagree. The disagreements with the \(p\)-value rule illustrate the importance of measuring uncertainty in welfare units. Without administrative costs, the \(A\)-value rule selects a study in Bangladesh with \(p=0.277\) and \(s_j=10.575\), whereas the \(p\)-value rule selects a study in Kenya with \(p=0.377\) and \(s_j=1.017\). Although the Kenyan estimate is less decisive relative to its standard error, the uncertainty of the Bangladeshi estimate is substantially larger. This greater uncertainty allows for substantially larger losses from an immediate decision, increasing the value of further experimentation. Their monetary equivalent \(A\)-values are consequently 198,000\$ and 30,000\$, respectively. The \(A\)-value rule also selects a Mexican study with a relatively large standard error of 8.33 despite its smaller \(p\)-value of 0.270.
This distinction becomes even more stark when administrative costs are included. 


The disagreements with the standard-error rule reveal the complementary limitation of prioritizing imprecision alone. Without administrative costs, this rule selects studies in Rwanda and Uganda with standard errors of 11 and 5.8, respectively, even though their \(p\)-values are close to zero. By contrast, the \(A\)-value rule selects another Mexican study with a much smaller standard error of 0.435 but a \(p\)-value of 0.993. 

\begin{figure}[!ht]
    \centering
    \includegraphics[scale = 0.55]{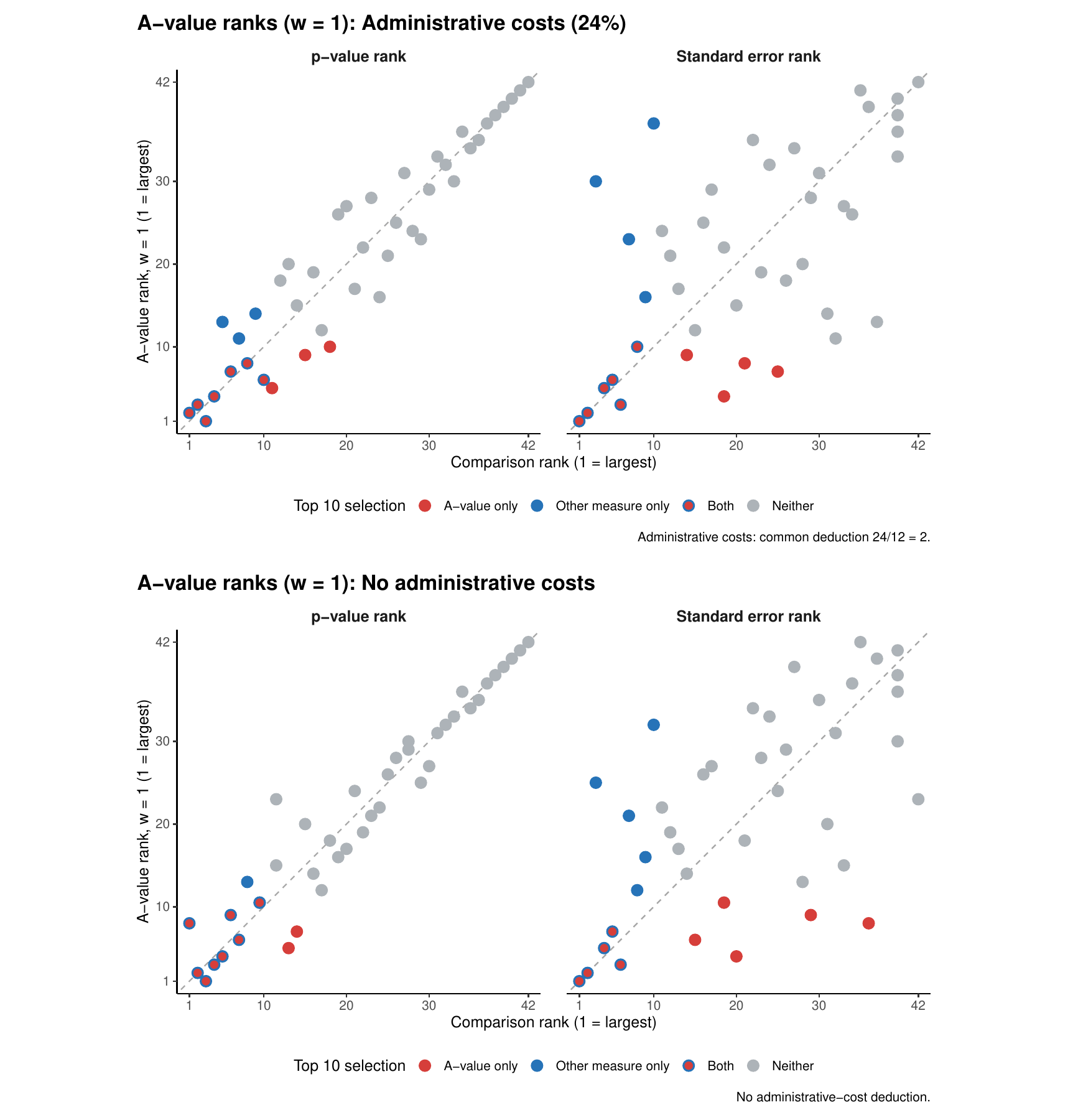}
   \caption{\footnotesize Comparison of rankings for further experimentation. Each point represents one of the 42 programs. The vertical axis reports the \(A\)-value rank; the horizontal axis reports the \(p\)-value rank (left) or standard-error rank (right). Rank 1 corresponds to the largest value. The upper panels include administrative costs equal to 24 percent of transfers; the lower panels exclude administrative costs. Dashed lines indicate identical rankings. Red points are selected only by the \(A\)-value rule, blue points only by the comparison rule, red points with blue outlines by both rules, and gray points by neither. }
    \label{fig:rank-no-admin}
\end{figure}

\begin{table}[!ht]
\centering
\setstretch{1}
\footnotesize
\renewcommand{\arraystretch}{1.1}
\setlength{\tabcolsep}{3pt}
\caption{Top-ten selection disagreements: $A$-value versus $p$-value and standard error}
\label{tab:selection-disagreements}
\label{tab:selection-disagreements-p_value}
\label{tab:selection-disagreements-standard_error}

\begin{tabular*}{\linewidth}{
    @{\extracolsep{\fill}}
    lrrr
    @{\hspace{0.5cm}}
    lrrr
    @{}
}
\toprule
\multicolumn{8}{l}{
    \textbf{Panel A: $A$-value versus $p$-value}
} \\
\addlinespace[0.5em]

\multicolumn{4}{c}{No administrative costs}
& \multicolumn{4}{c}{Administrative costs (24\%)} \\
\multicolumn{4}{c}{\textit{8 of 10 selections shared}}
& \multicolumn{4}{c}{\textit{7 of 10 selections shared}} \\
\cmidrule(lr){1-4}\cmidrule(lr){5-8}

Country
& \shortstack{Monetary\\Equivalent\\$A$-value (\$)}
& $p$ & $s_j$
& Country
& \shortstack{Monetary\\Equivalent\\$A$-value (\$)}
& $p$ & $s_j$ \\
\midrule

\multicolumn{4}{l}{\textit{Selected only by $A$-value}}
& \multicolumn{4}{l}{\textit{Selected only by $A$-value}} \\
\addlinespace[0.2em]
Bangladesh & 198,200 & 0.277 & 10.575
& Bangladesh & 304,698 & 0.369 & 10.575 \\
Mexico & 150,275 & 0.270 & 8.333
& Togo & 71,042 & 0.256 & 4.250 \\
& & &
& Mali & 68,638 & 0.183 & 6.625 \\

\addlinespace[0.6em]
\multicolumn{4}{l}{\textit{Selected only by $p$-value}}
& \multicolumn{4}{l}{\textit{Selected only by $p$-value}} \\
\addlinespace[0.2em]
Kenya & 30,343 & 0.377 & 1.017
& Rwanda & 49,765 & 0.803 & 0.400 \\
Uganda & 83,892 & 0.370 & 2.900
& Myanmar & 52,855 & 0.601 & 0.803 \\
& & &
& Lesotho & 43,091 & 0.506 & 0.896 \\

\addlinespace[0.6em]
\midrule
\addlinespace[0.4em]
\multicolumn{8}{l}{
    \textbf{Panel B: $A$-value versus standard error}
} \\
\addlinespace[0.5em]

\multicolumn{4}{c}{No administrative costs}
& \multicolumn{4}{c}{Administrative costs (24\%)} \\
\multicolumn{4}{c}{\textit{5 of 10 selections shared}}
& \multicolumn{4}{c}{\textit{6 of 10 selections shared}} \\
\cmidrule(lr){1-4}\cmidrule(lr){5-8}

Country
& \shortstack{Monetary\\Equivalent\\$A$-value (\$)}
& $p$ & $s_j$
& Country
& \shortstack{Monetary\\Equivalent\\$A$-value (\$)}
& $p$ & $s_j$ \\
\midrule

\multicolumn{4}{l}{\textit{Selected only by $A$-value}}
& \multicolumn{4}{l}{\textit{Selected only by $A$-value}} \\
\addlinespace[0.2em]
Malawi & 288,903 & 0.796 & 2.371
& Tanzania & 402,328 & 0.836 & 2.900 \\
Liberia & 167,141 & 0.473 & 3.900
& Niger & 107,930 & 0.695 & 1.221 \\
Mexico & 104,537 & 0.993 & 0.435
& Kenya & 98,352 & 0.528 & 1.900 \\
Uganda & 86,884 & 0.689 & 1.000
& Togo & 71,042 & 0.256 & 4.250 \\
Tanzania & 83,892 & 0.370 & 2.900
& & & & \\

\addlinespace[0.6em]
\multicolumn{4}{l}{\textit{Selected only by standard error}}
& \multicolumn{4}{l}{\textit{Selected only by standard error}} \\
\addlinespace[0.2em]
Rwanda & 1,332 & 0.006 & 11.133
& Rwanda & 2,598 & 0.010 & 11.133 \\
South Sudan & 4,438 & 0.021 & 7.700
& South Sudan & 10,403 & 0.040 & 7.700 \\
Mali & 30,900 & 0.102 & 6.625
& Kenya & 30,678 & 0.105 & 6.300 \\
Kenya & 12,184 & 0.053 & 6.300
& Uganda & 41 & $5.89\times10^{-4}$ & 5.858 \\
Uganda & 9 & $1.58\times10^{-4}$ & 5.858
& & & & \\
\bottomrule
\end{tabular*}

\par\vspace{0.6em}
\begin{minipage}{\linewidth}
\footnotesize\raggedright
\textit{Notes:}
Within each comparison and cost specification, only programs selected
by exactly one rule are shown. Each rule selects the ten largest values
of its measure, with ascending program ID breaking exact ties.
Repeated country names within a specification refer to different
programs. All $A$-values are evaluated at $w_j=1$.
The administrative-cost specification subtracts $24/H=24/12=2$
from each study's normalized gross estimate, using the common
horizon $H=12$ months.
Monetary equivalents are calculated from unrounded $A$-values as
$(10{,}000\times1{,}000\times12/100)A=1{,}200{,}000A$,
assuming a total transfer of \$1,000 to each of 10,000 recipient
households and average monthly welfare gains over 12 months.
$p$-values are from two-sided normal tests of a zero welfare effect
under the respective cost specification.
Harmonized standard errors $s_j$ are expressed in monthly household
consumption per \$100 of total transfers (2010 USD PPP).
\end{minipage}
\end{table}

\paragraph{Empirical Bayes}
As a second comparison, we consider a Gaussian empirical Bayes benchmark based on \cite{abadie2023estimating} (Section 2.3), estimated via heteroskedastic maximum likelihood. Empirical Bayes uses evidence across programs to estimate the distribution of policy effects, allowing each evaluation to benefit from information contained in other studies. This approach is particularly useful when a collection of comparable evaluations is available at the time of the decision, but not feasible when the study effect may not have a reference distribution when making a policy recommendation.

We follow \cite{abadie2023estimating}, Section 2.3 and consider a common Gaussian prior by maximum likelihood, allowing sampling variances to differ across the 42 programs. The estimated prior has mean \(2.985\) and standard deviation \(2.088\) without administrative costs (and subtracts 24/12 = 2 from the mean with administrative costs). We then compute the expected welfare gain from an independent follow-up experiment with the same precision.\footnote{Non parametric empirical Bayes as in \cite{chen2026empirical} is also interesting but requires a stand for tuning parameters in our decision problem and, due to its non-parametric nature, a larger number of studies. } 

Figure \ref{fig:comparison-eb} shows substantial differences between the rankings. The two rules share five of their ten selected programs without administrative costs and three with administrative costs. Pooling also affects immediate recommendations: without administrative costs, all 42 posterior means are positive, including those for the four programs with negative initial estimates. With administrative costs, 29 posterior means remain positive. 

Table \ref{tab:selection-disagreements-V_EB} reports the programs where the two ranking disagree. 
The program in Nepal illustrates the mechanism behind these differences. Its estimated consumption effect is \(-3.875\), with a standard error of \(37.325\), but its posterior mean is \(2.963\): the empirical Bayes calculation assigns approximately \(99.7\%\) of the weight to the fitted prior mean. This program ranks first by \(A\)-value and near to last by empirical Bayes value in both specifications. Under the fitted model, a second estimate with the same large standard error would provide little information relative to the existing posterior and would be very unlikely to change the policy recommendation. The large \(A\)-value reflects the welfare losses that remain possible when assessing this imprecise evaluation without relying on the chosen prior distribution. The empirical Bayes procedure on the other hand strongly relies on the underlying exchangeability/prior assumption.  
 These differences illustrate complementarities of the two approaches. 
 
A similar pattern appears in Kenya and Bangladesh. Without administrative costs, programs with standard errors of \(14.7\) and \(10.6\), respectively, rank second and fifth by \(A\)-value but 38th and 37th by empirical Bayes value. The fitted prior mean receives approximately \(98\%\) and \(96\%\) of the weight in their respective posterior means, making an equally imprecise follow-up experiment unlikely to change the policy recommendation. Both programs remain selected by the top-ten \(A\)-value rule and excluded by the empirical Bayes rule.

\begin{figure}[!ht]
    \centering
    \includegraphics[scale = 0.55]{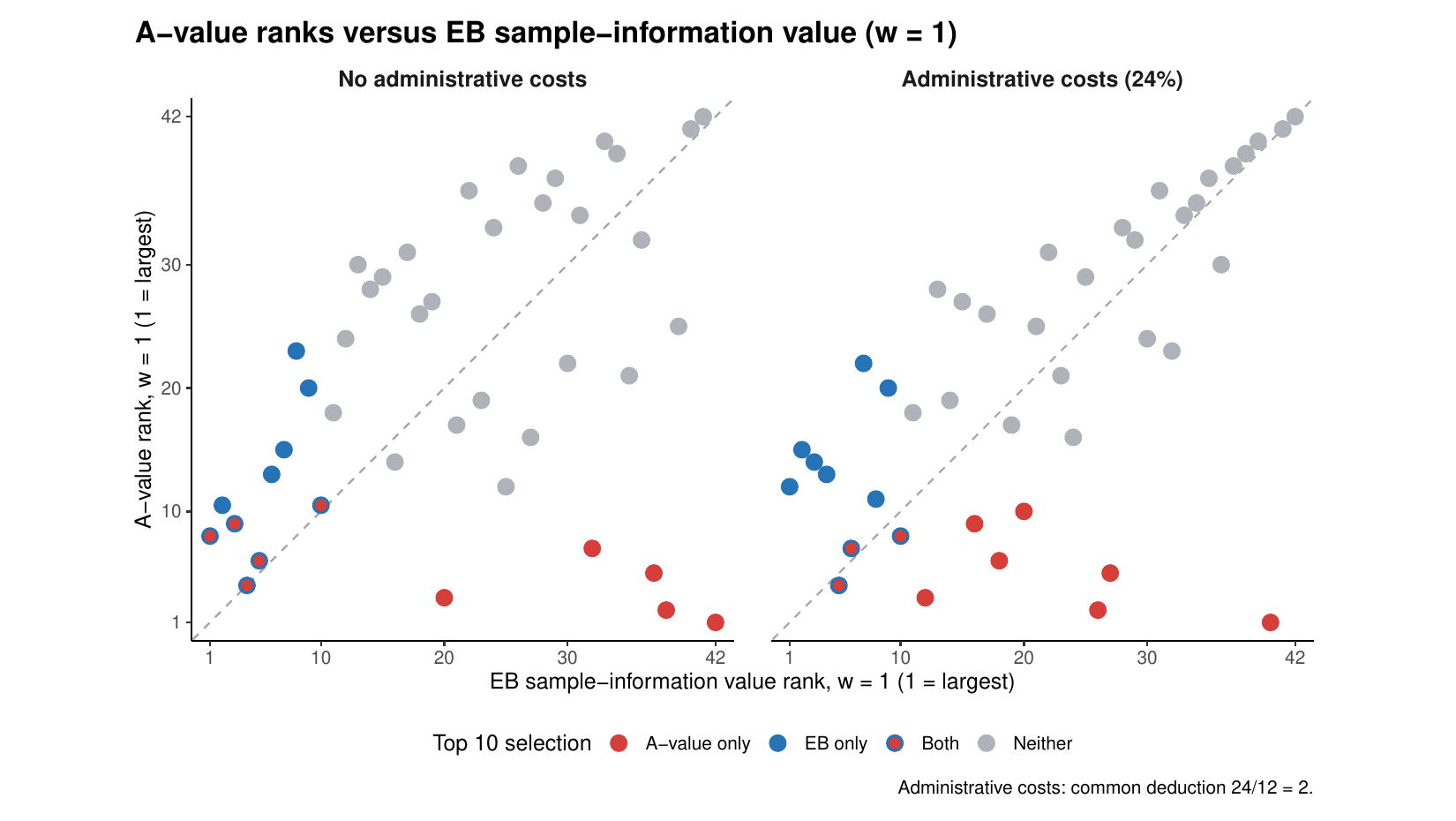}
    \caption{ \footnotesize Comparison of $A$-value and empirical Bayes rankings.
Each point represents one of the 42 programs, with its $A$-value rank
on the vertical axis and empirical Bayes sample-information value
rank on the horizontal axis. Rank 1 denotes the largest value.
The left panel excludes administrative costs; the right panel includes
administrative costs equal to 24\% of transfers.
The dashed diagonal indicates identical ranks.
Red points denote programs selected only by the $A$-value rule,
blue points those selected only by empirical Bayes, red points with
blue outlines those selected by both, and gray points those selected
by neither. 
Empirical Bayes values condition on the initial evidence and hold
the common Gaussian prior.}
    \label{fig:comparison-eb}
\end{figure}

\begin{table}[p]
\centering
\setstretch{1}
\small
\renewcommand{\arraystretch}{1.08}
\setlength{\tabcolsep}{4pt}
\caption{Top-ten selection disagreements: $A$-value versus empirical Bayes}
\label{tab:selection-disagreements-V_EB}

\scalebox{0.8}{\begin{tabular*}{\linewidth}{
    @{\extracolsep{\fill}}lrrrrrr@{}
}
\toprule
\multicolumn{7}{l}{
    \textbf{Panel A: No administrative costs}
} \\
\addlinespace[0.4em]

Country
& \shortstack{Monetary\\Equivalent\\$A$-value (\$)}
& $V_j^{\mathrm{UR}}$
& \shortstack{Estimate\\$X_{1j}$}
& \shortstack{Standard error\\$s_j$}
& \shortstack{EB Prior weight\\(\%)}
& \shortstack{EB Posterior mean} \\
\midrule

\multicolumn{7}{l}{\textit{Selected only by the $A$-value rule}} \\
\addlinespace[0.2em]
Nepal      & 6,815,621 & $<0.001$ & -3.875 & 37.325 & 99.7 & 2.963 \\
Kenya      & 3,414,505 & $<0.001$ &  0.300 & 14.700 & 98.0 & 2.932 \\
Ghana      & 1,449,510 & $<0.001$ &  0.900 &  8.100 & 93.8 & 2.855 \\
Bangladesh &   198,200 & $<0.001$ & 11.500 & 10.575 & 96.2 & 3.304 \\
Mexico     &   150,275 & $<0.001$ &  9.200 &  8.333 & 94.1 & 3.352 \\

\addlinespace[0.5em]
\multicolumn{7}{l}{\textit{Selected only by the empirical Bayes rule}} \\
\addlinespace[0.2em]
Uganda  & 83,892 & 0.042    & -2.600 & 2.900 & 65.9 & 1.078 \\
Kenya   & 30,343 & 0.004    &  0.897 & 1.017 & 19.2 & 1.297 \\
Malawi  & 18,337 & 0.003    &  0.800 & 0.800 & 12.8 & 1.080 \\
Malawi  &  2,292 & 0.002    &  0.100 & 0.100 &  0.2 & 0.107 \\
Lesotho &  5,025 & $<0.001$ &  1.404 & 0.896 & 15.5 & 1.650 \\

\addlinespace[0.5em]
\midrule
\addlinespace[0.4em]
\multicolumn{7}{l}{
    \textbf{Panel B: Administrative costs (24\%)}
} \\
\addlinespace[0.4em]

Country
& \shortstack{Monetary\\Equivalent\\$A$-value (\$)}
& $V_j^{\mathrm{UR}}$
& \shortstack{Estimate\\$X_{1j}$}
& \shortstack{Standard error\\$s_j$}
& \shortstack{EB Prior weight\\(\%)}
& \shortstack{EB Posterior mean} \\
\midrule

\multicolumn{7}{l}{\textit{Selected only by the $A$-value rule}} \\
\addlinespace[0.2em]
Nepal      & 5,892,394 & $<0.001$ & -5.875 & 37.325 & 99.7 & 0.963 \\
Kenya      & 2,597,933 & $<0.001$ & -1.700 & 14.700 & 98.0 & 0.932 \\
Ghana      & 1,354,962 & 0.008    & -1.100 &  8.100 & 93.8 & 0.855 \\
Bangladesh &   304,698 & $<0.001$ &  9.500 & 10.575 & 96.2 & 1.304 \\
Mexico     &   259,138 & $<0.001$ &  7.200 &  8.333 & 94.1 & 1.352 \\
Togo       &    71,042 & 0.003    &  4.829 &  4.250 & 80.6 & 1.732 \\
Mali       &    68,638 & $<0.001$ &  8.829 &  6.625 & 91.0 & 1.693 \\

\addlinespace[0.5em]
\multicolumn{7}{l}{\textit{Selected only by the empirical Bayes rule}} \\
\addlinespace[0.2em]
Liberia & 51,984 & 0.185 & -4.800 & 3.900 & 77.7 & -0.304 \\
Malawi  & 42,961 & 0.128 & -2.612 & 2.371 & 56.3 & -0.587 \\
Lesotho & 43,091 & 0.090 & -0.596 & 0.896 & 15.5 & -0.350 \\
Rwanda  & 49,765 & 0.082 & -0.100 & 0.400 &  3.5 & -0.062 \\
Uganda  & 15,456 & 0.061 & -4.600 & 2.900 & 65.9 & -0.922 \\
Myanmar & 52,855 & 0.046 &  0.420 & 0.803 & 12.9 &  0.493 \\
Kenya   & 19,178 & 0.038 & -1.103 & 1.017 & 19.2 & -0.703 \\
\bottomrule
\end{tabular*}
}

\par\vspace{0.6em}
\begin{minipage}{\linewidth}
\footnotesize\raggedright
\textit{Notes:}
Only programs selected by exactly one rule are shown.
Both rules select their ten largest values within each specification.
Both $A$-values and empirical Bayes values use $w_j=1$, so the
follow-up standard error equals $s_j$.
$V_j^{\mathrm{UR}}$ is the expected empirical Bayes welfare gain from a
follow-up experiment conditional on the initial evidence, excluding
the cost of conducting the experiment.
The common Gaussian prior is fitted to consumption effects before
deducting administrative costs, with $\hat\mu_0=2.985$ and
$\hat\tau^2=4.360$.
Prior weight is reported as
$100(1-\hat\tau^2/(\hat\tau^2+s_j^2))$.
In Panel B, the administrative-cost deduction is $24/H=2$ for every
program, with $H=12$ months; both $X_{1j}$ and the reported EB
posterior mean are net of this deduction.
Monetary equivalents equal $1{,}200{,}000$ times the unrounded
$A$-value, assuming 10,000 recipient households, \$1,000 transferred
per household, and 12 months of average monthly welfare gains.
\end{minipage}
\end{table}

\paragraph{Calibrated simulations}
To evaluate the procedure, we consider the following simulation exercise
calibrated to the 42 program-level estimates. For each program, we set
the true welfare effect equal to $\delta_j=\widehat b_j$ without
administrative costs and $\delta_j=\widehat b_j-24/H$ with
administrative costs. These effects and the reported standard errors
$s_j$ remain fixed across simulations. In each replication $r$,
we generate initial and follow-up signals according to
$ 
X_{tjr}=\delta_j+s_j\varepsilon_{tjr},
\varepsilon_{tjr}\overset{\mathrm{iid}}{\sim}N(0,1),
t\in\{1,2\},
$ 
where $r$ indexes replications. Thus, both experiments have the
precision of the original study.

For each experimentation budget $q\in\{5,\ldots,20\}$, we compare
rules selecting the $q$ programs with the largest $A$-values,
two-sided $p$-values, standard errors, or empirical Bayes
sample-information values. Selection uses only the initial evidence;
the follow-up signal is available to the implementation rule only
for selected programs. For empirical Bayes, we re-estimate the common
Gaussian prior by maximum likelihood from the simulated data.

\begin{figure}[!ht]
    \centering
    \includegraphics[scale=0.6]{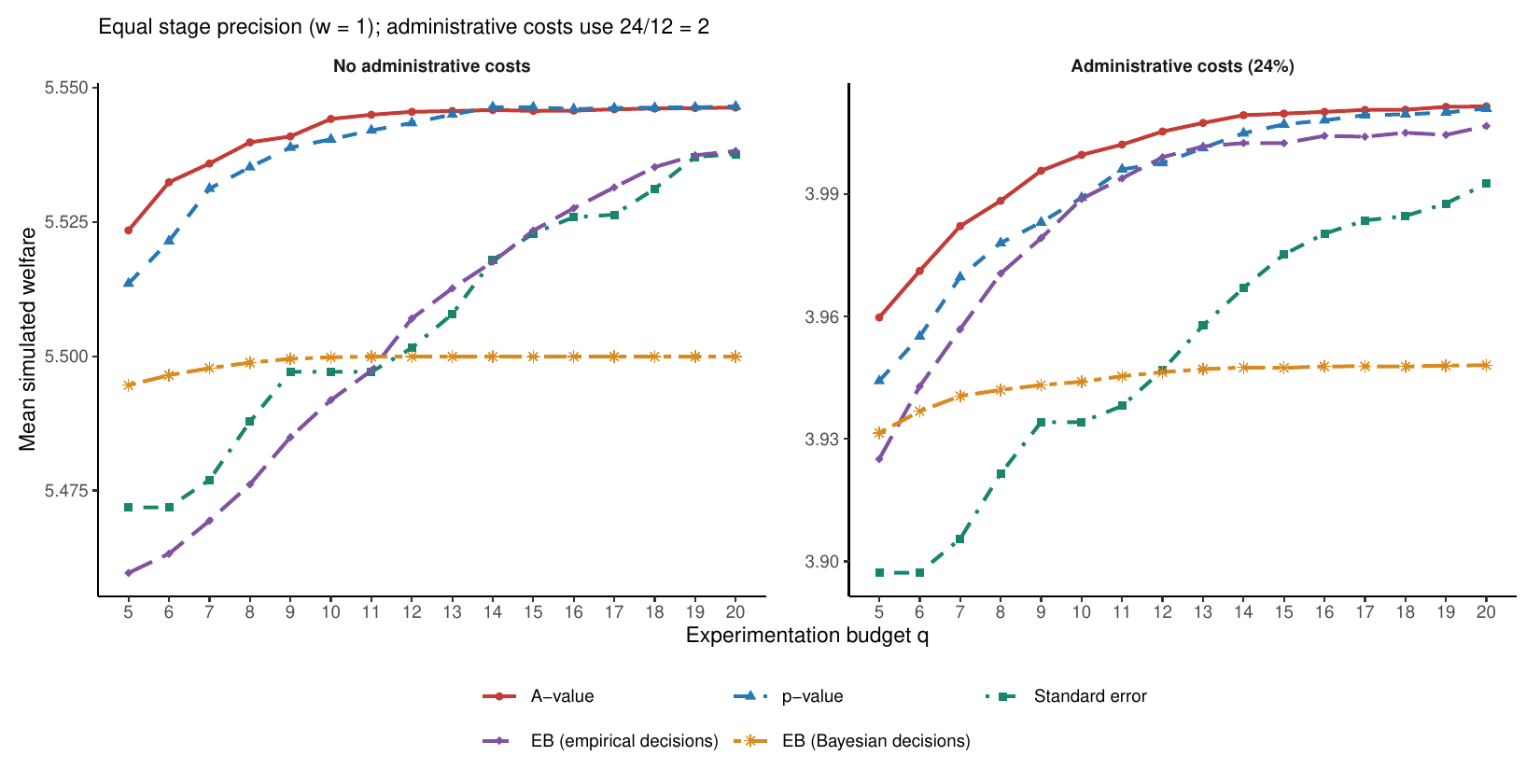}
    \caption{\footnotesize Calibrated simulations: welfare under alternative follow-up
    selection rules. Curves report averages over 500 simulations calibrated
    to the 42 program-level estimates. 
    The $A$-value, $p$-value, standard-error, and EB (empirical decisions)
    rules implement programs when the initial net signal is nonnegative
    for unselected programs, or when the average of the two net signals
    is nonnegative for selected programs. EB (Bayesian decisions) uses
    the same follow-up selections as EB (empirical decisions), but
    implements when the posterior net mean is nonnegative.
    The Gaussian EB prior is fitted to initial signals in each
    replication and held fixed during follow-up.
    Experimentation costs are represented by the budget $q$ and are not
    deducted separately from welfare. }
    \label{fig:placeholder}
\end{figure}

We distinguish the choice of programs for further experimentation
from the subsequent implementation decision. Under the $A$-value,
$p$-value, standard-error, and EB (empirical decisions) rules,
an unselected program is implemented when $X_{1jr}\geq0$, while a
selected program is implemented when $X_{1jr}+X_{2jr}\geq0$.
We also report EB (Bayesian decisions), which selects the same
programs for follow-up as EB (empirical decisions), but implements
when the posterior welfare is nonnegative given all the available evidence.

Figure \ref{fig:placeholder} reports mean simulated welfare
across experimentation budgets. In both specifications, the $A$-value rule achieves the highest
simulated welfare with administrative costs, and without administrative cost it is larger for budget below 13 experiments and comparable with the $p$-value ranking above 13. Intuitively, $A$-values prioritize programs
for which uncertainty about the policy recommendation is accompanied
by potentially large welfare losses. The $p$-value rule captures
uncertainty about the sign relative to the standard error, but ignores
its scale in welfare units. It can therefore allocate scarce experiments
to programs for which an incorrect recommendation has relatively small
welfare consequences. As the budget increases, this distinction becomes
less consequential because both rules can investigate more programs.

The empirical Bayes comparison illustrates the benefits and limitations
of pooling evidence across programs in this calibrated simulation. Pooling can improve immediate recommendations, particularly when
individual estimates are imprecise. However, strong shrinkage can also produce persistent
errors when a program's effect differs from the fitted prior.
Moreover, the EB information criterion assigns little value to
experiments unlikely to change these posterior recommendations.
These features help explain why Bayesian welfare gains flatten as
the budget expands, while the $A$-value rule continues to benefit from
additional evidence without imposing a common prior across programs.

\section{Guide to practice} \label{sec:conclusion}

The $A$-value connects the strength of existing evidence to the welfare cost of collecting more data. It can help individual researchers make recommendations about further experimentation and funding agencies allocate research resources. We propose the following guide.

\vspace{3mm}

\noindent \textbf{Individual research team:}

We envision researchers report results as in Table \ref{tab:a_value_report} as a template. 

\begin{itemize}

\item \textit{Report the $A$-value and the decision it informs (``Evidence'' in Table \ref{tab:a_value_report}).}
Specify the intervention and the target outcome. Report the estimated welfare effect $X_1$, any intervention cost $\kappa$ used for computing $X_1$, its standard error, sample size, $p$-value and the $A$-value. A smaller $A$-value supports an immediate recommendation. We recommend reporting as a benchmark the standardized  $A_{\mathrm{value}}^\star$ in Corollary \ref{rem:scale_invariant_a_value}; in these case experimentation costs should be interpret per outcome standard deviation. When available, researchers may also report the average $A$-value observed in similar contexts.

\item \textit{Translate the break-even cost into interpretable units (``Calibration and costs'' in Table \ref{tab:a_value_report}).}
When feasible, we encourage researchers to accompany the $A$-value with a monetary calibration, as in our applications. 
This calibration may use information about the number of units expected to be affected by scaling up the intervention ($N$), the monetary value per each unit ($V$), the welfare horizon $H$ for the intervention, the intervention overall budget $K$ (not needed for drug approval). The reader may refer to Section \ref{sec:examples} for translating these quantities into monetary calibrations. 

\end{itemize}

\begin{table}[!ht]
\centering
\caption{Illustrative reporting template for the $A$-value}
\label{tab:a_value_report}
\begingroup
\setstretch{1}
\small
\setlength{\tabcolsep}{4pt}
\renewcommand{\arraystretch}{1.08}

\scalebox{0.85}{\begin{tabular*}{\linewidth}
{@{\extracolsep{\fill}}lrlr@{}}
\toprule
\multicolumn{4}{@{}l@{}}{
\textit{Intervention:} [name]
\qquad
\textit{Welfare outcome:} [definition and units]
}\\
\addlinespace[4pt]
\multicolumn{2}{@{}l}{\textbf{A. Evidence}}
&
\multicolumn{2}{l@{}}{\textbf{B. Calibration and costs}}
\\
\cmidrule(lr){1-2}\cmidrule(l){3-4}
Estimate net of cost (standard error)
& [$X_1$ ($s_1$)]
& Population affected
& [$N$]
\\
Per unit intervention cost
& [$\kappa$]
& Monetary value per outcome unit
& [$V$]
\\
Two-sided $p$-value
& [$p$]
& Welfare horizon
& [$H$]
\\
Initial sample size (treatment probability)
& [$n_1 (\pi)$]
& Intervention budget (\texttt{NA} for drug approval)
& [$K$]
\\
Normalization 
& [$\sigma$]
& Monetary $A$-value
& [$C_A$]
\\
Standardized $A$-value
& [$A^\star$]
& Original study cost
& [$C_1$]
\\
Context mean $A^\star$ (studies)
& [$\overline A^\star$ ($K$)]
& Proposed follow-up cost (if available)
& [$C_{\mathrm{exp}}$]
\\
\bottomrule
\end{tabular*}
}
\vspace{5pt}

\scalebox{0.85}{\begin{tabular*}{\linewidth}
{@{\extracolsep{\fill}}lcccc@{}}
\multicolumn{5}{@{}l@{}}{
\textbf{C. Sensitivity of the recommendation}
}\\
\addlinespace[3pt]
Scenario
& \shortstack{Standardized\\$A$-value}
& \shortstack{Monetary\\$A$-value}
& \shortstack{Follow-up\\cost}
& \shortstack{Recommend\\follow-up?}
\\
\midrule
Baseline: $w=1$
& [$A^\star$]
& [$C_A$]
& [$C_{\mathrm{exp}}$]
& [Yes/No]
\\
Alternative calibration: $[N,V,H]$
& [$A^\star$]
& [$C_A'$]
& [$C_{\mathrm{exp}}$]
& [Yes/No]
\\
Alternative design: $w=[\,\cdot\,]$
& [$A_{\mathrm{std}}(w)$]
& [$C_A(w)$]
& [$C_{\mathrm{exp}}(w)$]
& [Yes/No]
\\
\bottomrule
\end{tabular*}
}
\par\vspace{4pt}
\begin{minipage}{\linewidth}
\footnotesize
\textit{Notes:}
The baseline assumes an equally precise follow-up.
Here
$A_{\mathrm{std}}(w)=A_{\mathrm{value}}(X_1;s_1,w)/\sigma$
and $A^\star=A_{\mathrm{std}}(1)$.
The monetary threshold is
$C_A(w)=M\sigma A_{\mathrm{std}}(w)$,
where $M$ converts one unit of the reported effect into
total monetary welfare using the stated population, valuation and
horizon.
Recommend follow-up when $C_{\mathrm{exp}}(w)<C_A(w)$,
with costs expressed in matching monetary welfare units.
Repeat the sensitivity rows for plausible alternative calibrations
and designs, updating costs for each design.
Context mean $A^\star$ should report the standardazied $A$-value observed in comparable contexts (when available) for calibration. 
\end{minipage}
\endgroup
\end{table}

\vspace{2mm} 

\noindent \textbf{Funding agencies and meta-science studies}
\begin{itemize} 
\item \textit{Allocate follow-up funding using costs and welfare priorities.} $A$-values can be used also from funding agencies, that can review evidence across several interventions. 
With equal costs and welfare weights, it can select the $q$ studies with the largest benchmark $A$-values. With heterogeneous costs $c_j$, $c_j$ may simply correspond to monetary research costs standardized so that $\max_j c_j = 1$, and divide each $A$-value by $c_j$. If the agency places preferences with welfare weights $\omega_j$ over different interventions, for example reflecting population sizes, multiply the $A$-value by $\omega_j$.

\item \textit{Build evidence repositories.}
Databases assembled for systematic reviews and meta-analyses may report standardized $A$-values and current research costs. Researchers could use these records to benchmark new findings. 

\end{itemize}

By introducing a novel notion of economic significance, this paper also opens several 
questions for future research. One is to extend our notion of $A$-value under minimax regret with prior information or model-based analog used for site selection models  \citep[e.g.][]{olea2024externally, gechter2024selecting} in the spirit of our extension in Appendix \ref{sec:prior_information}. A second is to choose the follow-up sample size as a function of the initial evidence, extending our Appendix  \ref{rem:endogenous_precision}. A third is to accommodate settings where unbiased estimates of welfare effects are unavailable \citep[e.g.][]{cesa2025adaptive}.  

 \bibliographystyle{plainnat}
\bibliography{bibliography}

\newpage 

\appendix 

\begin{center}
    \LARGE Online Appendix 
\end{center}


\numberwithin{equation}{section}

\section{Extensions details} \label{sec:extensions}

\subsection{Multiple time periods}
\label{sec:multiple_periods}

We now allow researchers to collect evidence over multiple periods. The setup extends Assumption \ref{ass:2}, with researchers choosing in each period whether to recommend or reject the policy or to defer the decision and collect additional evidence.

\begin{ass}[Multiple rounds of data collection]
\label{ass:multiple_periods}
In each period $h=1,\ldots,H$, researchers observe
$
X_h\mid\Delta
\sim
\mathcal{N}(\Delta,s_h^2),
X_1,\ldots,X_H
\text{ mutually independent}.
$
After observing $(X_1,\ldots,X_h)$, researchers make a decision
$
\delta_h(X_1,\ldots,X_h)
\in
\{0,1,\texttt{NA}\},
 h<H,
$ where $\delta_h=1$ corresponds to recommending the policy,
$\delta_h=0$ corresponds to rejecting the policy, and
$\delta_h=\texttt{NA}$ corresponds to collecting additional evidence.

If $\delta_h=\texttt{NA}$, researchers collect $X_{h+1}$ at welfare cost
$c_h\geq0$. In the terminal period, no additional evidence can be collected and $
\delta_H(X_1,\ldots,X_H)\in\{0,1\}.
$
\qed
\end{ass}

Assumption \ref{ass:multiple_periods} allows researchers to defer the policy recommendation several times. The standard errors $s_h$ and the costs $c_h$ may vary across periods. For example, researchers may first conduct an inexpensive pilot study, subsequently collect evidence in a larger experiment, and finally conduct an expensive confirmatory experiment. Denote by
$$
\small 
\begin{aligned} 
\mathsf{h}(\boldsymbol{\delta})
=
\inf\left\{
h\leq H:
\delta_h(X_1,\ldots,X_h)\in\{0,1\}
\right\}
\end{aligned} 
$$

the period in which the terminal policy decision is made, where
$\boldsymbol{\delta}=(\delta_1,\ldots,\delta_H)$. The realized loss equals
$
\ell_H(\boldsymbol{\delta};\boldsymbol{c},\Delta)
=
-\left(
\Delta
1\left\{
\delta_{\mathsf{h}(\boldsymbol{\delta})}=1
\right\}
-
\sum_{r<\mathsf{h}(\boldsymbol{\delta})}c_r
\right)
$
where $\boldsymbol{c}=(c_1,\ldots,c_{H-1})$. Thus, if researchers stop in period $h$, they implement or reject the policy using all the evidence collected up to that period and incur the cumulative experimentation cost
$
\sum_{r<h}c_r.
$
As in the two-period problem, we compare the researcher's loss to the highest achievable welfare under knowledge of $\Delta$. Expected regret and minimax regret rules are 
\begin{equation} \label{eqn:regret2}
\small 
\begin{aligned} 
\mathrm{Reg}_H(
\boldsymbol{\delta};
\Delta,\boldsymbol{c}
)
=
\mathbb{E}_\Delta\left[
\ell_H(
\boldsymbol{\delta};
\boldsymbol{c},\Delta
)
\right]
+
\Delta 1\{\Delta\geq0\}, \qquad \boldsymbol{\delta}^\star
\in
\operatorname*{arg\,min}_{\boldsymbol{\delta}}
\sup_{\Delta:|\Delta|\leq\bar{\Delta}}
\mathrm{Reg}_H(
\boldsymbol{\delta};
\Delta,\boldsymbol{c}
). 
\end{aligned} 
\end{equation} 
The multiple-period problem introduces two features relative to the two-period problem. First, the value of collecting evidence in period $h+1$ depends on all the subsequent experimentation opportunities. Second, reaching a later period is endogenous because it depends on the evidence collected and the continuation decisions made in the previous periods. Minimax-regret therefore evaluates the costs and policy consequences over the entire path.

Despite this path dependence, the Gaussian structure reduces the relevant history to a one-dimensional statistic. Define cumulative precision and precision-weighted evidence as
$ 
S_h^2
=
\sum_{r=1}^h\frac{1}{s_r^2},
T_h
=
\frac{\sum_{r=1}^h\frac{X_r}{s_r^2}}{S_h}.
$
Note that $\mathbb{V}(T_h) = 1$ so that the statistic $T_h$ is standardized. 
We can now state the multiple-period extension of Lemma \ref{thm:dead_zone}.

\begin{theorem}[Multiple-period rule]
\label{thm:multi_period_dead_zone}
Under Assumption \ref{ass:multiple_periods}, with $s_h^2 > 0$ for all $h$, $\bar{\Delta} \in (0,\infty)$, there exist minimax-regret optimal decisions
$
\boldsymbol{\delta}^\star
=
(\delta_1^\star,\ldots,\delta_H^\star)
$
and thresholds
$
\tau_h^\star\in[0,\infty],
h=1,\ldots,H,
$
with $\tau_H^\star=0$, such that, for every $h<H$,
$$
\small 
\begin{aligned} 
\delta_h^\star(X_1,\ldots,X_h)
=
\begin{cases}
1
& \text{if } T_h \geq\tau_h^\star,\\
0
& \text{if }  T_h \leq-\tau_h^\star,\\
\texttt{NA}
& \text{if } |T_h|<\tau_h^\star,
\end{cases}, \qquad \delta_H^\star(X_1,\ldots,X_H)
=
1\{T_H\geq0\}. 
\end{aligned} 
$$
To define the thresholds jointly, let
$
\boldsymbol{\tau}
=
(\tau_1,\ldots,\tau_{H-1},0)
\in
\mathcal{T}_H
=
[0,\infty]^{H-1}\times\{0\},
$
and denote
$$
\small 
\begin{aligned} 
\sigma(\boldsymbol{\tau})
=
\inf\left\{
h\leq H:
|T_h|\geq\tau_h
\right\}, \qquad R_H(
\boldsymbol{\tau};
\Delta,\boldsymbol{c}
)
=
\Delta
\underbrace{
\mathbb{P}_\Delta\left(
T_{\sigma(\boldsymbol{\tau})}<0
\right)
}_{\substack{\text{probability of an incorrect}\\
\text{terminal policy decision}}}
+
\sum_{h=1}^{H-1}
c_h
\underbrace{
\mathbb{P}_\Delta\left(
\sigma(\boldsymbol{\tau})>h
\right)
}_{\text{probability of collecting }X_{h+1}}
\end{aligned} 
$$
The optimal thresholds can be selected so that
$
\boldsymbol{\tau}^\star
\in
\operatorname*{arg\,min}_{\boldsymbol{\tau}\in\mathcal{T}_H}
\max_{\Delta\in[0,\bar{\Delta}]}
R_H(
\boldsymbol{\tau};
\Delta,\boldsymbol{c}
).
$

In addition, any minimax-regret optimal rule equals almost surely
$\delta_1^\star$ in the first period and $\delta_h^\star$
on the event of continuation up to time $h$, under every
$\Delta\in[-\bar\Delta,\bar\Delta]$.\footnote{Formally, for every
$\Delta\in[-\bar\Delta,\bar\Delta]$ and any maxmin regret optimal rule $\delta'$, and every $h\leq H$,
$ 
\mathbb P_\Delta\left(
\sigma(\boldsymbol{\tau}^\star)\geq h,\,
\delta_h'(X_1,\ldots,X_h)
\neq
\delta_h^\star(X_1,\ldots,X_h)
\right)=0.
$}
\end{theorem}

\begin{proof}
See Appendix \ref{proof:multi_period_dead_zone}.
\end{proof}

Theorem \ref{thm:multi_period_dead_zone} generalizes the confidence-interval representation in Lemma \ref{thm:dead_zone}. In every non-terminal period, sufficiently positive estimates lead researchers to recommend the policy, sufficiently negative estimates lead researchers to reject it, and intermediate estimates lead to additional data collection. The sign of $T_h$ determines the direction of the policy recommendation, while $|T_h|$ measures the strength of the accumulated evidence.

The definition of $\tau_h^\star$ provides a direct economic interpretation as it depends on the cost of deferral in each period and the precision of each subsequent estimate.

To gain intuition on the behavior of $\tau_h$, Figure \ref{fig:threshold} reports the threshold for a five-period horizon as a function of the costs of future data collections, assuming for simplicity $s_1 = s_2 = \cdots = s_H = 1$ and costs are constant over each period ($c_1 = c_2 = \cdots c_{H -1}$), with $\Delta \in \mathbb{R}$.  The figure shows two expected patterns. First, as the number of subsequent periods decreases (i.e., we move from the first period $h = 1$, to $h = 2$ and so on), the threshold becomes less stringent (smaller). Second, as the continuation cost increases, the threshold also becomes less stringent due to higher costs of future data collection.  

\begin{figure}[H]
    \centering
    \includegraphics[scale = 0.5]{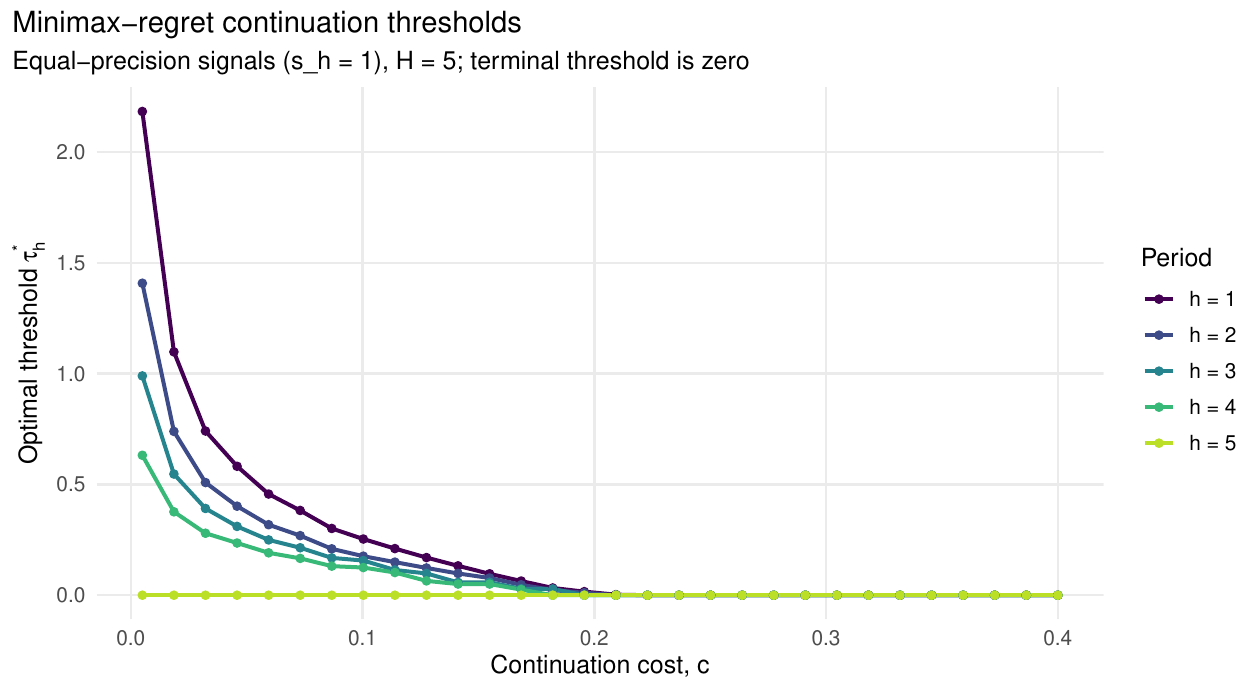}
    \caption{\footnotesize Optimal threshold in Theorem \ref{thm:multi_period_dead_zone} with five experimentation periods.}
    \label{fig:threshold}
\end{figure}

\subsection{Asymmetric regret loss}
\label{sec:asymmetric_regret}

We now allow the two types of incorrect policy decisions to have
different welfare consequences similar to models in \cite{tetenov2012}. Researchers incur regret equal to
$\Delta$ when rejecting a policy with $\Delta>0$, and regret equal
to $K|\Delta|$ when implementing a policy with $\Delta<0$, where
$K>0$. Thus, $K>1$ places a larger weight on implementing a harmful
policy. We maintain the two-period experiment in Assumptions
\ref{ass:experiment} and \ref{ass:2}. The cost of collecting
follow-up evidence remains $c\geq0$.

\begin{defn}[Asymmetric regret]
\label{defn:asymmetric_regret}
For $K>0$ and a measurable strategy
$(\delta_1,\delta_2)\in\mathcal D$, define
\[
\small 
\begin{aligned}
\mathrm{Reg}_{1,K}
(\delta_1,\delta_2;\Delta,c)
={}&
\Delta_+\mathbb E_\Delta\left[
1\{\delta_1=0\}
+1\{\delta_1=\texttt{NA}\}1\{\delta_2=0\}
\right]\\
&+K\Delta_-\mathbb E_\Delta\left[
1\{\delta_1=1\}
+1\{\delta_1=\texttt{NA}\}1\{\delta_2=1\}
\right]
+c\mathbb P_\Delta(\delta_1=\texttt{NA}),
\end{aligned}
\]
where $\delta_1$ and $\delta_2$ are evaluated at $X_1$ and
$(X_1,X_2)$, respectively, and
$\Delta_+=\max\{\Delta,0\}$,
$\Delta_-=\max\{-\Delta,0\}$.
\qed
\end{defn}

The criterion coincides with $\mathrm{Reg}$ when $K=1$.
Researchers solve
\[
(\delta_1^\star,\delta_2^\star)
\in
\operatorname*{arg\,min}_{(\delta_1,\delta_2)\in\mathcal D}
\sup_{|\Delta|\leq\bar\Delta}
\mathrm{Reg}_{1,K}
(\delta_1,\delta_2;\Delta,c).
\]

Write
$ 
T_1=\frac{X_1}{s_1},
T_2=
\frac{X_1/s_1^2+X_2/s_2^2}
     {\sqrt{s_1^{-2}+s_2^{-2}}}
$ 
for the standardized initial and pooled evidence. 

\begin{theorem}[Rule under asymmetric regret]
\label{thm:asymmetric_dead_zone}
Under Assumptions \ref{ass:experiment} and \ref{ass:2}, with
$s_1,s_2>0$, $\bar\Delta\in(0,\infty)$, $c\geq0$, and $K>0$,
there exist minimax-regret optimal decisions
$(\delta_1^\star,\delta_2^\star)$ and thresholds
$ 
-\infty\leq\underline\tau^\star
\leq\overline\tau^\star\leq\infty, \tau_2^\star\in\mathbb R,
$ 
such that
\[
\small 
\begin{aligned} 
\delta_1^\star(X_1)
=
\begin{cases}
1,&T_1\geq\overline\tau^\star,\\
0,&T_1<\underline\tau^\star,\\
\texttt{NA},&
\underline\tau^\star\leq T_1<\overline\tau^\star,
\end{cases}
\qquad
\delta_2^\star(X_1,X_2)=1\{T_2\geq\tau_2^\star\}.
\end{aligned} 
\]

To define the thresholds jointly, let
$ 
\mathcal T_K
=\left\{
\boldsymbol\tau=(\underline\tau,\overline\tau,\tau_2):
-\infty\leq\underline\tau\leq\overline\tau\leq\infty,
\ \tau_2\in\mathbb R
\right\},
$ 
and define
\[
\small 
\begin{aligned}
R_K(\boldsymbol\tau;\Delta,c)
={}&\Delta_+\left[
\mathbb P_\Delta(T_1<\underline\tau)
+\mathbb P_\Delta\left(
\underline\tau\leq T_1<\overline\tau,\ T_2<\tau_2
\right)
\right] \\ &+K\Delta_-\left[
\mathbb P_\Delta(T_1\geq\overline\tau)
+\mathbb P_\Delta\left(
\underline\tau\leq T_1<\overline\tau,\ T_2\geq\tau_2
\right)
\right] +c\mathbb P_\Delta\left(
\underline\tau\leq T_1<\overline\tau
\right).
\end{aligned}
\]
The optimal thresholds can be selected so that
$ 
\boldsymbol\tau^\star
\in\operatorname*{arg\,min}_{\boldsymbol\tau\in\mathcal T_K}
\max_{\Delta\in[-\bar\Delta,\bar\Delta]}
R_K(\boldsymbol\tau;\Delta,c).
$ 

In addition, any minimax-regret optimal rule equals almost surely
$\delta_1^\star$ in the first period and $\delta_2^\star$ on the
event of continuation, under every
$\Delta\in[-\bar\Delta,\bar\Delta]$.
When $c=0$, the thresholds satisfy
$\underline\tau^\star=-\infty$ and
$\overline\tau^\star=\infty$.
\end{theorem}

\begin{proof}
See Appendix \ref{proof:asymmetric_dead_zone}.
\end{proof}

Theorem \ref{thm:asymmetric_dead_zone} preserves the interpretation
of additional experimentation as a response to intermediate
evidence. One can define accordingly and study the $A$-value. 

\subsection{Endogenous precision of the follow-up experiment}
\label{rem:endogenous_precision}

The precision of the follow-up experiment can be endogenized by allowing researchers to choose its sample size. Let $n\in\mathcal{N}$ denote a feasible follow-up sample size and suppose that
$
X_2(n)
\sim
\mathcal{N}\left(\Delta,\frac{\zeta_2^2}{n}\right)
$
at welfare cost $c(n)$. The cost $c(n)$ may include both the direct cost of collecting additional observations and the opportunity cost of delaying the policy decision. Each sample size therefore generates a feasible pair
$
\left(s_2^2(n),c(n)\right)
=
\left(\frac{\zeta_2^2}{n},c(n)\right),
$
capturing the trade-off between the informativeness and the cost of the follow-up experiment. If the follow-up design must be chosen before observing $X_1$, its optimal sample size solves
$$
n_2^\star
\in
\argmin_{n\in\mathcal{N}}
\;
\min_{(\delta_1,\delta_2)\in\mathcal{D}}
\sup_{|\Delta|\leq\bar{\Delta}}
\mathrm{Reg}
\left(
\delta_1,\delta_2;
\Delta,c(n)
\right),
\qquad
s_2^2=\frac{\zeta_2^2}{n}.
$$
Conditional on $n_2^\star$, Lemma \ref{thm:dead_zone} continues to apply directly with
$
\frac{s_2}{s_1} =\frac{\zeta_2}{s_1\sqrt{n_2^\star}},
\frac{c}{s_1} =\frac{c(n_2^\star)}{s_1}.
$
We can think of this problem as in contexts where funding agencies must commit to a given budget before seeing the evidence reported in pilot studies. 

When the sample size $n_2^\star$ is chosen as a function of $X_1$ the problem becomes however more complex, and left to future research.

\subsection{Comparison of $A$-value ranking with alternative procedures} \label{sec:comparisons}

In this section we show that 
the approximation guarantee in Theorem \ref{thm:optimality1}
need not hold for alternative methods even in large $J$ regimes. 
We compare the top $A$-value procedure with the $p$-value-based and
standard-error-only experimentation rules, fixed-common-prior Bayes
rules, and unrestricted conditional likelihood procedure defined
below. For each comparison, we construct admissible sequences with
$q/J\to\rho\in(0,1)$ along which the ratio of worst-case regret to
$R_q^\star$ remains bounded away from one, a standard argument in statistical decision theory \citep[e.g.][]{manski2004}. Under the same
configurations, Theorem \ref{thm:optimality1} implies that the
corresponding ratio for the top $A$-value procedure converges to one.

\paragraph{$p$-value and standard error ranking} \label{sec:p_value}
 We define the $p$-value based procedure below. 

\begin{defn}[$p$-value-based experimentation rule] \label{defn:p_value}
For $\mathbf x\in\mathbb R^J$, write \\
$ 
\mathbf p(\mathbf x)
=
\left(
2\Phi(-|x_1|/s_1),\ldots,
2\Phi(-|x_J|/s_J)
\right).
$ 
A feasible strategy $\boldsymbol\delta$ has a
$p$-value-based experimentation rule if there exists a
measurable function
$ 
g:[0,1]^J\longrightarrow\{0,1\}^J
$ 
such that, for every $\mathbf x\in\mathbb R^J$,
\[
\left(
1\{\delta_{1j}(\mathbf x)=\texttt{NA}\}
\right)_{j=1}^J
=
g(\mathbf p(\mathbf x))
\]
and for every permutation $\pi$ of $\{1,\ldots,J\}$, $g(\pi\mathbf p) = \pi g(\mathbf p)$. 
Thus, the choice of policies receiving follow-up evidence
depends on the observations only through their $p$-values,
and relabeling the $p$-values relabels the distribution of
the experimentation decisions in the same way. These restrictions concern only whether
$\delta_{1j}=\texttt{NA}$. Immediate implementation or
rejection decisions and terminal implementation decisions
may use all information permitted by the strategy's
information constraints.
\end{defn}

The following proposition illustrates lack of asymptotic optimality for procedures that only use the $p$-value or the standard error ranking. 

\begin{prop}[Lower bounds for alternative experimentation rules]
\label{prop:p_value}
Maintain Assumption
\ref{ass:multiple}, with $c_j=1$ and
$\bar\Delta_j=s_j$ for every $j$. Fix $s> 0$ and
$\rho\in(0,1)$, and let $q\in\{1,\ldots,J-1\}$ satisfy
$q/J\rightarrow\rho$. Then there exist constants $\eta_\rho>0$ and
$\varepsilon_\rho\in(0,1)$ independent of $J$ such that
\[
\liminf_{J\rightarrow\infty}
\frac{
\sup_{\boldsymbol\Delta}
\overline{\mathrm{Reg}}(\boldsymbol\delta;\boldsymbol\Delta)
}{
R_q^\star
}
\geq 1+\eta_\rho
\]
in each of the following cases:
\begin{enumerate}
\item (Sub-optimality of standard-error alone based strategies)
$s_j=s$ for every $j$, and
$\boldsymbol\delta\in\mathcal D_q$ is any sequence of strategies
whose set of selected policies is independent of
$\mathbf X_1$.

\item (Sub-optimality of $p$-values based strategies) $s_j = s 1\{j \le q\} + \varepsilon_\rho s 1\{j > q\}$
and $\boldsymbol\delta\in\mathcal D_q$ is any sequence of
strategies with $p$-value-based experimentation rules.
\end{enumerate}
\end{prop}

\begin{proof} See Appendix \ref{proof:p_value}.  
\end{proof} 

Proposition \ref{prop:p_value} shows that we can find configurations of DGP and standard errors so that neither ranking strategies that only use the $p$-value, nor those that only use standard errors achieves the minimax optimal regret $R_q^\star$ even with $J \rightarrow \infty$. This differs from the top $A$-value ranking that is shown to achieve $R_q^\star$ as $J \rightarrow \infty$.

\paragraph{Standard Bayesian procedures} \label{sec:Bayes} Next, we study properties of Bayes procedures.

\begin{defn}[Class of Bayes rules]
\label{defn:bayes_rule}
Let $\Pi$ be a Borel probability measure on $\mathbb R$ satisfying
$\int |\theta|\,\Pi(d\theta)<\infty, \Pi(( -\infty,0))>0,
\Pi((0,\infty))>0$.
Under the working prior, the effects $\Delta_1,\ldots,\Delta_J$
are independent with common distribution $\Pi$, so their joint
distribution is $\Pi^{\otimes J}$.
For a budget $0\leq q\leq\sum_{j=1}^Jc_j$, a strategy
$\boldsymbol\delta^{\Pi,q}\in\mathcal D_q$ is a Bayes rule
under $\Pi$ if
\[
\boldsymbol\delta^{\Pi,q}
\in
\operatorname*{arg\,min}_{\boldsymbol\delta\in\mathcal D_q}
\int_{\mathbb R^J}
\overline{\mathrm{Reg}}(\boldsymbol\delta;\boldsymbol\Delta)
\,\Pi^{\otimes J}(d\boldsymbol\Delta).
\]
\end{defn}

The Bayes procedure uses a prior $\Pi$ over treatment effects to make the optimal regret decision. We establish sub-optimality for minimax regret even in settings with $J \rightarrow \infty$.

\begin{prop}[Sub-optimality under a common prior]
\label{prop:general_prior}
Fix $s_0>0$ and $\rho\in(0,1)$. Maintain Assumption \ref{ass:multiple}, with
$c_j=1$, $\bar\Delta_j=s_j$, and
$ 
s_j=s1\{j\leq q\}+s_0 1\{j>q\}.
$ 
There exist a standard error
$s>s_0$ and a constant $\eta_\rho>0$ independent of $J$ such that the for every sequence $q\in\{1,\ldots,J-1\}$ with $q/J\to\rho$,
$$
\liminf_{J\to\infty}
\frac{\sup_{\boldsymbol\Delta}
\overline{\mathrm{Reg}}(\boldsymbol\delta^{\Pi,q};\boldsymbol\Delta)}
{R_q^\star}
\geq1+\eta_\rho.
$$ 
The choice of $s$ may depend on
$\Pi$, $s_0$, and $\rho$; $\eta_\rho$ can be chosen to depend
only on $\rho$.
\end{prop}

\begin{proof} See Appendix \ref{proof:prop:general_prior}.  
\end{proof}

The intuition is that the least-favorable prior needs not to be the same across different policy decisions $j$ with heteroskedastic errors. The exchangeability assumption breaks this construction and leads to a sub-optimal minimax regret (for the oracle that knows $\Delta$).

\textit{Estimating} heterogeneous
priors from the data as a function of observables (e.g., $s_j$)
introduces a separate estimation problem, requiring smoothness assumptions on how $\Delta_j$ changes with $s_j$ for non-parametric methods to be feasible.\footnote{In particular, typical non-parametric estimators require smoothness assumptions of the prior in $s_j$ that can make rates of convergence sensitive to such smoothness restrictions.}  We provide a simple example. 

\paragraph{Unrestricted conditional likelihood estimator} 

\begin{defn}[Unrestricted conditional likelihood estimator]
\label{defn:nonparametric_eb}
Suppose that $\bar\Delta_j=s_j$ for every $j$. Let $\mathcal G$
be the class of all Borel probability kernels $v\mapsto G_v$,
$v>0$, such that $G_v([-v,v])=1$, with no restrictions linking
the distributions at different standard errors. After observing
$\mathbf X_1$, estimate the conditional prior by
\begin{equation}
\label{eq:nonparametric_eb_fit}
\small 
\begin{aligned} 
\widehat G_\cdot
\in\operatorname*{arg\,max}_{G_\cdot\in\mathcal G}
\sum_{j=1}^J
\log\left\{
\int\varphi_{s_j}(X_{1j}-t)\,G_{s_j}(dt)
\right\},
\end{aligned} 
\end{equation}
where $\varphi_v$ is the $\mathcal N(0,v^2)$ density.
An unrestricted conditional rule
$\widehat{\boldsymbol\delta}^{\,\mathrm{UR},q}\in\mathcal D_q$
solves the remaining Bayesian decision problem, conditional on
$\mathbf X_1$, under the fitted product prior
$\bigotimes_{j=1}^J\widehat G_{s_j}$ and the Gaussian sampling
model. The fitted prior family is held fixed when evaluating
and incorporating follow-up evidence.
\end{defn}

\begin{prop}[Unrestricted conditional likelihood estimator]
\label{prop:nonparametric_eb}
Fix $s>0$ and $\rho\in(0,1)$. Maintain Assumption
\ref{ass:multiple}, with
$ 
c_j=1,  \bar\Delta_j=s_j, 
s_j=s\left(1+\frac{j}{J^2}\right),\quad j=1,\ldots,J.
$ 
For every sequence $q\in\{1,\ldots,J-1\}$ with $q/J\to\rho$
and every sequence of rules $\widehat{\boldsymbol\delta}^{\,\mathrm{UR}}$ in Definition
\ref{defn:nonparametric_eb}, for a constant $\eta_\rho>0$ that only depends on $\rho$ 
\[
\liminf_{J\to\infty}
\frac{\sup_{\boldsymbol\Delta}
\overline{\mathrm{Reg}}(\widehat{\boldsymbol\delta}^{\,\mathrm{UR},q};
\boldsymbol\Delta)}{R_q^\star}
\geq 1+\eta_\rho.
\]
\end{prop}

\begin{proof} See Appendix \ref{sec:empirical_bayes}.  
\end{proof} 

The core intuition is that unrestricted estimators do not impose restrictions on how the prior of $\Delta$ varies with $s$. As a result, they form point masses priors. Because the prior is held fixed using the observed first period results, Bayesian updating leaves these point-mass unchanged making new data collection not beneficial. This intuition follows in spirit to negative results for the different binary treatment choice problem with continuous covariates in \cite{stoye2009}, Section 4.

\subsection{Non-Gaussian experiments}
\label{sec:non_gaussian_experiments}

We now allow unbiased estimates to have non-Gaussian distributions. We maintain
the two-period decision problem and regret criterion, replacing the Gaussian
sampling assumptions with the following condition.

\begin{ass}[Unbiased exponential-family experiments]
\label{ass:non_gaussian_experiments}
Conditional on $\Delta\in[-\bar\Delta,\bar\Delta]$, where
$\bar\Delta\in(0,\infty)$, the estimates $X_1$ and $X_2$ are independent with
densities
\begin{equation}
 f_j(x;\Delta)
 =h_j(x)\exp\!\left\{a_j\big[\eta(\Delta)x-
                 \psi(\eta(\Delta))\big]\right\},\qquad j=1,2,
 \label{eq:non_gaussian_family}
\end{equation}
where $a_j>0$, $h_j$ are positive, even densities on $\mathbb R$, and
$h_2$ is log-concave. The common function $\psi$ is even and twice continuously
differentiable on an open symmetric interval $I$ containing zero, with
$\psi(0)=0$ and $\psi''>0$. For every $\eta\in I$,
$\int_{\mathbb R}e^{a_j\eta x}h_j(x)\,dx=e^{a_j\psi(\eta)}$.
The inverse $\eta(\Delta)=(\psi')^{-1}(\Delta)$ is defined on a neighborhood
of $[-\bar\Delta,\bar\Delta]$. Follow-up evidence is collected at cost $c\ge0$.
\end{ass}

The mean parametrization gives $\mathbb E_\Delta[X_j]=\Delta$. Moreover,
$f_j(x;\Delta')/f_j(x;\Delta)$ is strictly increasing in $x$ whenever
$\Delta'>\Delta$, so the experiment has monotone likelihood ratios.
The common exponential-family form and log-concavity impose additional
structure on how initial and follow-up evidence combine. Write
$F_j(\cdot;\Delta)$ for the corresponding distribution functions.

\begin{theorem}[Non-Gaussian rule]
\label{thm:non_gaussian_rule}
Under Assumption~\ref{ass:non_gaussian_experiments}, there exist minimax-regret
optimal decisions $(\delta_1^\star,\delta_2^\star)$ and a threshold
$\tau^\star\in[0,\infty]$ such that
\[
\small 
\begin{aligned} 
 \delta_1^\star(X_1)=
 \begin{cases}
  1, & X_1\ge\tau^\star,\\
  0, & X_1<-\tau^\star,\\
  \mathrm{NA}, & -\tau^\star\le X_1<\tau^\star,
 \end{cases}
 \qquad
 \delta_2^\star(X_1,X_2)
 =\mathbf 1\{a_1X_1+a_2X_2\ge0\}.
 \end{aligned} 
\]
The threshold solves
\begin{equation}
\small 
\begin{split}
 \tau^\star\in\arg\min_{\tau\in[0,\infty]}
 \max_{0\le\Delta\le\bar\Delta}\Bigg\{
 &\Delta F_1(-\tau;\Delta)
 +\Delta\int_{-\tau}^{\tau}
 F_2\!\left(-\frac{a_1}{a_2}x;\Delta\right)f_1(x;\Delta)\,dx\\
 &+c\big[F_1(\tau;\Delta)-F_1(-\tau;\Delta)\big]\Bigg\}.
\end{split}
\label{eq:non_gaussian_threshold}
\end{equation}
Every minimax-regret optimal rule agrees almost surely with
$\delta_1^\star$ in the first period and with $\delta_2^\star$ on the event
of continuation, under every $\Delta\in[-\bar\Delta,\bar\Delta]$.
\end{theorem}

\noindent\textit{Proof.} See Appendix~\ref{app:proof_non_gaussian_rule}.

The result preserves the threshold representation and uniqueness in Lemma~\ref{thm:dead_zone}.

\begin{rem}[Biased studies] \label{rem:biased_studies}
Suppose $X\mid(\Delta,b)\sim\mathcal{N}(\Delta+b,\sigma_x^2)$,
where $b\sim\mathcal{D}$ has a known distribution with finite mean.
Then $\widetilde X=X-\mathbb{E}_{\mathcal{D}}[b]$ is unbiased for
$\Delta$. If the induced sampling model satisfies the assumptions
of Theorem~\ref{thm:non_gaussian_rule}, its conclusions apply to minimax regret over $\Delta$,
with regret averaged over the known bias distribution $\mathcal{D}$.
\end{rem}

\subsection{Incorporating a prior for the $A$-value}
\label{sec:prior_information}

We now allow researchers to incorporate prior information about policy
effects while controlling the resulting increase in worst-case regret in the spirit of contamination models.
We maintain the two-period experiment in Assumptions~\ref{ass:experiment}
and~\ref{ass:2}, with $w=s_2/s_1$ and welfare cost $c\geq0$.

\begin{ass}[Prior information]
\label{ass:prior_information}
The prior $\pi$ is a Borel probability measure on
$\Theta=[-\bar\Delta,\bar\Delta]$, symmetric around zero, with
$\pi(\{0\})<1$. It is specified before observing $(X_1,X_2)$ and is held
fixed as the cost $c$ varies.
\end{ass}

Assumption~\ref{ass:prior_information} allows both discrete and continuous
priors. Symmetry preserves the sign of the optimal policy recommendation;
prior information affects the decision to collect additional evidence. Let $\delta^\tau$ denote the rule in Lemma~\ref{thm:dead_zone} with continuation threshold $\tau$. 
Define the minimax and prior-optimal thresholds (fixing $w = s_2/s_1, s_1$) by
\[
\small 
\begin{aligned} 
\tau_M(c)=\tau^\star(w,c/s_1),\qquad
 \tau_\pi(c)\in\operatorname*{arg\,min}_{\tau\in[0,\infty]}
       \int_\Theta \mathrm{Reg}(\delta^{\tau};\Delta,c)\,\pi(d\Delta)
       \end{aligned}
\]
where $\tau^\star$ is as defined in Lemma \ref{thm:dead_zone}. 
It follows that  $\delta^{\tau_\pi(c)}$ is Bayes optimal under prior $\pi$
among all admissible strategies \citep{sobel1953} (see the theorem below).

Our goal here is to derive a rule that provides guarantees on both minimax regret and Bayes regret. To do so, fix a weight $\lambda \in [0,1]$. The weight $\lambda$ is user-specific and will serve as a weight between the Bayes and minimax objectives. 
Set
\begin{equation}
\small 
\begin{aligned} 
 \tau_\lambda(c)
 =\Phi^{-1}\!\left((1-\lambda)\Phi(\tau_M(c))
                        +\lambda\Phi(\tau_\pi(c))\right).
 \label{eq:prior_weighted_threshold}
 \end{aligned} 
\end{equation}
Thus the associated two-sided significance level
$2\Phi(-\tau_\lambda(c))$ is the weighted average of the minimax and
Bayesian significance levels. With an abuse of notation, define
$\delta_c^\lambda=\delta^{\tau_\lambda(c)}$, the decision rule that uses the threshold $\tau_\lambda(c)$.

\begin{defn}[$A$-value with prior information]
\label{def:prior_a_value}
For $\lambda\in[0,1]$, define
\begin{equation}
\small 
\begin{aligned} 
 A_{\lambda,\pi}(X_1;s_1,w)
 =\inf\left\{c\geq0:
       \frac{|X_1|}{s_1}\geq\tau_\lambda(c)\right\}.
 \label{eq:prior_a_value}
 \end{aligned} 
\end{equation}
\end{defn}

Write the prior-average and worst-case regrets of this rule as
\[
 \small 
 \begin{aligned} 
B_\lambda(c)=\int_\Theta
       \mathrm{Reg}(\delta_c^\lambda;\Delta,c)\,\pi(d\Delta),
 \qquad
 M_\lambda(c)=\sup_{\Delta\in\Theta}
       \mathrm{Reg}(\delta_c^\lambda;\Delta,c).
       \end{aligned}
\]
Note that $B_1$ corresponds to the Bayes risk under the Bayes optimal decision, $M_0$ corresponds to the minimax risk under the minimax optimal decision, vice-versa $B_0, M_1$. 

\begin{theorem}
\label{thm:prior_a_value}
Under Assumptions~\ref{ass:experiment}, \ref{ass:2},
and~\ref{ass:prior_information}, the following hold.
\begin{itemize}
\item[(A)] The Bayes rule under $\pi$ has the same threshold form as in
Lemma~\ref{thm:dead_zone}, with the unique threshold $\tau_\pi(c)$.
This threshold is continuous and nonincreasing in $c$, equals infinity
at $c=0$, is finite for $c>0$, and converges to zero as $c \rightarrow \infty$.

\item[(B)] For every $\lambda\in[0,1]$ and $c\geq0$,
\[
 \delta_{1,c}^\lambda(X_1)=\texttt{NA}
 \text{ if and only if }
 c<A_{\lambda,\pi}(X_1;s_1,w),
\]
almost surely under every $\Delta\in\Theta$.
The endpoints $\lambda=0$ and $\lambda=1$ give the minimax and Bayesian
switching values, respectively. For every finite $x$, suppressing
$(s_1,w)$,
\[
 \min\{A_{0,\pi}(x),A_{1,\pi}(x)\}
 \leq A_{\lambda,\pi}(x)
 \leq\max\{A_{0,\pi}(x),A_{1,\pi}(x)\}.
\]

\item[(C)] For every $c\geq0$ and $\lambda\in[0,1]$,
$$  
\small 
\begin{aligned}
 M_\lambda(c)-M_0(c)
   &\leq\lambda\{M_1(c)-M_0(c)\},
, \qquad 
B_\lambda(c)-B_1(c)
   &\leq(1-\lambda)\{B_0(c)-B_1(c)\}.
\end{aligned}
$$ 
Moreover, $M_\lambda(c)$ is
nondecreasing and $B_\lambda(c)$ is nonincreasing in $\lambda$.
\end{itemize}
\end{theorem}

\begin{proof}
See Appendix~\ref{proof:prior_a_value}.
\end{proof}

Theorem~\ref{thm:prior_a_value} preserves the interpretation of the
$A$-value as the switching welfare cost of additional evidence.
The weight $\lambda$ determines the use of prior information where
$\lambda=0$ recovers the minimax rule and $\lambda=1$ recovers the Bayes
rule. An intermediate weight attains at least a fraction $\lambda$ of
the improvement in prior-average regret available from the Bayes rule,
while incurring at most a fraction $\lambda$ of its excess worst-case
regret. This construction therefore allows for guarantees on the Bayes and minimax error.

\section{Proofs of main results}

\subsection{Proof of Theorem \ref{thm:multi_period_dead_zone} and  Lemma \ref{thm:dead_zone}}
\label{proof:dead_zone}
We first prove the generalize multi-period version Theorem \ref{thm:multi_period_dead_zone}. We then establish Lemma \ref{thm:dead_zone}. 

\subsubsection{Proof of Theorem \ref{thm:multi_period_dead_zone}}
\label{proof:multi_period_dead_zone}

Recall the notation $S_h^2
=
\sum_{r=1}^h\frac{1}{s_r^2},
T_h
=
\frac{\sum_{r=1}^h\frac{X_r}{s_r^2}}{S_h}$  in Section \ref{sec:multiple_periods}.

\paragraph{Existence and worst-case prior}
The case $\bar\Delta=0$ is immediate, so suppose that
$\bar\Delta>0$. Temporarily enlarge the class of admissible decisions $\boldsymbol{\delta} \in \mathcal D_H$ to include
randomized behavioral strategies. 
Our setting satisfies Assumptions 3.1--3.7 of
\citet{wald1950statistical}. Indeed, the Gaussian experiment is
absolutely continuous and dominated; the parameter space
$[-\bar\Delta,\bar\Delta]$ is compact; the distribution of
$(X_1,\ldots,X_H)$ is continuous in $\Delta$ in total variation; and
regret is bounded and continuous in $\Delta$, uniformly over decision
rules. Moreover, the horizon and action sets are finite, and
the class of randomized strategies is convex and closed.

Therefore, \citet[Theorem 3.4]{wald1950statistical} implies that
\begin{equation} \label{eqn:Bayes_minimax}
\small 
\begin{aligned}
V
&:=
\inf_{\boldsymbol\delta\in\mathcal D_H}
\sup_{|\Delta|\leq\bar\Delta}
\mathrm{Reg}_H
(\boldsymbol\delta;\Delta,\boldsymbol c) =
\sup_{\pi\in\mathcal P([-\bar\Delta,\bar\Delta])}
\inf_{\boldsymbol\delta\in\mathcal D_H}
\int
\mathrm{Reg}_H
(\boldsymbol\delta;\Delta,\boldsymbol c)
\,\pi(d\Delta).
\end{aligned}
\end{equation} 
By \citet[Theorems 3.7 and 3.14]{wald1950statistical}, respectively,
there exist a minimax strategy and a
least-favorable prior $\pi^\star$. In particular,  
\[
\small 
\begin{aligned} 
b(\pi)
:=
\inf_{\boldsymbol\delta\in\mathcal D_H}
\int
\mathrm{Reg}_H
(\boldsymbol\delta;\Delta,\boldsymbol c)
\,\pi(d\Delta), \qquad b(\pi^\star)=V, \qquad \pi^\star
\in
\operatorname*{arg\,max}_{\pi\in
\mathcal P([-\bar\Delta,\bar\Delta])}
b(\pi). 
\end{aligned} 
\]
\paragraph{Symmetric worst-case prior}
We may take the least-favorable prior to be symmetric. For any prior
$\pi$, define its reflection around zero by
\[
\small 
\begin{aligned} 
\pi^-(A)=\pi(-A),
\qquad
-A=\{-\Delta:\Delta\in A\},
\end{aligned} 
\]
for every Borel set
$A\subseteq[-\bar\Delta,\bar\Delta]$. For any strategy
$\boldsymbol\delta$, define its reflected strategy by reversing the
signs of all observed signals, interchanging terminal actions $0$ and
$1$, and leaving continuation unchanged. Formally, for $h\leq H$,
\[
\small 
\begin{aligned} 
\delta_h^-(x_1,\ldots,x_h)
=
\begin{cases}
1-\delta_h(-x_1,\ldots,-x_h),
&
\delta_h(-x_1,\ldots,-x_h)\in\{0,1\},\\
\texttt{NA},
&
\delta_h(-x_1,\ldots,-x_h)=\texttt{NA}.
\end{cases}
\end{aligned} 
\]
Because the distribution of
$(-X_1,\ldots,-X_H)$ under $\Delta$ is the same as the distribution
of $(X_1,\ldots,X_H)$ under $-\Delta$, while the continuation costs
are unchanged,
$ 
\mathrm{Reg}_H
(\boldsymbol\delta^-;\Delta,\boldsymbol c)
=
\mathrm{Reg}_H
(\boldsymbol\delta;-\Delta,\boldsymbol c).
$ 
It follows that
$ 
b(\pi^-)=b(\pi).
$ 

The function $b(\pi)$ is concave in $\pi$, because it is the infimum
over strategies of functions that are linear in $\pi$. Hence, if
$\pi^\star$ is least favorable, then
$ 
b\left(
\frac{\pi^\star+(\pi^\star)^-}{2}
\right)
\geq
\frac{1}{2}b(\pi^\star)
+
\frac{1}{2}b((\pi^\star)^-)
=
V.
$ 
Since $V$ is the largest attainable Bayes risk, equality must hold.
Thus, the symmetrized prior
$ 
\frac{\pi^\star+(\pi^\star)^-}{2}
$
is also least favorable, and we may take $\pi^\star$ to be symmetric.

\paragraph{Nondegeneracy}
The symmetric least-favorable prior is nondegenerate. To establish
this, fix $a\in(0,\bar\Delta]$ and consider the symmetric two-point
prior $\pi(a) = \pi(-a) = 1/2$. 
Consider the more informative experiment in which all $H$ signals are
observed at zero cost. Under $\Delta=\pm a$,
$ 
T_H\sim\mathcal N(\pm aS_H,1).
$ 
The likelihood ratio between $\Delta=a$ and $\Delta=-a$ is increasing
in $T_H$, so the Bayes terminal decision under $\pi_a$ implements the
policy if and only if $T_H\geq0$. Its probability of error under
either parameter value is $\Phi(-aS_H)$. Since an incorrect terminal
decision has regret $a$, the Bayes regret in this more informative,
costless experiment is
$ 
a\Phi(-aS_H)>0.
$ 
The original sequential problem cannot have smaller Bayes regret.
Therefore,
$ 
b(\pi_a)
\geq
a\Phi(-aS_H)
>0,
$ 
and consequently
$ 
V
=
\sup_{\pi}b(\pi)
\geq
b(\pi_a)
>0.
$ 

By contrast, if $\delta_0$ denotes the point mass at $\Delta=0$, then
$b(\delta_0)=0$: either immediate terminal action has zero regret, and
the planner can avoid all continuation costs by stopping in the first
period. Hence a least-favorable prior cannot equal $\delta_0$. Because
$\pi^\star$ is symmetric, it follows that $\pi^\star$ is
nondegenerate.

\paragraph{Bayes characterization of the minimax strategy}
Taking $\boldsymbol\delta^M$ to be a minimax strategy and $\pi^\star$ is
least favorable,
$ 
V
=
b(\pi^\star)
\leq
\int
\mathrm{Reg}_H
(\boldsymbol\delta^M;\Delta,\boldsymbol c)
\,\pi^\star(d\Delta) \leq
\sup_{|\Delta|\leq\bar\Delta}
\mathrm{Reg}_H
(\boldsymbol\delta^M;\Delta,\boldsymbol c)
=
V.
$ 
Both inequalities must therefore hold with equality. In particular,
$\boldsymbol\delta^M$ attains the minimum Bayes risk under
$\pi^\star$. Hence the minimax strategy is Bayes under a symmetric,
nondegenerate least-favorable prior. It remains to characterize its
Bayes form.

\paragraph{Posterior policy decision}
Define
$ 
Y_h
=
\sum_{r=1}^h\frac{X_r}{s_r^2}
=
S_hT_h.
$ 
The likelihood of $(X_1,\ldots,X_h)$ satisfies
\[
\small 
\begin{aligned}
\prod_{r=1}^h
\exp\left\{
-\frac{(X_r-\Delta)^2}{2s_r^2}
\right\}
&\propto
\exp\left\{
\Delta\sum_{r=1}^h\frac{X_r}{s_r^2}
-
\frac{\Delta^2}{2}
\sum_{r=1}^h\frac{1}{s_r^2}
\right\} =
\exp\left\{
Y_h\Delta-\frac{S_h^2}{2}\Delta^2
\right\}.
\end{aligned}
\]
Hence, $Y_h$ is sufficient for $\Delta$. Define
\begin{equation} 
\label{eqn:M_h}
\small 
\begin{aligned} 
M_h(y)
=
\int_{-\bar\Delta}^{\bar\Delta}
\exp\left\{
y\Delta-\frac{S_h^2}{2}\Delta^2
\right\}
\pi^\star(d\Delta).
\end{aligned} 
\end{equation} 
The posterior distribution after period $h$ is
$ 
\pi^\star(d\Delta\mid Y_h=y)
=
\frac{
\exp\left\{
y\Delta-\frac{S_h^2}{2}\Delta^2
\right\}
}{
M_h(y)
}
\pi^\star(d\Delta).
$ 
Its posterior mean is
\[
\small 
\begin{aligned} 
\mu_h(y)
=
\mathbb E_{\pi^\star}[\Delta\mid Y_h=y]
=
\frac{M_h'(y)}{M_h(y)}.
\end{aligned} 
\]

Symmetry of $\pi^\star$ implies that $M_h$ is even. Therefore,
$M_h'$ and $\mu_h$ are odd, and in particular $\mu_h(0)=0$.
Moreover, because $\Delta$ has bounded support, differentiation under
the integral is justified and
\[
\small 
\begin{aligned}
\mu_h'(y)
&=
\frac{M_h''(y)}{M_h(y)}
-
\left(
\frac{M_h'(y)}{M_h(y)}
\right)^2 =
\mathbb E_{\pi^\star}
[\Delta^2\mid Y_h=y]
-
\left(
\mathbb E_{\pi^\star}
[\Delta\mid Y_h=y]
\right)^2 =
\operatorname{Var}_{\pi^\star}
(\Delta\mid Y_h=y)
>0.
\end{aligned}
\]
The final inequality follows because the posterior is obtained from
the nondegenerate prior $\pi^\star$ through a strictly positive
exponential tilt and therefore has the same support. Thus, $\mu_h$ is
strictly increasing. Together with oddness, this gives
$ 
\operatorname{sign}\{\mu_h(y)\}
=
\operatorname{sign}(y)
=
\operatorname{sign}(T_h).
$ 

If the planner stops in period $h$, the posterior regret from rejecting
the policy is
$ 
\mathbb E_{\pi^\star}[\Delta_+\mid Y_h=y],
$ 
whereas the posterior regret from implementing it is
$ 
\mathbb E_{\pi^\star}[\Delta_-\mid Y_h=y].
$ 
Their difference is
$ 
\mathbb E_{\pi^\star}
[\Delta_+-\Delta_-\mid Y_h=y]
=
\mathbb E_{\pi^\star}[\Delta\mid Y_h=y]
=
\mu_h(y).
$ 
Consequently, the Bayes terminal action is
$ 
1\{\mu_h(Y_h)\geq0\}
=
1\{Y_h\geq0\}
=
1\{T_h\geq0\}.
$ 
At $Y_h=0$, the two actions have the same posterior regret, and the
display adopts implementation as the tie-breaking convention.(Since
$Y_h$ has a continuous distribution under every $\Delta$, this
convention does not affect frequentist regret.)
\paragraph{Continuation decision}
Let
\[
\small 
\begin{aligned} 
r_h(y)
=
\min\left\{
\mathbb E_{\pi^\star}[\Delta_+\mid Y_h=y],
\mathbb E_{\pi^\star}[\Delta_-\mid Y_h=y]
\right\}
\end{aligned} 
\]
denote the posterior regret from stopping in period $h$. Excluding
costs already incurred, define the Bayes value functions recursively
by
\[
\small 
\begin{aligned} 
V_H(y)=r_H(y), \qquad V_h(y)
=
\min\left\{
r_h(y),\,
c_h+
\mathbb E_{\pi^\star}
\left[
V_{h+1}(Y_{h+1})
\mid Y_h=y
\right]
\right\}, \quad  h < H.
\end{aligned} 
\]
Continuation is optimal in period $h$ if and only if
$A_h(y)>0$, where
\[
\small 
\begin{aligned} 
A_h(y)
=
r_h(y)-c_h-
\mathbb E_{\pi^\star}
\left[
V_{h+1}(Y_{h+1})
\mid Y_h=y
\right].
\end{aligned} 
\]
We show by backward induction that $A_h$ is even and strictly
decreasing in $|y|$.

\paragraph{Posterior-predictive ordering}
Recall the expression of $M_h(y)$ in Equation \eqref{eqn:M_h}. 
Conditional on $\Delta$,
\[
Y_{h+1}-Y_h
=
\frac{X_{h+1}}{s_{h+1}^2}
\sim
\mathcal N\left(
\frac{\Delta}{s_{h+1}^2},
\frac{1}{s_{h+1}^2}
\right).
\]
Let $K_h(y,y')$ denote the posterior-predictive density of
$Y_{h+1}=y'$ conditional on $Y_h=y$. Integrating the preceding
Gaussian density with respect to the posterior distribution of
$\Delta$ gives
$ 
K_h(y,y')
=
\phi_{h+1}(y'-y)
\frac{M_{h+1}(y')}{M_h(y)},
$ 
where $\phi_{h+1}$ is the density of
$\mathcal N(0,s_{h+1}^{-2})$. For $y_2>y_1$,
\[
\small 
\begin{aligned}
\frac{K_h(y_2,y')}{K_h(y_1,y')}
=
\frac{M_h(y_1)}{M_h(y_2)}
\exp\left\{
-\frac{s_{h+1}^2}{2}(y_2^2-y_1^2)
+
s_{h+1}^2(y_2-y_1)y'
\right\}.
\end{aligned}
\]
This likelihood ratio is strictly increasing in $y'$. Therefore,
$Y_{h+1}\mid Y_h=y$ is strictly stochastically increasing in $y$. In addition, symmetry of $\pi^\star$ implies that $M_h$ is even and
$ 
K_h(-y,-y')=K_h(y,y').
$ 
For $x\geq0$ and $y\geq0$, the conditional density of
$|Y_{h+1}|=x$ is
\[
\small 
\begin{aligned} 
\widetilde K_h(y,x)
=
K_h(y,x)+K_h(y,-x) \propto \frac{M_{h+1}(x)}{M_h(y)}
\exp\left\{
-\frac{s_{h+1}^2}{2}(x^2+y^2)
\right\}
\cosh(s_{h+1}^2xy),
\end{aligned} 
\]
where $\cosh(z) = e^z/2 + e^{-z}/2$. 
Hence, if $0\leq y_1<y_2$,
$ 
\frac{\widetilde K_h(y_2,x)}
     {\widetilde K_h(y_1,x)}
=
\frac{M_h(y_1)}{M_h(y_2)}
\exp\left\{
-\frac{s_{h+1}^2}{2}(y_2^2-y_1^2)
\right\}
\frac{
\cosh(s_{h+1}^2xy_2)
}{
\cosh(s_{h+1}^2xy_1)
}.
$ 
The right-hand side is increasing in $x\geq0$. Thus,
$|Y_{h+1}|\mid Y_h=y$ is stochastically increasing in
$y\geq0$.

\paragraph{Backward induction}
Set $A_H\equiv0$. Suppose that $A_{h+1}$ is even and nonincreasing in
$|y|$, which holds for $h+1=H$. From the definitions of $V_{h+1}$ and
$A_{h+1}$,
$ 
V_{h+1}(y)
=
r_{h+1}(y)-(A_{h+1}(y))_+.
$ 
Therefore,
\begin{equation}
\label{eqn:a7}
\small 
\begin{aligned}
A_h(y)
={}&
r_h(y)-c_h-
\mathbb E_{\pi^\star}
\left[
r_{h+1}(Y_{h+1})
\mid Y_h=y
\right] +
\mathbb E_{\pi^\star}
\left[
(A_{h+1}(Y_{h+1}))_+
\mid Y_h=y
\right].
\end{aligned}
\end{equation}
Posterior symmetry implies that $r_h$ is even. Together with
$K_h(-y,-y')=K_h(y,y')$ and the induction hypothesis, this shows that
$A_h$ is even.

It remains to establish monotonicity for $y\geq0$. For every $y$,
$ 
r_h(y)
=
\frac{1}{2}
\left\{
\mathbb E_{\pi^\star}
[|\Delta|\mid Y_h=y]
-
|\mu_h(y)|
\right\}.
$ 
When $y\geq0$, $\mu_h(y)\geq0$. By the law of iterated expectation applied to the posterior mean under $\pi^\star$, which implies 
$ 
\mathbb E_{\pi^\star}
\left[
\mu_{h+1}(Y_{h+1})
\mid Y_h=y
\right]
=
\mu_h(y)
$, 
we have 
\begin{equation}
\label{eqn:a8}
\begin{aligned}
&r_h(y)-
\mathbb E_{\pi^\star}
\left[
r_{h+1}(Y_{h+1})
\mid Y_h=y
\right] =
\mathbb E_{\pi^\star}
\left[
\bigl(-\mu_{h+1}(Y_{h+1})\bigr)_+
\mid Y_h=y
\right].
\end{aligned}
\end{equation}
Because $\mu_{h+1}$ is strictly increasing,
$y'\mapsto(-\mu_{h+1}(y'))_+$ is nonconstant and nonincreasing.
The strict stochastic ordering of
$Y_{h+1}\mid Y_h=y$ therefore implies that the sum of the first three
terms in \eqref{eqn:a7} is strictly decreasing in $y\geq0$.

By the induction hypothesis,
$y'\mapsto(A_{h+1}(y'))_+$ is even and nonincreasing in $|y'|$.
The stochastic ordering of
$|Y_{h+1}|\mid Y_h=y$ therefore implies that the final term in
\eqref{eqn:a7} is nonincreasing in $y\geq0$. Consequently, $A_h$ is
even and strictly decreasing in $|y|$, completing the induction.

\paragraph{Threshold representation}
The functions $r_h$, $V_h$, and $A_h$ are continuous by backward
induction. Because $A_h$ is even and strictly decreasing in $|y|$,
there exists $b_h^\star\in[0,\infty]$ such that 
$ 
A_h(y)>0$ if and only if $ 
|y|<b_h^\star.$ 
Set $b_H^\star=0$ and define
$ 
\tau_h^\star
=
\frac{b_h^\star}{S_h},
h=1,\ldots,H.
$ 
Since $Y_h=S_hT_h$, the Bayes rule continues in period $h<H$ if and
only if
$ 
|T_h|<\tau_h^\star.
$ 
Upon stopping, the posterior policy decision derived above implements
the policy if and only if $T_h\geq0$. The Bayes action is unique except when $Y_h=0$ or
$|Y_h|=b_h^\star$. Each of these events has probability zero under
every $\Delta$. Since the minimax strategy
$\boldsymbol\delta^M$ is Bayes under $\pi^\star$, it agrees almost
surely, under every $\Delta$, with $\boldsymbol{\delta}^\star$. 

\paragraph{Characterization of the optimal thresholds}
For a symmetric threshold vector $\boldsymbol\tau$ and $\Delta\geq0$,
regret is incurred if the terminal statistic is negative and whenever
an additional signal is collected. That is, recalling the notation in Theorem \ref{thm:multi_period_dead_zone} for $\sigma(\boldsymbol{\tau})$
\[
\small 
\begin{aligned} 
\mathrm{Reg}_H
(\boldsymbol\delta^{\boldsymbol\tau};
 \Delta,\boldsymbol c)
=
\Delta\,
\mathbb P_\Delta
\left(
T_{\sigma(\boldsymbol\tau)}<0
\right)
+
\sum_{h=1}^{H-1}
c_h\,
\mathbb P_\Delta
\left(
\sigma(\boldsymbol\tau)>h
\right)
=
R_H(\boldsymbol\tau;\Delta,\boldsymbol c).
\end{aligned} 
\]
Reflection symmetry implies that the worst-case regret over
$[-\bar\Delta,\bar\Delta]$ equals the maximum of this expression over
$[0,\bar\Delta]$. Since
$\boldsymbol\delta^{\star}$ is minimax,
$ 
\boldsymbol\tau^\star
\in
\operatorname*{arg\,min}_{\boldsymbol\tau\in\mathcal T_H}
\max_{\Delta\in[0,\bar\Delta]}
R_H(\boldsymbol\tau;\Delta,\boldsymbol c).
$ 

\paragraph{Uniqueness} Finally, we prove uniqueness. Consider any
minimax-regret optimal rule $\boldsymbol{\delta}'$. Then
\[
\small 
\begin{aligned}
V
&=
\inf_{\boldsymbol{\delta}\in\mathcal D_H}
\int
\mathrm{Reg}_H(\boldsymbol{\delta};\Delta,\boldsymbol c)
\,\pi^\star(d\Delta) \leq
\int
\mathrm{Reg}_H(\boldsymbol{\delta}';\Delta,\boldsymbol c)
\,\pi^\star(d\Delta) \leq
\sup_{|\Delta|\leq\bar\Delta}
\mathrm{Reg}_H(\boldsymbol{\delta}';\Delta,\boldsymbol c)
=V.
\end{aligned}
\]
Thus, every minimax-regret optimal rule is Bayes under the same
prior $\pi^\star$.

The preceding arguments establish that the posterior mean
$\mu_h(y)$ is odd and strictly increasing. Hence the optimal
terminal action is unique except when $Y_h=0$. For $h<H$, the
continuation advantage $A_h(y)$ is even and strictly decreasing
in $|y|$. Consequently, stopping and continuation can tie only
at most two values of $Y_h$. These events have
probability zero under the prior-predictive distribution.

By finite-horizon dynamic programming, a Bayes rule must choose
a minimizing action at almost every history at which it has not
already stopped. Otherwise, its integrated regret would exceed
the minimum Bayes risk. Since the minimizing action is unique
outside the zero-probability events,
$\delta_1'$ agrees almost surely with $\delta_1^\star$.
Proceeding inductively, agreement through period $h-1$ implies
that the two rules continue to period $h$ on the same event,
up to a null set. On that event, uniqueness of the Bayes action
implies agreement between $\delta_h'$ and $\delta_h^\star$.
Under the prior-predictive distribution,
$ 
\mathbb P_{\pi^\star}\!\left(
\sigma(\boldsymbol{\tau}^\star)\geq h,\;
\delta_h'(X_1,\ldots,X_h)
\neq
\delta_h^\star(X_1,\ldots,X_h)
\right)=0
$ 
for every $h\leq H$. Finally, the prior-predictive density of $(X_1,\ldots,X_h)$ is
$ 
\int
\prod_{r=1}^h
\frac{1}{s_r}
\phi\!\left(\frac{x_r-\Delta}{s_r}\right)
\,\pi^\star(d\Delta)>0
$ 
at every $(x_1,\ldots,x_h)\in\mathbb R^h$.
Thus, the exceptional histories form a Lebesgue-null set.
Since the Gaussian distribution under every $\Delta$ is
absolutely continuous, the same agreement holds
$\mathbb P_\Delta$-almost surely for every
$\Delta\in[-\bar\Delta,\bar\Delta]$.

\subsubsection{Lemma \ref{thm:dead_zone}}
Lemma \ref{thm:dead_zone} is a corollary of 
the more general Theorem \ref{thm:multi_period_dead_zone} with $H=2$ and
$c_1=c$. Under this specialization, using the notation of Theorem \ref{thm:multi_period_dead_zone}, write 
$ 
S_1^2=\frac{1}{s_1^2},
T_1
=
\frac{X_1/s_1^2}{S_1}
=
\frac{X_1}{s_1},
S_2^2
=
\frac{1}{s_1^2}+\frac{1}{s_2^2},
T_2
=
\frac{X_1/s_1^2+X_2/s_2^2}{S_2}.
$ 
Because $S_2>0$,
$ 
1\{T_2\geq0\}
=
1\left\{
\frac{X_1}{s_1^2}
+
\frac{X_2}{s_2^2}
\geq0
\right\}
=
1\left\{
\frac{s_2^2}{s_1^2}X_1+X_2\geq0
\right\}.
$ 
Moreover, the loss and regret in the $H=2$ problem ($\mathrm{Reg}_2(\cdot)$ in Equation \eqref{eqn:regret2}) coincide with
$\mathrm{Reg}(\delta_1,\delta_2;\Delta,c)$ in Equation \eqref{eqn:reg1}. Theorem
\ref{thm:multi_period_dead_zone} therefore implies that there exists a
minimax-regret rule of the form
\[
\small 
\begin{aligned} 
\delta_1^\star(X_1)
=
\begin{cases}
1,
& X_1/s_1\geq\tau^\star,\\
0,
& X_1/s_1\leq-\tau^\star,\\
\texttt{NA},
& |X_1|/s_1<\tau^\star,
\end{cases}, \qquad \delta_2^\star(X_1,X_2)
=
1\left\{
\frac{s_2^2}{s_1^2}X_1+X_2\geq0
\right\}
\end{aligned} 
\]
where $\tau^\star=\tau_1^\star$. Next, we specialize the threshold characterization in
Theorem \ref{thm:multi_period_dead_zone}. Let
$ 
d=\frac{\Delta}{s_1},
w=\frac{s_2}{s_1}.
$ 
By symmetry (see the proof of Theorem \ref{thm:multi_period_dead_zone}), it is sufficient to consider $\Delta\geq0$. Under a
general threshold $\tau$, an incorrect first-period decision occurs
when $X_1/s_1\leq-\tau$, with probability
$ 
\Phi(-\tau-d).
$ 
Conditional on $X_1/s_1=x$, continuation followed by an incorrect
terminal decision occurs when
$ 
\frac{s_2^2}{s_1^2}X_1+X_2<0.
$ 
Since $X_2\sim\mathcal N(s_1d,s_2^2)$ independently of $X_1$, its
conditional probability is
$ 
\Phi\left(-wx-\frac{d}{w}\right).
$ 
Finally, the probability of continuation is
$ 
\mathbb P_\Delta
\left(
-\tau<\frac{X_1}{s_1}<\tau
\right)
=
\Phi(\tau-d)-\Phi(-\tau-d).
$ 
Since $X_1/s_1\sim\mathcal N(d,1)$, the $H=2$ regret function in
Theorem \ref{thm:multi_period_dead_zone} satisfies
\[
\small 
\begin{aligned}
\frac{1}{s_1}
R_2(\tau;\Delta,c)
={}&
d\Phi(-\tau-d)
+
d\int_{-\tau}^{\tau}
\Phi\left(-wx-\frac{d}{w}\right)
\phi(x-d)\,dx +
\frac{c}{s_1}
\left[
\Phi(\tau-d)-\Phi(-\tau-d)
\right].
\end{aligned}
\]
For $\Delta<0$, symmetry gives the same expression with
$d=|\Delta|/s_1$. Therefore, the threshold characterization in
Theorem \ref{thm:multi_period_dead_zone} becomes
\[
\small 
\begin{aligned}
\tau^\star
\in
\operatorname*{arg\,min}_{\tau\in[0,\infty]}
\max_{0\leq d\leq\bar\Delta/s_1}
\Bigg\{
&
d\Phi(-\tau-d)
+
d\int_{-\tau}^{\tau}
\Phi\left(-wx-\frac{d}{w}\right)
\phi(x-d)\,dx +
\frac{c}{s_1}
\left[
\Phi(\tau-d)-\Phi(-\tau-d)
\right]
\Bigg\}.
\end{aligned}
\]
To drop the dependence of the optimization problem over $\bar{\Delta}/s_1$ and write $d \ge 0$ the reader may refer to ``Proof of Part (A)'' in the proof of Lemma \ref{cor:1} below.  

Finally, Theorem \ref{thm:multi_period_dead_zone} also establishes the minimax rule is unique almost surely. 

\subsection{Proof of Lemma \ref{cor:1}} \label{proof:cor:1}
Throughout the proof we define 
$
d=\frac{|\Delta|}{s_1},
w=\frac{s_2}{s_1},
\gamma=\frac{c}{s_1}.
$
We prove each part separately. 

\paragraph{Proof of Part (A)}
Fix a threshold \(\tau\geq0\). By symmetry, it is sufficient to consider
\(\Delta\geq0\). 

For \(\Delta=s_1d\geq0\), let \(p(d,\tau)\) denote the probability that the policy is ultimately rejected under the first-period rule with threshold \(\tau\) and the optimal terminal rule \(\delta_2^\star\). Formally,
\begin{equation} \label{eqn:pd}
\small 
\begin{aligned}
p(d,\tau)
=&
\mathbb P_{\Delta=s_1d}
\left(
\frac{X_1}{s_1}\leq-\tau
\right) +
\mathbb P_{\Delta=s_1d}
\left(
-\tau<\frac{X_1}{s_1}<\tau,
\ \frac{s_2^2}{s_1^2}X_1+X_2<0
\right).
\end{aligned}
\end{equation} 
We suppress the dependence of \(p(d,\tau)\) on \(w=s_2/s_1\) to simplify notation. The normalized
regret can then be written for $\Delta = s_1 d \ge 0$ as
\begin{equation} \label{eqn:norm_regret}
\small 
\begin{aligned} 
\frac{1}{s_1}
\mathrm{Reg}
(\delta_1^\tau,\delta_2^\star;\Delta,c)
=
d\,p(d,\tau)+\gamma q(d,\tau), \qquad q(d,\tau)
=
\Phi(\tau-d)-\Phi(-\tau-d)
\end{aligned}
\end{equation} 
where $q(d,\tau)$ denotes the probability of follow up data collection. 
We show that both terms on the right-hand side are nonincreasing for
\(d\geq u_0\), where \(u_0\) is the unique positive solution to
$
\Phi(-u_0)=u_0\phi(u_0).
$
Numerically, \(u_0\simeq0.7518\).
First, note that the second term is non-increasing in $d$
$$
\small 
\begin{aligned} 
\frac{\partial q(d,\tau)}{\partial d}
=
\phi(\tau+d)-\phi(\tau-d)
\leq0,
\qquad d,\tau\geq0.
\end{aligned} 
$$
We are left to study the first term. 
The core idea is to represent $p(d,\tau)$ as a mixture of Gaussians. 
To do so, we first introduce additional notation that will simplify the expressions 
\begin{equation} \label{eqn:VV}
\small 
\begin{aligned} 
Z_1=\frac{X_1}{s_1},
\qquad
Z_2=\frac{X_2}{s_2} \Rightarrow 
Z_1\sim\mathcal N(d,1), \qquad V=\frac{wZ_1+Z_2}{\sqrt{1+w^2}},
\qquad
a=\frac{\sqrt{1+w^2}}{w}, \qquad \Gamma =\tau\sqrt{1+w^2}.
\end{aligned} 
\end{equation} 
The terminal decision rejects the policy if
\(wZ_1+Z_2<0\).  
Also, \(V\sim\mathcal N(ad,1)\), and it is easy to show that 
$
Z_1\mid V=v
\sim
\mathcal N\left(
\frac{w}{\sqrt{1+w^2}}v,
\frac{1}{1+w^2}
\right).
$\footnote{The pair $(Z_1,V)$ is jointly Gaussian with means $(d,ad)$,
unit variances, and covariance $\rho=w/\sqrt{1+w^2}$.
Gaussian conditioning gives mean $d+\rho(v-ad)=\rho v$, since
$\rho a=1$, and variance $1-\rho^2=1/(1+w^2)$.}
Importantly, this conditional distribution does not depend on \(d\).
Conditional on $V=v$, the probability that the policy is ultimately rejected is
\begin{equation} \label{eqn:representationa}
\small 
\begin{aligned} 
\mathbb P_{\Delta=s_1d}\left(
\left.
\left\{\frac{X_1}{s_1}\leq-\tau\right\}
\cup
\left\{
-\tau<\frac{X_1}{s_1}<\tau,\,
\frac{s_2^2}{s_1^2}X_1+X_2<0
\right\}
\,\right|\,V=v
\right)
=
\begin{cases}
\Phi(\Gamma-wv), & v<0,\\[3pt]
\Phi(-\Gamma-wv), & v>0.
\end{cases}
\end{aligned} 
\end{equation} 
To see why this is true, when \(v<0\), the second-period decision rejects the policy, so the
policy is ultimately rejected unless \(Z_1\geq\tau\) (so that we only use information from the first period). It follows that $P_{\Delta = s_1 d}(Z_1 \le \tau|V = v) = \Phi(\Gamma - wv)$. When \(v>0\), the
second-period decision implements the policy, so rejection occurs only when
\(Z_1\leq-\tau\), in which case $P_{\Delta = s_1 d}(Z_1 \leq -\tau|V = v) = \Phi(-\Gamma - wv)$.

Consider a symmetric random variable \(U\), independent of \(V\), with
$
P(U = 0) = 2 \Phi(\Gamma) - 1
$ 
and density
$
f_U(u)=w\phi(\Gamma+w|u|),
u\neq0.
$
By direct integration,
\[
\small 
\begin{aligned} 
\mathbb P(U>V|V = v)=
\begin{cases}
\Phi(\Gamma-wv), & v<0,\\
\Phi(-\Gamma-wv), & v>0.
\end{cases}
\end{aligned} 
\]
These probabilities coincide with the conditional rejection
probabilities in Equation \eqref{eqn:representationa}.
Using Equation \eqref{eqn:representationa} we can therefore write 
\begin{equation} \label{eqn:representationb}
\small 
\begin{aligned} 
p(d,\tau)
=
\mathbb P(U>V)
=
\mathbb P(U-\varepsilon>ad), \qquad \varepsilon|U\sim\mathcal N(0,1) 
\end{aligned} 
\end{equation} 
where for the latter equality we used the fact that $V \sim \mathcal{N}(ad, 1)$. 
For \(\Gamma>0\), the continuous part of the density of \(U\) satisfies
$
f_U(u)
=
w\phi(\Gamma)
\exp\left\{
-\Gamma w|u|-\frac{w^2u^2}{2}
\right\}.
$
We will use the identity 
\begin{equation} \label{eqn:exp}
\small 
\begin{aligned} 
e^{-b|u|}
=
\frac{b}{2\sqrt{\pi}}
\int_0^\infty
t^{-3/2}
\exp\left\{-\frac{b^2}{4t}\right\}
\exp\{-tu^2\}\,dt,
\qquad b>0,
\end{aligned} 
\end{equation} 
to represent \(U\) as a mixture of Gaussian random variables with mean zero and 
whose variances are between $0$ and 
$
\frac{1}{w^2},
$
where the atom at zero corresponds to zero variance. 
That is, for $w>0$ and $\Gamma>0$, let $H(\cdot;\tau,w)$ be the probability measure on
$[0,1/w^2]$ defined, for every Borel set $B\subseteq[0,1/w^2]$, by (using Equation \eqref{eqn:exp})\footnote{Apply the identity with $b=\Gamma w$ and set
$s=(w^2+2t)^{-1}$, so $t=(1-w^2s)/(2s)$ and
$|dt/ds|=1/(2s^2)$. Writing
$e^{-u^2/(2s)}=\sqrt{2\pi s}\,\varphi_s(u)$, where
$\varphi_s$ is the $\mathcal N(0,s)$ density, and collecting
factors gives the stated mixing density. Its total mass is
$2\Phi(-\Gamma)$; the remaining mass $2\Phi(\Gamma)-1$
at $s=0$ reproduces the atom of $U$.} 
\[
\begin{aligned}
H(B;\tau,w)
={}&
\bigl(2\Phi(\Gamma)-1\bigr)1\{0\in B\} +
\int_{B\cap(0,\,1/w^2)}
\frac{\Gamma w^2\phi(\Gamma)}
     {(1-w^2s)^{3/2}}
\exp\left\{
-\frac{\Gamma^2w^2s}{2(1-w^2s)}
\right\}ds.
\end{aligned}
\]
From Equation \eqref{eqn:representationb} and using Equation \eqref{eqn:exp}, it follow  for $\Gamma > 0$
\[
\small 
\begin{aligned} 
p(d,\tau)
=
\int_{[0,\,1/w^2]}
\Phi\left(
-\frac{ad}{\sqrt{1+s}}
\right)
\,dH(s;\tau,w).
\end{aligned} 
\]
When \(\Gamma=0\),  the Gaussian representaton follows directly since 
\(U\sim\mathcal N(0,1/w^2)\) and so $U - \varepsilon \sim \mathcal{N}(0, 1 + 1/w^2)$. Accordingly, define $H(\cdot;0,w)$ to place probability one
at $s=1/w^2$.
We therefore represent both cases using a unified notation with the following change of variable. Let 
 for every $s \in[0,1/w^2]$,
\begin{equation} \label{eqn:lambda_s}
\small 
\begin{aligned} 
1
\leq
\lambda(s)
\equiv
\frac{a}{\sqrt{1+s}}
\leq a, \qquad a^2=1+1/w^2. 
\end{aligned} 
\end{equation} 
We write for $\Gamma \ge 0$ (which covers all the relevant cases since $\tau \ge 0$) 
\begin{equation} \label{eqn:mixture_a}
\small 
\begin{aligned}
p(d,\tau)
=
\int_{[0,\,1/w^2]}
\Phi\bigl(-\lambda(s)d\bigr)
\,dH(s;\tau,w).
\end{aligned} 
\end{equation}

For every \(\lambda\geq1\),
$
\frac{\partial}{\partial d}
\left\{
d\Phi(-\lambda d)
\right\}
=
\Phi(-\lambda d)-\lambda d\phi(\lambda d).
$
By the definition of \(u_0\),
$
\Phi(-x)-x\phi(x)\leq0$ for every $x\geq u_0.
$
Hence, if \(d\geq u_0\), then \(\lambda d\geq u_0\), and
$
\frac{\partial}{\partial d}
\left\{
d\Phi(-\lambda d)
\right\}
\leq0.
$
Because \(d\,p(d,\tau)\) is a mixture of these terms, it is also
nonincreasing for \(d\geq u_0\). The same conclusion holds for
\(\tau=\infty\), for which
$
p(d,\tau)=\Phi(-ad),
q(d,\tau)=1.
$
Therefore, for every threshold \(\tau\), Equation \eqref{eqn:norm_regret} is nonincreasing
in \(d\) for \(d\geq u_0\). If
\(\bar\Delta/s_1\geq0.76>u_0\), it follows that

$$
\max_{0\leq d\leq\bar\Delta/s_1}
\left\{
d\,p(d,\tau)+\gamma q(d,\tau)
\right\}
=
\max_{0\leq d\leq u_0}
\left\{
d\,p(d,\tau)+\gamma q(d,\tau)
\right\}.
$$

This equality holds for every \(\tau\). Consequently, the minimax objective
as a function of \(\tau\), and hence its set of minimizers, does not depend
on \(\bar\Delta\) once \(\bar\Delta/s_1\geq0.76\). Finally, the normalized
objective depends on the remaining primitives only through
$
w=\frac{s_2}{s_1}, 
\gamma=\frac{c}{s_1}.
$

\paragraph{Proof of Part (B)}

We first establish curvature and locate the worst-case effects.
These results give uniqueness, continuity, and the endpoint behavior
of $\tau^\star$. Four further lemmas establish cost
monotonicity.
\\ 
\textit{Notation and preliminary properties}
Retain the notation from Part (A), fix $w>0$, and write
\begin{equation}
\label{eq:tau_cost_risk}
\begin{aligned}
L(\tau,d)&=d\,p(d,\tau), \quad 
R(\tau,d,\gamma)&=L(\tau,d)+\gamma q(d,\tau), \quad 
F(\tau,\gamma)&=\max_{0\leq d\leq u_0}R(\tau,d,\gamma).
\end{aligned}
\end{equation}
Thus $L$ is the expected loss from an incorrect recommendation,
$R$ adds the cost of continuation, and $F$ is worst-case normalized
regret. Part (A) reduces the threshold problem in
Lemma~\ref{thm:dead_zone} to minimizing $F$ over $[0,\infty]$.
Subscripts denote differentiation with respect to the named arguments;
dependence on $w$ is suppressed unless it is varied explicitly.

For finite $\tau\geq0$ and $d>0$, 
\begin{equation}
\small 
\begin{aligned} 
q_\tau=\phi(\tau-d)+\phi(\tau+d)>0,
\qquad q_d=-zq_\tau, \qquad z:=\tanh(\tau d).
\label{eq:tau_cost_qderiv}
\end{aligned} 
\end{equation}
Define 
\begin{equation}
\small 
\begin{aligned} 
b(\tau,d):=-\frac{L_\tau}{q_\tau}
=d\frac{\Phi(d/w-w\tau)-e^{2\tau d}\Phi(-d/w-w\tau)}
{1+e^{2\tau d}}, \qquad G(\tau,d):=L_d+zL_\tau. 
\label{eq:tau_cost_b}
\end{aligned} 
\end{equation}
Differentiating $R=L+\gamma q$ and using
$L_\tau=-bq_\tau$ gives
\begin{equation} \label{eqn:R_tau}
\small 
\begin{aligned} 
R_\tau=L_\tau+\gamma q_\tau
       =q_\tau\,\bigl[\gamma-b\bigr].
       \end{aligned} 
\end{equation} 
Since $q_\tau>0$, the sign of $R_\tau$ is determined by the
difference between the cost $\gamma$ and the marginal
continuation benefit $b$.
Moreover, $q_d+zq_\tau=0$ implies
\begin{equation} \label{eq:tau_cost_Gidentity}
\small 
\begin{aligned} 
R_d+zR_\tau
=L_d+zL_\tau+\gamma(q_d+zq_\tau)
=G.
\end{aligned} 
\end{equation}
Thus this combination of regret derivatives is independent
of $\gamma$.
Finally, recall the integral representation for $p(d,\tau)$ in Equation \eqref{eqn:mixture_a} and also note that 
Conditioning on $X_1/s_1=x$ and using independence of $X_1$ and $X_2$ Equation \eqref{eqn:pd} gives
\begin{equation}
\label{eqn:p_gaussian}
\small 
\begin{aligned} 
p(d,\tau)
=
\Phi(-\tau-d)
+
\int_{-\tau}^{\tau}
\Phi\left(-wx-\frac{d}{w}\right)\phi(x-d)\,dx.
\end{aligned} 
\end{equation}
\begin{lem}
\label{lem:tau_cost_curvature}
For every finite $\tau\geq0$ and $d>0$:
\begin{enumerate}
\item If $a^2d^2\leq2$, then $L_{dd}(\tau,d)\leq0$ and
$L_{\tau d}(\tau,d)\leq0$.
\item If $ad\geq1$ and $d(\tau+d)\geq1$, then $L_d(\tau,d)<0$.
\end{enumerate}
For $\tau,\gamma>0$, every effect $d$ attaining the maximum of $R(\tau,d,\gamma)$  in
\eqref{eq:tau_cost_risk} lies in $(0,u_0)$, satisfies
$R_d=0$ and $R_{dd}\leq0$, and has $L_d>0$.
\end{lem}

\begin{proof}
Using the Gaussian-mixture representation in Part (A), Equation \eqref{eqn:mixture_a},  and taking derivatives we have 
\begin{equation}
\small 
\begin{aligned} 
L_{dd}(\tau,d)=
\int_{[0,1/w^2]}\lambda(s)\phi(\lambda(s)d)
\bigl((\lambda(s)d)^2-2\bigr)\,dH(s;\tau,w)\leq0
\quad\text{if }a^2d^2\leq2.
\label{eq:tau_cost_Ldd}
\end{aligned} 
\end{equation}
Since $1\leq\lambda(s)\leq a$ by Equation \eqref{eqn:lambda_s}; the mixing measure is independent of
$d$, and the bounded range of $\lambda(s)$ justifies differentiation.
For the mixed derivative, differentiating the Gaussian expression
for $p$ and changing variables gives\footnote{Leibniz's rule applied to the Gaussian integral for $p$ in Equation \eqref{eqn:p_gaussian} gives
$ 
L_\tau
=d\Bigl[
\phi(\tau-d)\Phi(-d/w-w\tau)
-\phi(\tau+d)\Phi(d/w-w\tau)
\Bigr].
$ 
Writing the Gaussian tails as
$ 
\Phi(\pm d/w-w\tau)
=w\int_0^\infty
\phi\bigl(w(\tau+y)\mp d/w\bigr)\,dy,
$ 
and expanding the Gaussian exponents, using $a^2=1+1/w^2$, gives
$ 
\phi(\tau\pm d)\phi\bigl(w(\tau+y)\mp d/w\bigr)
=\phi(\tau)\phi\bigl(w(\tau+y)\bigr)
e^{-a^2d^2/2\pm dy}.
$ 
Substituting these identities and using
$e^{dy}-e^{-dy}=2\sinh(dy)$ yields the stated expression.}
\begin{equation}
\small 
\begin{aligned} 
L_\tau=-2w\phi(\tau)\int_0^\infty
\phi\bigl(w(\tau+y)\bigr)
d e^{-a^2d^2/2}\sinh(dy)\,dy.
\end{aligned} 
\label{eq:tau_cost_Ltintegral}
\end{equation}
The factors that depend on $d$ satisfy for $a^2 d^2 \le 2$
$$  
\small 
\begin{aligned}
\frac{\partial}{\partial d}
\left[d e^{-a^2d^2/2}\sinh(dy)\right]
&=e^{-a^2d^2/2}
\bigl[(1-a^2d^2)\sinh(dy)+dy\cosh(dy)\bigr]\\
&\geq e^{-a^2d^2/2}(2-a^2d^2)\sinh(dy)\geq0,
\end{aligned}
$$ 
using $x\cosh x\geq\sinh x$ for $x\geq0$.
Gaussian domination on compact intervals of $d$ justifies
differentiation under the integral, proving the first assertion of the lmma.

For the second assertion, we use an equivalent representation of $p(d,\tau)$  conditioning first on $V=-y<0$ as defined in Equation \eqref{eqn:VV}. First conditioning and then integrating Equation \eqref{eqn:pd} gives
\begin{equation} 
\small 
\begin{aligned} 
p(d,\tau)& =\Phi(-\tau-d)+
\int_0^\infty h_\tau(y)\phi(y+ad)\,dy,  \\  h_\tau(y)& :=\mathbb P(-\tau<Z_1<\tau\mid V=-y)
=\Phi(\Gamma-wy)-\Phi(-\Gamma-wy)\geq0
\label{eq:tau_cost_pcont}
\end{aligned} 
\end{equation}
Differentiating $L= d p(d,\tau)$ with respect to $d$ gives
\begin{equation} \label{eqn:ah}
\small 
\begin{aligned}
L_d(\tau,d)
={}&\Phi(-\tau-d)-d\phi(\tau+d) +\int_0^\infty h_\tau(y)\phi(y+ad)
\bigl(1-a^2d^2-ady\bigr)\,dy.
\end{aligned}
\end{equation}
Note that 
\[
\small 
\begin{aligned} 
\int_0^\infty h_\tau(y)\phi(y+ad)
\bigl(1-a^2d^2-ady\bigr)\,dy\leq0
\qquad\text{if }ad\geq1, 
\end{aligned} 
\]
while for  $d(\tau+d)\geq1$, $ 
\Phi(-\tau-d)-d\phi(\tau+d)
<\phi(\tau+d)\left(\frac1{\tau+d}-d\right)\leq0, 
$ 
where the first inequality follows from 
the Gaussian Mills inequality
$\Phi(-x)<\phi(x)/x$ for $x>0$. This proves the second assertion 

To prove the last assertion, fix $0<\tau<\infty$ and $\gamma>0$.
Differentiating the mixture representation in
Equation \eqref{eqn:mixture_a}, whose mixing measure is
independent of $d$, gives
\[
\small 
\begin{aligned} 
L_d(\tau,d)
=
\int_{[0,1/w^2]}
\Bigl[
\Phi(-\lambda(s)d)
-\lambda(s)d\,\phi(\lambda(s)d)
\Bigr]\,dH(s;\tau,w).
\end{aligned} 
\]
At $d=0$, the integrand equals $1/2$, so
$L_d(\tau,0)=1/2$ because $H$ is a probability measure.
Also,
$q_d(0,\tau)=\phi(\tau)-\phi(\tau)=0$.
Consequently,
$ 
R_d(\tau,0,\gamma)
=L_d(\tau,0)+\gamma q_d(0,\tau)
=\frac12>0.
$ 
At $d=u_0$, Equation \eqref{eqn:lambda_s} gives
$\lambda(s)u_0\geq u_0$.
Part (A) establishes that
$\Phi(-x)-x\phi(x)\leq0$ for every $x\geq u_0$.
Every integrand in the expression for $L_d(\tau,u_0)$
is therefore nonpositive, implying $L_d(\tau,u_0)\leq0$.
Moreover,
$ 
q_d(u_0,\tau)
=\phi(\tau+u_0)-\phi(\tau-u_0)<0,
$ 
since $|\tau-u_0|<\tau+u_0$ and the Gaussian density
strictly decreases with the absolute value of its argument.
Thus
$ 
R_d(\tau,u_0,\gamma)
=L_d(\tau,u_0)+\gamma q_d(u_0,\tau)<0.
$ 
Since $R$ is continuous on $[0,u_0]$, its maximum is attained,
and every maximizer lies in $(0,u_0)$.
At any such maximizer, the first- and second-order
necessary conditions give $R_d=0$ and $R_{dd}\leq0$.
Using $R_d=L_d+\gamma q_d$ and $q_d=-zq_\tau$ then yields
$ 
L_d=-\gamma q_d=\gamma zq_\tau>0,
$ 
because $\gamma>0$, $z=\tanh(\tau d)>0$, and $q_\tau>0$.
\end{proof}

\begin{lem}[Existence, uniqueness, continuity, and endpoint behavior]
\label{lem:tau_cost_unique}
For every $\gamma\geq0$, $F(\cdot,\gamma)$ has a unique minimizer
in $[0,\infty]$. This minimizer equals infinity at $\gamma=0$ and is
finite at every $\gamma>0$. It is continuous in $\gamma\geq0$ with
values in $[0,\infty]$, and therefore tends to infinity as
$\gamma\downarrow0$. In addition,
$\tau^\star(w,\gamma)=0$ whenever $\gamma>u_0$.
\end{lem}

\begin{proof}
Recall the definitions of $R$ and $F$ in Equation \eqref{eq:tau_cost_risk}. 
At finite thresholds, continuity of $R$ and compactness of
$[0,u_0]$ give continuity of $F$. To include $\tau=\infty$,
note that the threshold rule agrees with always-continuation
whenever it continues. Thus
$|p(d,\tau)-\Phi(-ad)|\leq1-q(d,\tau)$.
Using $R=dp+\gamma q$ and
$1-q(d,\tau)\leq2\Phi(u_0-\tau)$ for $0\leq d\leq u_0$ gives
\begin{equation}
\small 
\begin{aligned} 
\sup_{0\leq d\leq u_0}
|R(\tau,d,\gamma)-R(\infty,d,\gamma)|
\leq2(u_0+\gamma)\Phi(u_0-\tau) \rightarrow0, \qquad \text{ as } \tau \rightarrow \infty. 
\label{eq:tau_cost_endpoint_continuity}
\end{aligned} 
\end{equation}
Hence $F$ is continuous and attains a minimum.
Also, $0\leq q\leq1$ implies
$ 
|F(\tau,\gamma')-F(\tau,\gamma)|\leq|\gamma'-\gamma|,
$ 
so this continuity is joint in $(\tau,\gamma)$.

For uniqueness, differentiating \eqref{eqn:p_gaussian} and using
\eqref{eq:tau_cost_qderiv} gives, for $d>0$,
\begin{equation}
\small 
\begin{aligned}
R_{\tau\tau}+\tau R_\tau
={}&d^2\bigl[
\phi(\tau+d)\Phi(d/w-w\tau)
+\phi(\tau-d)\Phi(-d/w-w\tau)\bigr] +\gamma d zq_\tau>0.
\end{aligned}
\label{eq:tau_cost_convexityidentity}
\end{equation}
To make the parameter space compact, we use a change of variable where we write 
$x=2\Phi(\tau)-1$ to map $[0,\infty]$ onto $[0,1]$.
Here $x=1$ corresponds to $\tau=\infty$. Since $dx/d\tau=2\phi(\tau)$, the chain rule yields
$ 
\frac{\partial^2R}{\partial x^2}
=\frac{R_{\tau\tau}+\tau R_\tau}{4\phi(\tau)^2}>0,
0<x<1.
$ 
Thus $R$ is strictly convex in $x$ for every $d>0$, including
its continuous extensions to the endpoints.
The calculation at $d=0$ in the proof of
Lemma~\ref{lem:tau_cost_curvature}
gives $R_d(\tau,0,\gamma)=1/2$, so every maximizing effect at an
interior coordinate is positive. 
Since $R_d(\tau,0,\gamma)=1/2>0$, the maximum defining $F$
is attained at a positive effect. Thus, for each fixed $\gamma$,
$F$ is an attained pointwise maximum of functions that are
strictly convex in $x$. It is therefore strictly convex,
and its minimizer is unique.

We next determine the boundary behavior.
When $\gamma=0$, we have $R=L$, and
Equation \eqref{eq:tau_cost_Ltintegral} gives $L_\tau<0$
for every $d>0$. To see that worst-case regret also decreases,
take finite thresholds $\tau_1<\tau_2$ and choose $d^\star>0$
attaining $F(\tau_2,0)$. Then
$ 
F(\tau_1,0)
\geq L(\tau_1,d^\star)
>L(\tau_2,d^\star)
=F(\tau_2,0).
$ 
Thus $F(\tau,0)$ strictly decreases with $\tau$.
Together with continuity at infinity from
\eqref{eq:tau_cost_endpoint_continuity}, this gives
$\tau^\star(w,0)=\infty$. Now let $\gamma>0$. The expression for $b$ in
\eqref{eq:tau_cost_b} extends to $b(\tau,0)=0$ and gives
$ 
0\leq b(\tau,d)
\leq d\Phi(d/w-w\tau)
\leq u_0\Phi(u_0/w-w\tau),
 0\leq d\leq u_0.$ 
The last bound is independent of $d$ and tends to zero as
$\tau\to\infty$. Therefore, at every sufficiently large
threshold, $b(\tau,d)<\gamma$ for all effects under consideration.
Equations  \eqref{eq:tau_cost_qderiv} and \eqref{eqn:R_tau} then give
$ 
R_\tau=q_\tau\,\bigl[\gamma-b\bigr]>0.
$ 
To transfer this increase to $F$, take two such thresholds
$\tau_1<\tau_2$ and choose $d^\star$ attaining
$F(\tau_1,\gamma)$. Then
$ 
F(\tau_2,\gamma)
\geq R(\tau_2,d^\star,\gamma)
>R(\tau_1,d^\star,\gamma)
=F(\tau_1,\gamma).
$ 
Thus $F$ strictly increases beyond a finite threshold.
Infinity cannot minimize it, so $\tau^\star(w,\gamma)<\infty$. The same bound gives $b(\tau,d)\leq u_0$ at every
finite threshold. If $\gamma>u_0$, then $R_\tau>0$
everywhere, and the preceding comparison shows that $F$
strictly increases over all finite thresholds.
Its minimum is therefore at $\tau^\star(w,\gamma)=0$.

Finally, fix $w>0$ and define the minimizing coordinate
$ 
x^\star(\gamma):=2\Phi(\tau^\star(w,\gamma))-1.
$ 
Let $\gamma_n\to\gamma_0\geq0$. Consider any convergent
subsequence $x^\star(\gamma_{n_k})\to\bar x$, and let
$\bar\tau=\Phi^{-1}((1+\bar x)/2)$, with
$\bar\tau=\infty$ when $\bar x=1$.
By optimality,
$ 
F(\tau^\star(w,\gamma_{n_k}),\gamma_{n_k})
\leq F(\tau,\gamma_{n_k})
\text{ for every }\tau\in[0,\infty].
$ 
Joint continuity in $(x,\gamma)$ allows us to take limits 
$ 
F(\bar\tau,\gamma_0)\leq F(\tau,\gamma_0)
\text{ for every }\tau\in[0,\infty].
$ 
Thus $\bar\tau$ minimizes regret at cost $\gamma_0$.
Uniqueness implies $\bar\tau=\tau^\star(w,\gamma_0)$,
and hence $\bar x=x^\star(\gamma_0)$.
Because $[0,1]$ is compact and every convergent subsequence
has this same limit, the full sequence satisfies
$x^\star(\gamma_n)\to x^\star(\gamma_0)$. For $\gamma_0>0$, we have $x^\star(\gamma_0)<1$, so the
continuous inverse transformation gives
$\tau^\star(w,\gamma_n)\to\tau^\star(w,\gamma_0)$.
At $\gamma_0=0$, we have $x^\star(0)=1$, so
$\tau^\star(w,\gamma)\to\infty$ as $\gamma\downarrow0$.
\end{proof}

It remains to prove cost monotonicity. The next lemma controls the
shape of $b$ as the effect varies. This will let us control changes
in the sign of $G$.

\begin{lem}
\label{lem:tau_cost_bshape}
For fixed $\tau,w>0$, the positive function $d\mapsto b(\tau,d)$ is
strictly log-concave on $(0,\infty)$. In addition:
\begin{enumerate}
\item If $b_d(\tau,d)\leq0$, then $\tau d(1+z)>1$.
\item If $d/w=\sqrt2$ and $\tau d\leq1$, then $b_d(\tau,d)>0$.
\end{enumerate}
\end{lem}

\begin{proof}
Fix $\tau,w>0$ and retain $z=\tanh(\tau d)$ from
\eqref{eq:tau_cost_qderiv}. Changing variables in the Gaussian
tails in \eqref{eq:tau_cost_b} gives
\begin{equation}
\small 
\begin{aligned} 
b(\tau,d)=\frac{d}{1+e^{2\tau d}}
\int_0^\infty
\phi(d/w-w\tau-y)\bigl(1-e^{-2dy/w}\bigr)\,dy>0.
\label{eq:tau_cost_h}
\end{aligned} 
\end{equation}
The integrand is jointly log-concave for $d,y>0$.
Define it to be zero whenever $d\leq0$ or $y\leq0$.
Since the positive quadrant is convex, this extension is
log-concave on $\mathbb R^2$.\footnote{The Gaussian factor is log-concave because
its argument is affine. For the other factor, the Hessian of
$\log(1-e^{-2dy/w})$ has negative diagonal entries and determinant
$\frac{4}{w^2(e^{2dy/w}-1)^2}
\left(\frac{4dy/w}{1-e^{-2dy/w}}-1\right)>0$,
using $1-e^{-2dy/w}<2dy/w$.}
By Pr\'ekopa's theorem\footnote{Pr\'ekopa's theorem states that if $g(d,y)\geq0$
is jointly log-concave, then
$d\mapsto\int_{\mathbb R}g(d,y)\,dy$ is log-concave,
provided the integral is finite. Here $g$ is the integrand
in \eqref{eq:tau_cost_h}, extended by zero outside
$(0,\infty)^2$.}, integration over $y$ preserves
log-concavity. Since $\log d$ is strictly concave and
$-\log(1+e^{2\tau d})$ is concave, \eqref{eq:tau_cost_h}
implies that $b(\tau,d)$ is strictly log-concave in $d$.

For the first additional assertion, differentiating the
Gaussian expression in \eqref{eq:tau_cost_b} shows that the
integral factor in \eqref{eq:tau_cost_h} has derivative
$ 
\frac{2}{w}\phi(d/w-w\tau)
-2\tau e^{2\tau d}\Phi(-d/w-w\tau)>0, 
$ wherepositivity follows from Mills' inequality and the identity 
$e^{2\tau d}\phi(d/w+w\tau)=\phi(d/w-w\tau)$.\footnote{For $d,\tau,w>0$, Mills' inequality and the Gaussian
identity give
$
e^{2\tau d}\Phi(-d/w-w\tau)
<\frac{\phi(d/w-w\tau)}{d/w+w\tau}.
$
Thus the displayed derivative is strictly greater than
$
2\phi(d/w-w\tau)
\left(\frac1w-\frac{\tau}{d/w+w\tau}\right)
=\frac{2d\,\phi(d/w-w\tau)}{w(d+w^2\tau)}>0.
$}
Moreover,
$ 
\frac{\partial}{\partial d}\frac{d}{1+e^{2\tau d}}
=\frac{1-\tau d(1+z)}{1+e^{2\tau d}}.
$ 
Thus, if $\tau d(1+z)\leq1$, the product in
\eqref{eq:tau_cost_h} has strictly positive derivative.
Consequently, $b_d(\tau,d)\leq0$ requires
$\tau d(1+z)>1$, proving the first assertion.

For the second assertion, differentiating
\eqref{eq:tau_cost_b} gives
\begin{equation}
\small 
\begin{aligned}
b_\tau& =-\frac{d^2}{2}(1-z^2)
\bigl[\Phi(d/w-w\tau)+\Phi(-d/w-w\tau)\bigr]<0, \\ 
b_{\tau d} &=-d(1-z^2)\Bigl[
(1-\tau dz)\bigl\{\Phi(d/w-w\tau)+\Phi(-d/w-w\tau)\bigr\} +\frac{d}{2w}
\bigl\{\phi(d/w-w\tau)-\phi(d/w+w\tau)\bigr\}\Bigr].
\end{aligned}
\label{eq:tau_cost_btd}
\end{equation}
When $\tau d\leq1$, we have $1-\tau dz>0$ and
$\phi(d/w-w\tau)>\phi(d/w+w\tau)$, so $b_{\tau d}<0$.
Therefore, with $d,w$ fixed, $b_d(\tau,d)$ decreases as
$\tau$ increases up to $1/d$. It suffices to prove
$b_d(1/d,d)>0$ when $d/w=\sqrt2$.

Differentiating at this end point (where $d\tau = 1$)
\eqref{eq:tau_cost_b} and substituting gives for $z := \mathrm{tanh} 1$
\[
\small 
\begin{aligned}
(1+e^2)b_d(1/d,d)
={}&-z\Phi(1/\sqrt2)-e^2(2-z)\Phi(-3/\sqrt2)
     +2\sqrt2\phi(1/\sqrt2)\\
>{}&\frac{\sqrt2(4+z)}{3}\phi(1/\sqrt2)
     -z\Phi(1/\sqrt2),
\end{aligned}
\]
where the inequality uses Mills' inequality and
$e^2\phi(3/\sqrt2)=\phi(1/\sqrt2)$.\footnote{Mills' inequality gives
$e^2\Phi(-3/\sqrt2)
<(\sqrt2/3)e^2\phi(3/\sqrt2)
=(\sqrt2/3)\phi(1/\sqrt2)$.
Since $2-z>0$, multiplying by $-(2-z)$ reverses the
inequality. Therefore,
$
-e^2(2-z)\Phi(-3/\sqrt2)+2\sqrt2\phi(1/\sqrt2)
>\frac{\sqrt2(4+z)}{3}\phi(1/\sqrt2).
$}
The last expression decreases in $z$. Using
$z=\tanh1<4/5$ and
$\Phi(1/\sqrt2)<1/2+1/(2\sqrt\pi)$ therefore gives
$ 
(1+e^2)b_d(1/d,d)
>\frac{8e^{-1/4}-2\sqrt\pi-2}{5\sqrt\pi}>0,
$ 
where the last inequality follows from $e^{-1/4}>3/4$
and $\pi<4$. This proves the second assertion.
\end{proof}

\begin{lem}[Key sign restriction]
\label{lem:tau_cost_key}
Fix $\tau,w>0$. If $d>0$ satisfies
$ 
L_d(\tau,d)\geq0,  a^2d^2>2, b_d(\tau,d)\leq0,
$ 
then $G(\tau,d)<0$.
\end{lem}

\begin{proof}
The second assertion of Lemma~\ref{lem:tau_cost_curvature}
states that $L_d(\tau,d)<0$ whenever $ad\geq1$ and
$d(\tau+d)\geq1$. Here $ad>\sqrt2>1$, so the first
condition holds. If $d(\tau+d)\geq1$ also held, the lemma
would imply $L_d(\tau,d)<0$, contradicting our assumption
$L_d(\tau,d)\geq0$. Therefore,
$ 
d(\tau+d)<1.
$ 
Using $a^2=1+1/w^2$ and $a^2d^2>2$, we obtain
$ 
\frac{d^2}{w^2}>2-d^2>1+\tau d>2(\tau d)^2,
$ 
where the last inequality uses $0<\tau d<1$.
Thus $\tau\sqrt2w<1$, and
Lemma~\ref{lem:tau_cost_bshape} gives
$b_d(\tau,\sqrt2w)>0$.
The proof of Lemma~\ref{lem:tau_cost_bshape} also establishes strict log-concavity of $b$, so
$b_d/b$ decreases with $d$. 
If $d\leq\sqrt2w$, we would therefore have
$ 
\frac{b_d(\tau,d)}{b(\tau,d)}
\geq
\frac{b_d(\tau,\sqrt2w)}{b(\tau,\sqrt2w)}
>0.
$
Since $b(\tau,d)>0$, this would imply $b_d(\tau,d)>0$,
contradicting the assumption $b_d(\tau,d)\leq0$. Hence $d>\sqrt2w$. Dividing by $w>0$ and squaring gives
$d^2/w^2>2$. Since $a^2-1=1/w^2$ by \eqref{eqn:VV},
we obtain $(a^2-1)d^2>2$.

Next hold $(\tau,d)$ fixed and increase $a$, with
$w=(a^2-1)^{-1/2}$ as in \eqref{eqn:VV}.
Differentiating \eqref{eqn:p_gaussian} gives\footnote{Let $x$ be the integration variable in
\eqref{eqn:p_gaussian}. Differentiating Equation \eqref{eqn:p_gaussian} at fixed $\tau,d$ gives
$
p_w=\int_{-\tau}^{\tau}(d/w^2-x)
\phi(wx+d/w)\phi(x-d)\,dx.
$
The following identity
$\phi(wx+d/w)\phi(x-d)
=\phi(ad)\phi(\sqrt{1+w^2}\,x)$ follows directly from the function $\phi$ and 
makes the product even in $x$. Hence the term proportional
to $-x$ integrates to zero over $[-\tau,\tau]$.
Integrating the remaining term gives
$
p_w=\frac{d\phi(ad)\{2\Phi(\Gamma)-1\}}
{w^2\sqrt{1+w^2}}.
$
Finally, the chain rule gives $p_a=p_w/(da/dw)$, where
$da/dw=-1/(aw^3)$, yielding the stated expression.}
$ 
p_a=-d\phi(ad)\{2\Phi(\Gamma)-1\}.
$ 
Using $L=dp$ and $G=L_d+zL_\tau$ from
\eqref{eq:tau_cost_b}, with $z = \mathrm{tanh}(\tau d)$ in Equation \eqref{eq:tau_cost_qderiv}, gives
\begin{equation}
\small 
\begin{aligned} 
G_a=d\phi(ad)\Bigl[
(a^2d^2-2)\{2\Phi(\Gamma)-1\}
-\frac{2dz}{\tau}\Gamma\phi(\Gamma)
\Bigr].
\label{eq:tau_cost_Ga}
\end{aligned}
\end{equation}
Since
$ 
2\Phi(\Gamma)-1
=2\int_0^\Gamma\phi(t)\,dt>2\Gamma\phi(\Gamma),
z<\tau d,
$ 
we have
$ 
G_a>d\phi(ad)\{2\Phi(\Gamma)-1\}
\bigl[(a^2-1)d^2-2\bigr]>0.
$ 
We established above that $(a^2-1)d^2>2$.
Since $d$ is held fixed, this inequality continues to hold
as $a$ increases. Therefore, the lower bound for $G_a$
remains strictly positive, and $G$ is strictly increasing in $a$.

As $a\to\infty$, we have $w\downarrow0$.
The expression for $L_d$ in \eqref{eqn:ah}
and $L_\tau=-bq_\tau$ in \eqref{eq:tau_cost_b} together with $G=L_d+zL_\tau$ from
\eqref{eq:tau_cost_b} give\footnote{Since $0\leq h_\tau\leq1$, the third term (the integral)
in the expression for $L_d$ in Equation \eqref{eqn:ah} satisfies
$ 
\left|
\int_0^\infty h_\tau(y)\phi(y+ad)
(1-a^2d^2-ady)\,dy
\right|
\leq
\int_0^\infty \phi(y+ad)\bigl[1+ad(y+ad)\bigr]\,dy
=\Phi(-ad)+ad\phi(ad)\to0
$ 
as $a\to\infty$, with $d>0$ fixed. Moreover, \eqref{eq:tau_cost_b} gives
$b\to d/(1+e^{2\tau d})$, while
\eqref{eq:tau_cost_qderiv} gives
$q_\tau=(1+e^{2\tau d})\phi(\tau+d)$.
Thus $L_\tau=-bq_\tau\to-d\phi(\tau+d)$.}
$ 
\lim_{a\to\infty}G(\tau,d)
=\Phi(-\tau-d)-d\phi(\tau+d)(1+z).
$ 
Lemma~\ref{lem:tau_cost_bshape} and the hypothesis
$b_d(\tau,d)\leq0$ give $\tau d(1+z)>1$.
Since $z=\tanh(\tau d)$, this inequality depends only on
$\tau,d$, which remain fixed as $a$ increases.
Mills' inequality therefore yields
$ 
\lim_{a\to\infty}G(\tau,d)
<\phi(\tau+d)\left[\frac{1}{\tau+d}-d(1+z)\right]<0.
$ The final inequality holds because $\phi(\tau+d)>0$ and
$\tau d(1+z)>1$ implies $d(1+z)>1/\tau>1/(\tau+d)$.
For these fixed $\tau,d$, $G$ is increasing on the range
$(a^2-1)d^2>2$ and converges to a strictly negative limit.
Hence $G(\tau,d)<0$ throughout this range, which includes
the value of $a$ under consideration, completing the proof.
\end{proof}

For $\tau,\gamma>0$, define the set of worst-case e
$$  
\small 
\begin{aligned}
\mathcal M(\tau,\gamma)
=\arg\max_{0\leq d\leq u_0}R(\tau,d,\gamma).
\end{aligned} 
$$ 
By Lemma~\ref{lem:tau_cost_curvature}, this is a nonempty compact
subset of $(0,u_0)$. At any such effect, $R_d=0$, and hence
\begin{equation}
G(\tau,d)=zR_\tau(\tau,d,\gamma)
=zq_\tau\{\gamma-b(\tau,d)\},
\qquad d\in\mathcal M(\tau,\gamma).
\label{eq:tau_cost_Gactive}
\end{equation}
The next lemma controls both a zero of $G$ and ties. 

\begin{lem}[Zeros and ties]
\label{lem:tau_cost_zeros}
Fix $\tau,\gamma>0$.
\begin{enumerate}
\item If $d\in\mathcal M(\tau,\gamma)$ and $G(\tau,d)=0$, then
$G_d(\tau,d)<0$.
\item There do not exist $d_1,d_2\in\mathcal M(\tau,\gamma)$ with
$d_1<d_2$ and $G(\tau,d_2)\geq0\geq G(\tau,d_1)$.
\end{enumerate}
\end{lem}

\begin{proof}
By the last assertion of Lemma~\ref{lem:tau_cost_curvature},
every $d\in\mathcal M(\tau,\gamma)$ lies in $(0,u_0)$ and
satisfies $R_d=0$, $R_{dd}\leq0$, and $L_d>0$.

For the first assertion, suppose $G(\tau,d)=0$ at such a
maximizer. Since $z,q_\tau>0$,
\eqref{eq:tau_cost_Gactive} gives $R_\tau=0$ and $b=\gamma$.
Differentiating \eqref{eq:tau_cost_Gidentity} and
\eqref{eqn:R_tau}, then evaluating at this point, gives
$G_d
=R_{dd}+z_dR_\tau+zR_{\tau d}
=R_{dd}-zq_\tau b_d.
$\footnote{The first equality follows from the product rule applied
to $G=R_d+zR_\tau$. For the second, differentiating
$R_\tau=q_\tau(\gamma-b)$ with respect to $d$, holding
$\tau,\gamma$ fixed, gives
$
R_{\tau d}=q_{\tau d}(\gamma-b)-q_\tau b_d.
$
At the point under consideration, $R_\tau=0$ and $\gamma=b$.
Hence $z_dR_\tau=0$ and $R_{\tau d}=-q_\tau b_d$, yielding
$G_d=R_{dd}-zq_\tau b_d$.} 
If $b_d>0$, this is strictly negative because $R_{dd}\leq0$.
If $b_d\leq0$, Lemma~\ref{lem:tau_cost_key} rules out
$a^2d^2>2$: together with $L_d>0$, this would imply $G<0$.
Hence $a^2d^2\leq2$.
Using $G=L_d+zL_\tau$ from \eqref{eq:tau_cost_b},
$z_d=\tau(1-z^2)$, and $L_\tau=-bq_\tau=-\gamma q_\tau$,
we obtain
$ 
G_d=L_{dd}+zL_{\tau d}
-\gamma\tau(1-z^2)q_\tau<0.
$ 
Here the first two terms are nonpositive by the first assertion
of Lemma~\ref{lem:tau_cost_curvature}, and the last is strictly
negative.

For the second assertion, suppose $d_1<d_2$ are maximizers
with $G(\tau,d_1)\leq0\leq G(\tau,d_2)$.
Equation \eqref{eq:tau_cost_Gactive} gives
$b(\tau,d_1)\geq\gamma\geq b(\tau,d_2)$.
Strict log-concavity from Lemma~\ref{lem:tau_cost_bshape}
therefore yields
$ 
\frac{b_d(\tau,d_2)}{b(\tau,d_2)}
<
\frac{\log b(\tau,d_2)-\log b(\tau,d_1)}{d_2-d_1}
\leq0.
$ 
Thus $b_d(\tau,d_2)<0$.
Since $L_d(\tau,d_2)>0$ and $G(\tau,d_2)\geq0$,
Lemma~\ref{lem:tau_cost_key} again implies $a^2d_2^2\leq2$.
The first assertion of Lemma~\ref{lem:tau_cost_curvature}
then gives $L_{dd}\leq0$ throughout $[d_1,d_2]$.
Differentiating \eqref{eq:tau_cost_qderiv} also gives, on
this interval,
$ 
q_{dd}=-(\tau-dz)q_\tau
\leq-\tau(1-d^2)q_\tau<0,
$ 
where the inequalities follows from $z=\tanh(\tau d)\leq\tau d$ and
$d\leq d_2<u_0<1$.\footnote{Differentiating $q_d=\phi(\tau+d)-\phi(\tau-d)$ gives
$
q_{dd}=-(\tau+d)\phi(\tau+d)-(\tau-d)\phi(\tau-d)
=-\tau q_\tau-dq_d=-(\tau-dz)q_\tau.
$
Also, $z=\tanh(\tau d)\leq\tau d$ implies
$\tau-dz\geq\tau(1-d^2)>0$, since $\tau>0$ and
$d\leq d_2<u_0<1$. Multiplying by $-q_\tau<0$ gives
the stated inequalities.}
Consequently, $R_{dd}=L_{dd}+\gamma q_{dd}<0$ on
$[d_1,d_2]$. Thus $R$ is strictly concave there,
contradicting that both endpoints are global maximizers.
\end{proof}

\begin{lem}[Cost single crossing at a fixed threshold]
\label{lem:tau_cost_crossing}
Fix $\tau>0$ and select the smallest worst-case effect,
$ 
m(\gamma)=\min\mathcal M(\tau,\gamma), \gamma>0.
$ 
There do not exist $0<\gamma_1<\gamma_2$ such that
$ 
G(\tau,m(\gamma_1))>0>G(\tau,m(\gamma_2)).
$ 
\end{lem}

\begin{proof}
For $0<\gamma_1<\gamma_2$, take
$d_i\in\mathcal M(\tau,\gamma_i)$, $i=1,2$. Since $d_i\in\mathcal M(\tau,\gamma_i)$, each $d_i$ maximizes
regret at its own cost $\gamma_i$. Thus
\begin{equation} 
\small 
\label{eqn:opt1}
R(\tau,d_1,\gamma_1) \geq R(\tau,d_2,\gamma_1), \qquad  
R(\tau,d_2,\gamma_2) \geq R(\tau,d_1,\gamma_2).
\end{equation}   
Adding these inequalities and using $R=L+\gamma q$ from
\eqref{eq:tau_cost_risk}, the terms $L(\tau,d_1)$ and
$L(\tau,d_2)$ cancel, giving
$ 
(\gamma_2-\gamma_1)
\bigl[q(d_2,\tau)-q(d_1,\tau)\bigr]\geq0.
$ 
By \eqref{eq:tau_cost_qderiv}, $q_d=-zq_\tau<0$ for $d>0$,
so $d_2\leq d_1$. In particular, $m$ is nonincreasing because choosing
$d_i=m(\gamma_i)$ in the preceding comparison gives
$m(\gamma_2)\leq m(\gamma_1)$ whenever $0<\gamma_1<\gamma_2$.
Since $m$ is nonincreasing and $0<m(\gamma)<u_0$, both
one-sided limits exist and lie in $[0,u_0]$.
By optimality, as in \eqref{eqn:opt1},
$ 
R(\tau,m(\gamma'),\gamma')
\geq R(\tau,d,\gamma')
\text{ for every }d\in[0,u_0].
$ 
Letting $\gamma'\downarrow\gamma$ and using continuity of $R$ which follows immediately from the expression for $L = dp(d, \tau)$ and $q$ in Equation \eqref{eq:tau_cost_risk} and \eqref{eqn:norm_regret} respectively (where $p$ is defined in Equation \eqref{eqn:pd})
gives
$ 
R(\tau,m(\gamma+),\gamma)
\geq R(\tau,d,\gamma)
 \text{ for every }d\in[0,u_0].
$ 
Thus $m(\gamma+)$ itself maximizes regret at cost $\gamma$,
so $m(\gamma+)\in\mathcal M(\tau,\gamma)$.
The same argument with $\gamma'\uparrow\gamma$ gives
$m(\gamma-)\in\mathcal M(\tau,\gamma)$.
The comparison $d_2\leq d_1$, established from
\eqref{eqn:opt1} and \eqref{eq:tau_cost_qderiv}, holds for
every pair $d_i\in\mathcal M(\tau,\gamma_i)$ with
$0<\gamma_1<\gamma_2$.
Thus, taking the corresponding one-sided limits for given $\gamma$ in lieu of $\gamma_1, \gamma_2$ gives
$ 
m(\gamma+)\leq d\leq m(\gamma-)
 \text{ for every }d\in\mathcal M(\tau,\gamma).
$ 
Since both limits themselves belong to $\mathcal M(\tau,\gamma)$,
the right limit is its smallest member and the left limit its
largest. Hence
\begin{equation} 
\small 
\begin{aligned} 
m(\gamma+)=\min\mathcal M(\tau,\gamma)=m(\gamma),
m(\gamma-)=\max\mathcal M(\tau,\gamma).
\label{eq:tau_cost_limits}
\end{aligned}
\end{equation}
In particular, $m$ is right-continuous.

We now prove the statement by contradiction. Suppose
$G(\tau,m(\gamma_1))>0>G(\tau,m(\gamma_2))$
for some $0<\gamma_1<\gamma_2$, and define
$ 
\bar\gamma
=\inf\bigl\{\gamma\in[\gamma_1,\gamma_2]:
G(\tau,m(\gamma))\leq0\bigr\}.
$ 
Right-continuity of $m$ and continuity of $G$ which follows directly from its expression in Equation \eqref{eq:tau_cost_b}\footnote{For fixed $\tau,w>0$, the Gaussian representation
\eqref{eqn:p_gaussian} has a smooth integrand on a finite
integration interval. Differentiation under the integral
and Leibniz's rule therefore show that $L_d$ and $L_\tau$
are continuous in $d$.
Hence $G=L_d+zL_\tau$ in \eqref{eq:tau_cost_b}
is continuous in $d$, since $z=\tanh(\tau d)$ is continuous.} give
$\bar\gamma>\gamma_1$ and $G(\tau,m(\bar\gamma))\leq0$.
By definition, $G(\tau,m(\gamma))>0$ for
$\gamma_1\leq\gamma<\bar\gamma$. Taking the limit from below
therefore gives
$ 
G(\tau,m(\bar\gamma-))
\geq0\geq G(\tau,m(\bar\gamma)).
$ 

By Equation \eqref{eq:tau_cost_limits}, both effects are maximizers
at cost $\bar\gamma$, with
$m(\bar\gamma-)\geq m(\bar\gamma)$.
The second assertion of Lemma~\ref{lem:tau_cost_zeros}
rules out strict inequality. Thus
$m(\bar\gamma-)=m(\bar\gamma)$, and the displayed inequalities
imply $G(\tau,m(\bar\gamma))=0$.
The first assertion of Lemma~\ref{lem:tau_cost_zeros} then gives
$ 
G_d(\tau,m(\bar\gamma))<0.
$  As $\gamma\uparrow\bar\gamma$, the fact that $m$ is nonincreasing
and the equality $m(\bar\gamma-)=m(\bar\gamma)$ imply
$m(\gamma)\geq m(\bar\gamma)$ and
$m(\gamma)\to m(\bar\gamma)$.
Since $m(\bar\gamma)\in\mathcal M(\tau,\bar\gamma)$ and
$G(\tau,m(\bar\gamma))=0$, the first assertion of
Lemma~\ref{lem:tau_cost_zeros} gives
$G_d(\tau,m(\bar\gamma))<0$. Therefore,
Since $G(\tau,d)$ has a negative derivative with respect to $d$
at $d=m(\bar\gamma)$, we obtain
$ 
G(\tau,m(\gamma))
\leq G(\tau,m(\bar\gamma))=0
$ 
for $\gamma<\bar\gamma$ sufficiently close to $\bar\gamma$.
This contradicts the definition of $\bar\gamma$.
\end{proof}

We can now complete Part (B).
Lemma~\ref{lem:tau_cost_unique} establishes existence,
uniqueness, continuity, and the endpoint statements.
For cost monotonicity, suppose $0<\gamma_1<\gamma_2$ but
$ 
\tau_1:=\tau^\star(w,\gamma_1)
<
\tau_2:=\tau^\star(w,\gamma_2).
$ 
Choose $\tau\in(\tau_1,\tau_2)$ and
$d_i=\min\mathcal M(\tau,\gamma_i)$, $i=1,2$.
Since each $\tau_i$ uniquely minimizes $F(\cdot,\gamma_i)$,
the definition of $F$ in \eqref{eq:tau_cost_risk} gives
$ 
R(\tau_i,d_i,\gamma_i)
\leq F(\tau_i,\gamma_i)
<F(\tau,\gamma_i)
=R(\tau,d_i,\gamma_i).
$ 

The proof of Lemma~\ref{lem:tau_cost_unique}, using
\eqref{eq:tau_cost_convexityidentity}, establishes convexity
of $R$ in $x=2\Phi(\tau)-1$ for fixed $d_i,\gamma_i$.
A differentiable convex function lies on or above its tangent
line. Here we hold $d_i,\gamma_i$ fixed and regard $R$ as a
function of $x=2\Phi(\tau)-1$.
Since $dx/d\tau=2\phi(\tau)$, the tangent slope at the chosen
threshold $\tau$ is
$R_\tau(\tau,d_i,\gamma_i)/(2\phi(\tau))$.
The change in $x$ from $\tau$ to $\tau_i$ is
$2[\Phi(\tau_i)-\Phi(\tau)]$. Therefore,
$ 
R(\tau_i,d_i,\gamma_i)
\geq R(\tau,d_i,\gamma_i)
+\frac{R_\tau(\tau,d_i,\gamma_i)}{\phi(\tau)}
\bigl[\Phi(\tau_i)-\Phi(\tau)\bigr].
$ 
Subtracting $R(\tau,d_i,\gamma_i)$ and using the preceding
strict comparison gives gives
$ 
\frac{R_\tau(\tau,d_i,\gamma_i)}{\phi(\tau)}
\bigl[\Phi(\tau_i)-\Phi(\tau)\bigr]
\leq R(\tau_i,d_i,\gamma_i)-R(\tau,d_i,\gamma_i)<0.
$ 
Since $\Phi(\tau_1)<\Phi(\tau)<\Phi(\tau_2)$ and
$\phi(\tau)>0$, we obtain
$ 
R_\tau(\tau,d_1,\gamma_1)>0,
R_\tau(\tau,d_2,\gamma_2)<0.
$ 

Here $\tau>0$ and $d_i>0$ by
Lemma~\ref{lem:tau_cost_curvature}.
At each selected maximizer $d_i\in\mathcal M(\tau,\gamma_i)$,
\eqref{eq:tau_cost_Gactive} gives
$ 
G(\tau,d_i)
=\tanh(\tau d_i)\,R_\tau(\tau,d_i,\gamma_i),
i=1,2.
$ 
Since $\tau>0$ and $d_i>0$, the factor $\tanh(\tau d_i)$
is positive. \eqref{eq:tau_cost_Gactive} and since we established above  $R_\tau(\tau,d_1,\gamma_1)>0$
and $R_\tau(\tau,d_2,\gamma_2)<0$, 
$ 
G(\tau,d_1)
=\tanh(\tau d_1)\,R_\tau(\tau,d_1,\gamma_1)>0, 
G(\tau,d_2)
=\tanh(\tau d_2)\,R_\tau(\tau,d_2,\gamma_2)<0.
$ This implies 
$ 
G(\tau,d_1)>0>G(\tau,d_2) 
$ 
which contradicts Lemma~\ref{lem:tau_cost_crossing},
since $d_1,d_2$ are the smallest maximizing effects at
the same threshold $\tau$.
Therefore $\tau^\star(w,\gamma)$ is nonincreasing for
$\gamma>0$. Since $\tau^\star(w,0)=\infty$ by
Lemma~\ref{lem:tau_cost_unique}, the conclusion also
includes comparisons with zero cost.

\subsection{Proof of Theorem \ref{thm:a1}} \label{proof:thm:a1}

Fix $w>0$, a finite $t, \gamma \geq0$.
Define
$ 
\Gamma_w(t)
=
\inf\{\eta\geq0:t\geq\tau^\star(w,\eta)\},
$ 
with the convention $\inf\varnothing=\infty$. If $t\geq\tau^\star(w,\gamma)$, then $\gamma$ belongs to the
set defining the infimum, and therefore
$\Gamma_w(t)\leq\gamma$. Conversely, suppose that $t<\tau^\star(w,\gamma)$.
By Lemma \ref{cor:1}(B), there exists $\varepsilon>0$
such that
$ 
t<\tau^\star(w,\gamma+\varepsilon).
$ 
For $\gamma>0$, this follows from continuity at $\gamma$.
For $\gamma=0$, it follows from
$\tau^\star(w,\eta)\to\infty$ as $\eta\downarrow0$.
Monotonicity then implies that, for every
$0\leq\eta\leq\gamma+\varepsilon$,
$ 
\tau^\star(w,\eta)
\geq\tau^\star(w,\gamma+\varepsilon)>t.
$ 
Thus no such $\eta$ belongs to the set defining $\Gamma_w(t)$,
and
$ 
\Gamma_w(t)\geq\gamma+\varepsilon>\gamma.
$ 
Combining the implications gives
$ 
t<\tau^\star(w,\gamma)$  if and only if $\gamma<\Gamma_w(t).
$ 
Taking $t=|X_1|/s_1$ and $\gamma=c/s_1$, and using
$A_{\mathrm{value}}(X_1;s_1,w)=s_1\Gamma_w(t)$,
proves the theorem.

\subsection{Proof of Proposition \ref{prop:c_value_properties}}
\label{proof:prop:c_value_properties}

Let
$ 
t=\frac{|X_1|}{s_1}.
$ 
By the definitions of the $A$-value and $\Gamma_w$,
$ 
A_{\mathrm{value}}(X_1;s_1,w)
=
s_1\Gamma_w(t).
$ 
Moreover, since
$ 
p(t)=2\{1-\Phi(t)\},
$ 
and $t\geq0$, we have
$ 
t
=
\Phi^{-1}\left(1-\frac{p(t)}{2}\right).
$ 
Substitution gives
$ 
A_{\mathrm{value}}(X_1;s_1,w)
=
s_1\Gamma_w\left(
\Phi^{-1}\left(
1-\frac{p(|X_1|/s_1)}{2}
\right)
\right).
$ To establish monotonicity, consider $0\leq t_1\leq t_2$. Then  
$ 
\left\{
\gamma\geq0:
t_1\geq\tau^\star(w,\gamma)
\right\}
\subseteq
\left\{
\gamma\geq0:
t_2\geq\tau^\star(w,\gamma)
\right\}.
$ 
Taking infima over the two sets gives
$ 
\Gamma_w(t_2)\leq\Gamma_w(t_1).
$ 
Hence, for fixed $(s_1,w)$, the $A$-value is weakly decreasing in
$|X_1|/s_1$. Because
$ 
p'(t)=-2\phi(t)<0,
$ 
it is equivalently weakly increasing in the two-sided $p$-value. Next, holding the $p$-value and $w$ fixed, the $A$-value is proportional
to $s_1$ and is therefore weakly increasing in $s_1^2$. The increase is
strict whenever the $A$-value is positive. 

The last property follows directly from Equation \eqref{eqn:characterization}.

\subsection{Proof of Theorem \ref{thm:optimality1}}
\label{app:proof:optimality1}

Throughout the proof, write $\bar s=\max_{1\leq j\leq J}s_j$.
Constants denoted by $C$ are universal and may change from one
display to the next. Let $u_0\simeq0.7518$ denote the unique positive
solution to $\Phi(-u_0)=u_0\phi(u_0)$, so that
$\sup_{u\geq0}u\Phi(-u)=u_0\Phi(-u_0)$.  Denote $\Delta_{j+}=\max\{\Delta_j,0\}$,
$\Delta_{j-}=\max\{-\Delta_j,0\}$. For each policy $j$, denote the one-policy minimax-regret value at continuation cost $c$
\begin{equation} \label{eqn:r_starj}
\small 
\begin{aligned} 
r_j^\star(c)
=
\inf_{(\delta_{1j},\delta_{2j})}
\sup_{|\Delta_j|\leq\bar\Delta_j}
\mathrm{Reg}
(\delta_{1j},\delta_{2j};\Delta_j,c). 
\end{aligned} 
\end{equation} 

\paragraph{A lower bound on $R_q^\star$} The first step is to prove the following lower bound. 
\begin{lem} \label{lem:boundaa} Suppose that Assumption \ref{ass:multiple} holds, with $\min_j c_j = \underline{c} > 0$ and let $\bar{s} = \max_j s_j$. For every $0 \le q \le \sum_j c_j$, and every fixed $\lambda \ge 0$, 
\begin{equation}
\small 
\begin{aligned} 
\label{eq:optimality_penalized_lower}
R_q^\star
\geq
\frac1J\sum_{j=1}^Jr_j^\star(\lambda c_j)
-\frac{\lambda q}{J}.
\end{aligned} 
\end{equation}
\end{lem} 
\begin{proof} 

Fix $\lambda\geq0$. Let 
$\mathcal D_\infty$ denote the set of measurable $\delta_1(X_1), \delta_2(X_1,X_2)$ without imposing any experimentation budget constraint. Note that $\mathcal{D}_q \subseteq \mathcal{D}_\infty$. 
We claim that
\begin{equation}
\label{eq:optimality_penalized_separation}
\begin{aligned}
&\underbrace{\inf_{\boldsymbol\delta\in\mathcal D_\infty}
\sup_{\boldsymbol\Delta\in
\prod_{j=1}^J[-\bar\Delta_j,\bar\Delta_j]}
\left\{
\overline{\mathrm{Reg}}(\boldsymbol\delta;\boldsymbol\Delta)
+
\frac{\lambda}{J}
\mathbb E_{\boldsymbol\Delta}
\left[
\sum_{j=1}^J
c_j1\{\delta_{1j}(\mathbf X_1)=\texttt{NA}\}
\right]
\right\}}_{(I)}
=
\frac1J\sum_{j=1}^Jr_j^\star(\lambda c_j).
\end{aligned}
\end{equation}
To prove the equality proving the weak inequality in both directions (upper and lower bounds). 
For the upper bound, let
$(\delta_{1j}^\lambda,\delta_{2j}^\lambda)$ be a one-policy
minimax rule in Lemma \ref{thm:dead_zone} at continuation cost $\lambda c_j$ (i.e., $\delta_{1j}^{\lambda} = \texttt{NA}$ if $|X_{1j}|/s_{j}  \le \tau^\star(1, \lambda c_j/s_{j})$ and one or zero otherwise depending on the sign of $X_{1j}$ and $\delta_{2j}^{\lambda} = 1\{X_{1j} + X_{2j} \ge 0\}$). It follows that 
$\boldsymbol\delta^\lambda\in\mathcal D_\infty$.
Consequently, since $\boldsymbol\delta^\lambda$ is feasible for the
infimum in Equation \eqref{eq:optimality_penalized_separation}
\[
\small 
\begin{aligned}
(I) \le &\sup_{\boldsymbol\Delta}
\left\{
\overline{\mathrm{Reg}}(\boldsymbol\delta^\lambda;\boldsymbol\Delta)
+
\frac{\lambda}{J}
\mathbb E_{\boldsymbol\Delta}
\left[
\sum_{j=1}^J
c_j1\{\delta_{1j}^\lambda(X_{1j})=\texttt{NA}\}
\right]
\right\}
=
\sup_{\boldsymbol\Delta}
\frac1J\sum_{j=1}^J
\mathrm{Reg}
(\delta_{1j}^\lambda,\delta_{2j}^\lambda;
\Delta_j,\lambda c_j)
\\
&\qquad=
\frac1J\sum_{j=1}^J
\sup_{|\Delta_j|\leq\bar\Delta_j}
\mathrm{Reg}
(\delta_{1j}^\lambda,\delta_{2j}^\lambda;
\Delta_j,\lambda c_j)
=
\frac1J\sum_{j=1}^Jr_j^\star(\lambda c_j).
\end{aligned}
\]
The second equality follows because each summand depends only
on $\Delta_j$ and the parameter space is rectangular.

For the reverse inequality, fix $\lambda\geq0$. Let
$\mathcal P([-\bar\Delta_j,\bar\Delta_j])$ denote the set of Borel
probability measures on $[-\bar\Delta_j,\bar\Delta_j]$. By the
minimax duality and existence result established in Appendix
\ref{proof:multi_period_dead_zone} (applied to the two-period special case), we may choose a least-favorable
prior
\[
\small 
\begin{aligned} 
\pi_j^\star
\in
\operatorname*{arg\,max}_{
\pi_j\in\mathcal P([-\bar\Delta_j,\bar\Delta_j])}
\inf_{(\delta_{1j},\delta_{2j})}
\int
\mathrm{Reg}
(\delta_{1j},\delta_{2j};\Delta_j,\lambda c_j)
\,\pi_j(d\Delta_j) 
\end{aligned} 
\]
so that $ 
\inf_{(\delta_{1j},\delta_{2j})}
\int
\mathrm{Reg}
(\delta_{1j},\delta_{2j};\Delta_j,\lambda c_j)
\,\pi_j^\star(d\Delta_j)
=
r_j^\star(\lambda c_j).
$ 
The priors may depend on $\lambda$, which is held fixed throughout
this argument. Consider their product
$ 
\pi^\star
=
\bigotimes_{j=1}^J\pi_j^\star
\in
\mathcal P\left(
\prod_{j=1}^J[-\bar\Delta_j,\bar\Delta_j]
\right).
$ 
By Assumption \ref{ass:multiple}, the likelihood of the
first-period observations factors across policies, with
$X_{1j}$ depending on the effect vector only through $\Delta_j$.
Bayes' rule therefore gives
$ 
\pi^\star(d\boldsymbol\Delta\mid\mathbf X_1)
=
\bigotimes_{j=1}^J
\pi_j^\star(d\Delta_j\mid X_{1j}).
$ 
For any first-period rule, let
$S=\{j:\delta_{1j}(\mathbf X_1)=\texttt{NA}\}$.
Conditional on $\mathbf X_1$, this set is fixed. Independence
of the follow-up observations across $j$ then implies that,
for every $j\in S$,
$ 
\pi^\star
\left(
d\Delta_j\mid
\mathbf X_1,(X_{2k})_{k\in S}
\right)
=
\pi_j^\star(d\Delta_j\mid X_{1j},X_{2j}).
$ 
Therefore we can write 
\[
\small 
\begin{aligned}
(I) \ge &\inf_{\boldsymbol\delta\in\mathcal D_\infty}
\int
\Bigg[
\overline{\mathrm{Reg}}(\boldsymbol\delta;\boldsymbol\Delta)
+
\frac{\lambda}{J}
\mathbb E_{\boldsymbol\Delta}
\left[
\sum_{j=1}^J
c_j1\{\delta_{1j}(\mathbf X_1)=\texttt{NA}\}
\right]
\Bigg]
\,\pi^\star(d\boldsymbol\Delta)
\\
&\qquad=
\frac1J\sum_{j=1}^J
\inf_{(\delta_{1j},\delta_{2j})}
\int
\mathrm{Reg}
(\delta_{1j},\delta_{2j};\Delta_j,\lambda c_j)
\,\pi_j^\star(d\Delta_j)
=
\frac1J\sum_{j=1}^Jr_j^\star(\lambda c_j).
\end{aligned}
\]
The first inequality follows from the fact that 
the supremum penalized regret over $\boldsymbol{\Delta}$ is at least
its integrated penalized regret under $\pi^\star$.
It follows that Equation \eqref{eq:optimality_penalized_separation} holds. 
Finally, every $\boldsymbol\delta\in\mathcal D_q$ also belongs
to $\mathcal D_\infty$ and satisfies the budget constraint $\sum_j c_j 1\{\delta_{1j}(\mathbf{X}_1) = \texttt{NA}\} \le q$
almost surely. Therefore,
\[
\small 
\begin{aligned}
\frac1J\sum_{j=1}^Jr_j^\star(\lambda c_j)
&\leq
\int
\Bigg[
\overline{\mathrm{Reg}}(\boldsymbol\delta;\boldsymbol\Delta)
+
\frac{\lambda}{J}
\mathbb E_{\boldsymbol\Delta}
\left[
\sum_{j=1}^J
c_j1\{\delta_{1j}(\mathbf X_1)=\texttt{NA}\}
\right]
\Bigg]
\,\pi^\star(d\boldsymbol\Delta)
\\
&\leq
\int
\overline{\mathrm{Reg}}(\boldsymbol\delta;\boldsymbol\Delta)
\,\pi^\star(d\boldsymbol\Delta)
+\frac{\lambda q}{J}
\leq
\sup_{\boldsymbol\Delta}
\overline{\mathrm{Reg}}(\boldsymbol\delta;\boldsymbol\Delta)
+\frac{\lambda q}{J}.
\end{aligned}
\]
Taking the infimum over $\boldsymbol{\delta} \in \mathcal{D}_q$ completes the claim. 
\end{proof} 

\paragraph{Properties of $A_j$} Before we continue we establish basic properties of $A_j$. 
By Proposition \ref{prop:c_value_properties},
$ 
A_j=\frac{s_j}{c_j}
\Gamma_1\!\left(\frac{|X_{1j}|}{s_j}\right).
$ 
By Lemma \ref{cor:1}(B),
$\tau^\star(1,\gamma)\to\infty$ as $\gamma\downarrow0$,
while Lemma \ref{lem:tau_cost_unique} gives
$\tau^\star(1,\gamma)=0$ whenever $\gamma>u_0$.
The definition of $\Gamma_1$ therefore implies
$0<\Gamma_1(t)\leq\Gamma_1(0)\leq u_0<\infty$
for every finite $t\geq0$.
In addition, by Lemma \ref{cor:1}(B), the threshold is continuous in cost, and by definition 
$ 
\tau^\star(1,\Gamma_1(t))=t, t \ge 0.
$ 
Since $\Gamma_1$ is nonincreasing by
Proposition \ref{prop:c_value_properties}, this identity
implies that $\Gamma_1$ is strictly decreasing.
Each $A_j$ therefore has a continuous distribution and are independent, so that we can ignore ties almost surely. In addition, by letting 
$ 
\Lambda=\Gamma_1(0)\frac{\bar s}{\underline c},
$ 
$0<A_j\leq\Lambda$. 

\paragraph{Selection rule}  Write $j \in j_1 ,\ldots,j_J$ for the policy indices in decreasing order
of $A_j$, and let $m$ be the number selected by the procedure.
Because $q<\sum_{j=1}^Jc_j$, we have $m<J$. Define
$\widehat\lambda=A_{j_{m+1}}$, the score of the first excluded
policy. The selected policies satisfy  $A_j>\widehat\lambda$. Moreover, by construction 
\begin{equation}
\small 
\begin{aligned} 
\label{eq:optimality_budget_remainder}
0\leq
q-\sum_{j=1}^Jc_j1\{A_j>\widehat\lambda\}
<c_{j_{m+1}}\leq1.
\end{aligned}
\end{equation}
Fix $\boldsymbol\Delta$.
For the analysis, generate all potential follow-up observations
according to Assumption \ref{ass:multiple}; the procedure uses
$X_{2j}$ only when policy $j$ is selected.
Define
\begin{equation} \label{eqn:U_j}
\small 
\begin{aligned}
U_j
&=\Delta_{j,+}1\{X_{1j}<0\}
  +\Delta_{j,-}1\{X_{1j}\geq0\},\\
B_j
&=\Delta_{j,+}1\{X_{1j}+X_{2j}<0\}
  +\Delta_{j,-}1\{X_{1j}+X_{2j}\geq0\}-U_j.
\end{aligned}
\end{equation} 
Thus, $U_j$ is the loss from an immediate sign decision and
$B_j$ is the change in loss from using both observations.
For each fixed $\lambda\in[0,\Lambda]$, Lemma
\ref{thm:dead_zone} and Theorem \ref{thm:a1} identify
the rule that continues when $A_j>\lambda$ as a one-policy
minimax rule at cost $\lambda c_j$. Its penalized regret
at any $\Delta_j$ is therefore at most
$r_j^\star(\lambda c_j)$. Summing across policies gives for any $\boldsymbol{\Delta}$
\begin{equation}
\label{eq:optimality_fixed_threshold}
\small 
\begin{aligned} 
\frac1J\sum_{j=1}^J
\left\{
\mathbb E_{\boldsymbol\Delta}[U_j]
+\mathbb E_{\boldsymbol\Delta}
\left[(B_j+\lambda c_j)1\{A_j>\lambda\}\right]
\right\} \le \frac{1}{J}\sum_{j=1}^J r_j^\star(\lambda c_j) \le 
R_q^\star + \frac{\lambda q}{J}
\end{aligned} 
\end{equation}
where the first inequality follows from the fact that each summand in the left corresponds to the regret of the minimax optimal rule for fixed $\Delta$, while $r_j^\star$ corresponds to the worst-case over $\Delta$ (since the rule that continues for $A_j > \lambda$ is minimax regret optimal by Lemma \ref{thm:dead_zone}); the second inequality follows from Lemma \ref{lem:boundaa}. 
By definition of $U_j$ we rewrite
\begin{equation} \label{eqn:hh1}
\small 
\begin{aligned}
\overline{\mathrm{Reg}}(\boldsymbol\delta^{A,q};\boldsymbol\Delta)
={}&\frac1J\sum_{j=1}^J\mathbb E_{\boldsymbol\Delta}[U_j]
+\mathbb E_{\boldsymbol\Delta}
\left[
\frac1J\sum_{j=1}^J
(B_j+\widehat\lambda c_j)1\{A_j>\widehat\lambda\}
-\frac{\widehat\lambda q}{J}
\right] +\frac1J\mathbb E_{\boldsymbol\Delta}
\left[
\widehat\lambda
\left(q-\sum_{j=1}^Jc_j1\{A_j>\widehat\lambda\}\right)
\right],
\end{aligned}
\end{equation} 
and define 
\[
\small 
\begin{aligned}
Z(\lambda)
=\frac1J\sum_{j=1}^J
\left\{
(B_j+\lambda c_j)1\{A_j>\lambda\}
-\mathbb E_{\boldsymbol\Delta}
\left[(B_j+\lambda c_j)1\{A_j>\lambda\}\right]
\right\}, \quad M(\lambda)
=\frac1J\sum_{j=1}^J
\mathbb E_{\boldsymbol\Delta}
\left[(B_j+\lambda c_j)1\{A_j>\lambda\}\right]
\end{aligned} 
\]
with fixed $\lambda$. 
Equation \eqref{eq:optimality_fixed_threshold} gives
$ 
\frac1J\sum_{j=1}^J\mathbb E_{\boldsymbol\Delta}[U_j]
+M(\lambda)-\frac{\lambda q}{J}
\leq R_q^\star
 \text{ for every }\lambda\in[0,\Lambda].
$ 
This deterministic inequality also holds when evaluated at
$\widehat\lambda$.
Thus, the first two terms in the right-hand side of Equation \eqref{eqn:hh1} are at most
$R_q^\star+\mathbb E_{\boldsymbol\Delta}[Z(\widehat\lambda)]$.

The remaining term depends on the unused budget.
By Equation \eqref{eq:optimality_budget_remainder}, this is
less than the cost of the first excluded policy, which is at most one since $c_j \le 1$ by assumption.
Since $\widehat\lambda\leq\Lambda$ almost surely,
$ 
\frac1J\mathbb E_{\boldsymbol\Delta}
\left[
\widehat\lambda
\left(q-\sum_{j=1}^Jc_j1\{A_j>\widehat\lambda\}\right)
\right]
\leq\frac{\Lambda}{J}.
$ 
Combining these bounds with
$Z(\widehat\lambda)\leq
\sup_{0\leq\lambda\leq\Lambda}|Z(\lambda)|$
gives
\begin{equation}
\small 
\begin{aligned} 
\label{eq:optimality_random_threshold}
\mathrm{Reg}(\boldsymbol\delta^{A,q};\boldsymbol\Delta)
-R_q^\star
\leq
\mathbb E_{\boldsymbol\Delta}
\left[\sup_{0\leq\lambda\leq\Lambda}|Z(\lambda)|\right]
+\frac{\Lambda}{J}.
\end{aligned} 
\end{equation}

\paragraph{Controlling the random cutoff}
We first record a bound for the centered threshold sums in
$Z(\lambda)$. For independent pairs $(A_j,H_j)$ with
$\mathbb E[H_j^2]<\infty$, introduce independent random signs
$\varepsilon_j$, taking values $-1$ and $1$ with equal
probabilities, independently of these pairs. A  standard symmetrization argument \citep[Lemma 6.4.2][]{vershynin2019high} gives 
\[
\small 
\begin{aligned}
&\mathbb E\left[
\sup_{0\leq\lambda\leq\Lambda}
\left|
\sum_{j=1}^J
\left\{H_j1\{A_j>\lambda\}
-\mathbb E[H_j1\{A_j>\lambda\}]\right\}
\right|
\right] \leq
2\mathbb E\left[
\sup_{0\leq\lambda\leq\Lambda}
\left|\sum_{j=1}^J\varepsilon_jH_j1\{A_j>\lambda\}\right|
\right].
\end{aligned}
\]
Next, condition on all the pairs
$(A_j,H_j)_{j=1}^J$, and let $j_1,\ldots,j_J$ order the
policies by decreasing $A_j$, breaking any ties using a fixed
ordering. Write $\mathbb E_\varepsilon$ for expectation over
the random signs conditional on these pairs, and define
$ 
S_0=0,
S_k=\sum_{\ell=1}^k\varepsilon_{j_\ell}H_{j_\ell},
 k=1,\ldots,J.
$ 
For every $\lambda$, the indices satisfying $A_j>\lambda$
form an initial segment of this ordering. Therefore,
$ 
\sup_{0\leq\lambda\leq\Lambda}
\left|
\sum_{j=1}^J\varepsilon_jH_j1\{A_j>\lambda\}
\right|
\leq
\max_{0\leq k\leq J}|S_k|.
$ 

Since the random signs $\varepsilon_j$ are independent
of the ordering, the reordered signs remain independent, with
mean zero and variance one. In particular,
$ 
\mathbb E_\varepsilon
[S_{k+1}\mid\varepsilon_{j_1},\ldots,\varepsilon_{j_k}]
=S_k,
k=0,\ldots,J-1.
$ 
Thus $(S_k)_{k=0}^J$ is a square-integrable martingale under
the conditional distribution of the signs. Doob's $L^2$
maximal inequality and Cauchy-Schwarz give 
\[
\small 
\begin{aligned}
\mathbb E_\varepsilon
\left[
\sup_{0\leq\lambda\leq\Lambda}
\left|
\sum_{j=1}^J\varepsilon_jH_j1\{A_j>\lambda\}
\right|
\right]
&\leq
\mathbb E_\varepsilon
\left[\max_{0\leq k\leq J}|S_k|\right] \leq
\left(
\mathbb E_\varepsilon
\left[\max_{0\leq k\leq J}|S_k|^2\right]
\right)^{1/2} \leq
2\left(\sum_{j=1}^JH_j^2\right)^{1/2}.
\end{aligned}
\]
Taking expectations over the pairs and applying Jensen's
inequality to the square-root function yields
$ 
\mathbb E
\left[
\sup_{0\leq\lambda\leq\Lambda}
\left|
\sum_{j=1}^J\varepsilon_jH_j1\{A_j>\lambda\}
\right|
\right]
\leq
2\mathbb E\left[
\left(\sum_{j=1}^JH_j^2\right)^{1/2}
\right] \leq
2\left(\sum_{j=1}^J\mathbb E[H_j^2]\right)^{1/2}.
$ 

For our losses, $B_j$ is nonzero only if the immediate or pooled
decision is incorrect. Hence\footnote{The immediate decision implements policy $j$ when
$X_{1j}\geq0$; the pooled decision averages both estimates and
implements when $(X_{1j}+X_{2j})/2\geq0$.
Each decision incurs loss $|\Delta_j|$ if incorrect and zero otherwise.
Thus $B_j$, the pooled loss minus the immediate loss, can be
nonzero only when exactly one decision is incorrect.}
\[
\begin{aligned}
\mathbb E_{\boldsymbol\Delta}[B_j^2]
&\leq
\Delta_j^2
\left\{
\Phi\!\left(-\frac{|\Delta_j|}{s_j}\right)
+\Phi\!\left(-\frac{\sqrt2|\Delta_j|}{s_j}\right)
\right\}
\leq C s_j^2,
\end{aligned}
\]
for a finite constant $C < \infty$
where the last inequality follows because
$u^2\Phi(-u)$ is bounded on $[0,\infty)$. Applying the preceding
threshold bound with $H_j=B_j$ and $H_j=c_j$, respectively,
and using $\lambda\leq\Lambda$, gives
\[
\begin{aligned}
\mathbb E_{\boldsymbol\Delta}
\left[\sup_{0\leq\lambda\leq\Lambda}|Z(\lambda)|\right]
&\leq
\frac4J
\left(\sum_{j=1}^J\mathbb E_{\boldsymbol\Delta}[B_j^2]\right)^{1/2}
+\frac{4\Lambda}{J}\left(\sum_{j=1}^Jc_j^2\right)^{1/2} \leq C\frac{\bar s}{\underline c\sqrt J},
\end{aligned}
\]
since $c_j\leq1$ and $\underline c\leq1$. Combining this with
Equation \eqref{eq:optimality_random_threshold}, proves
$ 
\sup_{\boldsymbol\Delta}
\overline{\mathrm{Reg}}(\boldsymbol\delta^{A,q};\boldsymbol\Delta)
-R_q^\star
\leq C\frac{\bar s}{\underline c\sqrt J}.
$ 

\paragraph{A bound for small budgets} We now complete the proof showing a bound of order $q/J$. To do so note that at continuation cost
$\Lambda c_j$, the one-policy minimax rule in Lemma \ref{thm:dead_zone} (and Theorem \ref{thm:a1}) makes
an immediate decision using the sign of $X_{1j}$. Consequently,
$ 
r_j^\star(\Lambda c_j)
=
\sup_{|\Delta_j|\leq\bar\Delta_j}
|\Delta_j|\Phi\!\left(-\frac{|\Delta_j|}{s_j}\right)
=
u_0\Phi(-u_0)s_j,
$ 
where the last equality uses
$\bar\Delta_j/s_j>0.76>u_0$.
Applying Equation \eqref{eq:optimality_penalized_lower}
at $\lambda=\Lambda$ gives
\begin{equation} \label{eqn:lla}
\small 
\begin{aligned} 
R_q^\star
\geq
\frac{u_0\Phi(-u_0)}{J}\sum_{j=1}^Js_j
-\frac{\Lambda q}{J}.
\end{aligned} 
\end{equation} 
We next bound the regret of the proposed procedure.
For $\Delta_j\neq0$, let $E_{2j}$ denote the event that the
pooled decision is incorrect so that $\mathrm{Sign}(X_{1j} + X_{2j}) \neq \mathrm{Sign}(\Delta_j)$. Since $U_j$ is determined by
$\mathbf X_1$, Equation \eqref{eqn:U_j} gives
$ 
\mathbb E_{\boldsymbol\Delta}[B_j\mid\mathbf X_1]
=
|\Delta_j|\,
\mathbb P_{\boldsymbol\Delta}(E_{2j}\mid\mathbf X_1)-U_j.
$ 
If the immediate decision is incorrect, then
$U_j=|\Delta_j|$, so this expression is nonpositive.

If the immediate decision is correct, then $U_j=0$ from Equation \eqref{eqn:U_j}.
Conditional on $\mathbf X_1$, independence implies
$ 
X_{1j}+X_{2j}| \mathbf{X}_1
\sim\mathcal N(X_{1j}+\Delta_j,s_j^2).
$ 
Since $X_{1j}$ has the correct sign, the conditional mean
has the sign of $\Delta_j$ and magnitude
$|X_{1j}|+|\Delta_j|$. Hence,
$ 
\mathbb P_{\boldsymbol\Delta}(E_{2j}\mid\mathbf X_1)
=
\Phi\!\left(-\frac{|X_{1j}|+|\Delta_j|}{s_j}\right)
\leq
\Phi\!\left(-\frac{|\Delta_j|}{s_j}\right).
$ 
Thus, in both cases,
\[
\mathbb E_{\boldsymbol\Delta}[B_j\mid\mathbf X_1]
\leq
|\Delta_j|\Phi\!\left(-\frac{|\Delta_j|}{s_j}\right)
\leq u_0\Phi(-u_0)s_j,
\]
where the last inequality uses
$\sup_{u\geq0}u\Phi(-u)=u_0\Phi(-u_0)$.
For $\Delta_j=0$, the bound holds because $B_j=0$.
The selected set $\mathcal{J}_q$ is measurable with respect to $\mathbf X_1$
and contains at most $q/\underline c$ elements by construction of the budget constraint. 
It follows that
\[
\small 
\begin{aligned}
\overline{\mathrm{Reg}}(\boldsymbol\delta^{A,q};\boldsymbol\Delta)
&=
\frac1J\sum_{j=1}^J
\mathbb E_{\boldsymbol\Delta}[U_j]
+
\frac1J\mathbb E_{\boldsymbol\Delta}
\left[
\sum_{j=1}^J
1\{j\in\mathcal J_q(\mathbf X_1)\}
\mathbb E_{\boldsymbol\Delta}[B_j\mid\mathbf X_1]
\right]
\\
&\leq
\frac{u_0\Phi(-u_0)}{J}\sum_{j=1}^Js_j
+
\frac{u_0\Phi(-u_0)\bar s}{J}
\mathbb E_{\boldsymbol\Delta}
\left[|\mathcal J_q(\mathbf X_1)|\right]
 \leq
\frac{u_0\Phi(-u_0)}{J}\sum_{j=1}^Js_j
+
\frac{u_0\Phi(-u_0)\bar s q}{\underline c J}.
\end{aligned}
\]
Combining this upper bound with the preceding lower bound
on $R_q^\star$ in Equation \eqref{eqn:lla}, and substituting the definition of $\Lambda$,
yields
$\sup_{\boldsymbol\Delta}
\overline{\mathrm{Reg}}(\boldsymbol\delta^{A,q};\boldsymbol\Delta)
-R_q^\star
\leq
\left\{u_0\Phi(-u_0)+\Gamma_1(0)\right\}
\frac{\bar s q}{\underline c J}
 \leq
C\frac{\bar s}{\underline c}\frac{q}{J},
$ 
where $C$ is a universal constant. 

\paragraph{Relative optimality} We conclude the proof with the relative optimality statement. 
Apply Lemma \ref{lem:boundaa} with $\lambda=0$ to obtain
$ 
R_q^\star\geq\frac1J\sum_{j=1}^J r_j^\star(0).
$ 
We now evaluate these one-policy values.
By Lemmas \ref{thm:dead_zone} and \ref{cor:1}(B),
at zero continuation cost the minimax rule always collects
the follow-up observation and uses the sign of the pooled mean.
Since
$(X_{1j}+X_{2j})/2\sim\mathcal N(\Delta_j,s_j^2/2)$,
its regret at $\Delta_j$ is
$|\Delta_j|\Phi(-\sqrt{2}|\Delta_j|/s_j)$.
Consequently,
$ 
r_j^\star(0)
=
\max_{|\Delta_j|\leq\bar\Delta_j}
|\Delta_j|\Phi\!\left(-\frac{\sqrt{2}|\Delta_j|}{s_j}\right)
=
\frac{s_j}{\sqrt{2}}u_0\Phi(-u_0).
$ 
The maximum occurs at
$|\Delta_j|=u_0s_j/\sqrt{2}$, which belongs to the parameter
space by Assumption \ref{ass:multiple}. Therefore,
$ 
R_q^\star
\geq
\frac{u_0\Phi(-u_0)}{\sqrt{2}}
\left(\frac1J\sum_{j=1}^Js_j\right)>0.
$ 
Dividing the previously established excess-regret bound by
the preceding lower bound gives
\[
\small 
\begin{aligned}
\frac{\sup_{\boldsymbol\Delta}
\overline{\mathrm{Reg}}(\boldsymbol\delta^{A,q};\boldsymbol\Delta)
-R_q^\star}{R_q^\star} \leq
\frac{\sqrt{2}C}
{u_0\Phi(-u_0)\underline c\sqrt J}
\frac{\bar s}{J^{-1}\sum_{j=1}^Js_j}
=\mathcal O(J^{-1/2}),
\end{aligned}
\]
provided $\underline c$ is bounded away from zero and
$\bar s/(J^{-1}\sum_j s_j)=\mathcal O(1)$.

\subsection{Proof of Theorem \ref{prop:top_q_c_value}}
\label{app:proof:top_q_c_value}

\paragraph{A one-policy bound} We will use the same notation in the proof of Theorem \ref{app:proof:optimality1} with 
$ 
r_j^\star(c)$ as defined in Equation \eqref{eqn:r_starj}. 
Evaluating a minimax rule for the zero-cost problem at continuation
cost $\lambda c_j$ gives, for every $\lambda\geq0$,
\begin{equation}
\label{eq:budget_value_bound}
r_j^\star(\lambda c_j)
\leq r_j^\star(0)+\lambda c_j,
\end{equation}
because its probability of continuation is at most one.

Fix $\lambda\geq0$ and consider the one-policy rule $\delta_{1j}^{\lambda}$ as defined in the proof of Theorem \ref{app:proof:optimality1}, so that $\delta_{1j}^{\lambda}(X_{1j}) = \texttt{NA}$ if $A_j > \lambda$ and otherwise equal one or zero depending on the sign of $X_{1j}$, with $\delta_{2j}^{\lambda} = 1\{X_{1j} + X_{2j} \ge 0\}$. 
Since $A_j=A_{\mathrm{value}}(X_{1j};s_j, w = 1)/c_j$,
Lemma \ref{thm:dead_zone} and Theorem \ref{thm:a1} imply that this
rule is minimax at continuation cost $\lambda c_j$.
Consequently, for every $|\Delta_j|\leq\bar\Delta_j$,
\[
\begin{aligned}
\mathrm{Reg}
\left(
\delta_{j}^\lambda;
\Delta_j,0
\right)
&=
\mathrm{Reg}
\left(
\delta_{j}^\lambda;
\Delta_j,\lambda c_j
\right)
-\lambda c_j\mathbb P_{\Delta_j}(A_j>\lambda)\\
&\leq
r_j^\star(\lambda c_j)
-\lambda c_j\mathbb P_{\Delta_j}(A_j>\lambda) \leq
r_j^\star(0)
+\lambda c_j\mathbb P_{\Delta_j}(A_j\leq\lambda),
\end{aligned}
\]
where the last inequality follows from
Equation \eqref{eq:budget_value_bound}.
Because $A_j$ follows a continuous distribution and are independent (see Paragraph ``Properties of $A_j$'' in the proof of Theorem \ref{app:proof:optimality1}) the case where $A_j = \lambda$ has measure zero and we can break ties arbitrary. 
Thus, we can write 
\begin{equation}
\label{eq:budget_one_policy}
\small 
\begin{aligned} 
\mathrm{Reg}
\left(
\delta_{1j}^\lambda;
\Delta_j,0
\right)
\leq
r_j^\star(0)
+\lambda c_j\mathbb P_{\Delta_j}
\left(
\delta_{1j}^\lambda(X_{1j})\neq\texttt{NA}
\right).
\end{aligned} 
\end{equation}

\paragraph{Relaxed problem}
Let $\mathcal D_\infty$ denote the class of strategies with the
information structure specified in the definition of
$\mathcal D_q$, but without the budget constraint. 
By Lemma \ref{lem:boundaa} with $\lambda = 0$, we have 
\begin{equation}
\small 
\begin{aligned} 
\label{eq:budget_separation}
R_\infty^\star
=
\inf_{\boldsymbol\delta\in\mathcal D_\infty}
\sup_{\boldsymbol\Delta\in
\prod_{j=1}^J[-\bar\Delta_j,\bar\Delta_j]}
\overline{\mathrm{Reg}}(\boldsymbol\delta;\boldsymbol\Delta)
=
\frac1J\sum_{j=1}^Jr_j^\star(0).
\end{aligned} 
\end{equation}
The inequality $\geq$ follows from Lemma \ref{lem:boundaa}
with $\lambda=0$ and $q=\sum_j c_j$, since
$\mathcal D_\infty=\mathcal D_{\sum_j c_j}$.
For the reverse inequality, let $\boldsymbol\delta^0$
apply a zero-cost one-policy minimax rule separately to each
policy. Then $\boldsymbol\delta^0\in\mathcal D_\infty$, and
rectangularity of the parameter space gives
$ 
R_\infty^\star
\leq
\sup_{\boldsymbol\Delta}
\overline{\mathrm{Reg}}
(\boldsymbol\delta^0;\boldsymbol\Delta)
=
\frac1J\sum_{j=1}^J r_j^\star(0).
$ 

\paragraph{Capacity constrained problem}
Fix $q\in\{1,\ldots,J-1\}$ with $q\leq\sum_jc_j$ and
$\boldsymbol\Delta\in
\prod_{j=1}^J[-\bar\Delta_j,\bar\Delta_j]$.
Let $\mathcal J_q(\mathbf X_1)$ be as in Definition \ref{defn:top_A}.

For each policy $j$, define its realized regret by with $\delta_{2j}^\star = 1\{X_{1j} + X_{2j} \ge 0\}$
\[
\small 
\begin{aligned}
L_j^{A,q}
={}&
\Delta_{j,+}1\{\delta_{1j}^{A,q}=0\}
+\Delta_{j,-}1\{\delta_{1j}^{A,q}=1\} +
1\{\delta_{1j}^{A,q}=\texttt{NA}\}
\left[
\Delta_{j,+}1\{\delta_{2j}^\star=0\}
+\Delta_{j,-}1\{\delta_{2j}^\star=1\}
\right].
\end{aligned}
\]
Thus,
$ 
\mathrm{Reg}
(\boldsymbol\delta^{A,q};\boldsymbol\Delta)
=\frac1J\sum_{j=1}^J
\mathbb E_{\boldsymbol\Delta}[L_j^{A,q}].
$ 

Fix policy $j$ and define
$ 
\lambda_j
=
\inf\left\{
\lambda\geq0:
\sum_{k\neq j}c_k1\{A_k>\lambda\}
\leq q-c_j
\right\}.
$ 
This is the smallest cutoff for which the total cost of other
policies with strictly larger $A$-values leaves room for
policy $j$. The defining set is nonempty because
$q\geq1\geq c_j$. Its infimum is attained, since the sum is a
right-continuous, nonincreasing step function of $\lambda$.
Moreover, $\lambda_j\geq0$ is measurable with respect to
$\mathbf X_{1,-j}$.

If $A_j>\lambda_j$, every policy preceding $j$ in the ranking
has $A_k\geq A_j>\lambda_j$. Their total cost is therefore at
most $q-c_j$, so policy $j$ is selected. If $A_j<\lambda_j$,
the policies satisfying $A_k>A_j$ have total cost strictly
greater than $q-c_j$. They all precede $j$, so policy $j$
is not selected. Hence,
$ 
\{A_j>\lambda_j\}
\subseteq
\{j\in\mathcal J_q(\mathbf X_1)\}
\subseteq
\{A_j\geq\lambda_j\}.
$ 
Conditional on $\mathbf X_{1,-j}$, the cutoff $\lambda_j$ is
fixed and $(X_{1j},X_{2j})$ retains its marginal distribution
under $\Delta_j$. The conditional rule for policy $j$ is
therefore the one-policy rule with cutoff $\lambda_j$. Consequently,
$ 
\mathbb E_{\boldsymbol\Delta}
[L_j^{A,q}\mid\mathbf X_{1,-j}]
\leq
r_j^\star(0)
+c_j\lambda_j
\mathbb P_{\boldsymbol\Delta}
\left(
j\notin\mathcal J_q(\mathbf X_1)
\mid\mathbf X_{1,-j}
\right)
$ 
almost surely. Taking expectations, summing over $j$, and
using Equation \eqref{eq:budget_separation} gives
\begin{equation}
\label{eq:budget_cutoff_bound}
\small 
\begin{aligned} 
\overline{\mathrm{Reg}}
(\boldsymbol\delta^{A,q};\boldsymbol\Delta)
-R_\infty^\star
\leq
\frac1J\mathbb E_{\boldsymbol\Delta}
\left[
\sum_{j\notin\mathcal J_q(\mathbf X_1)}c_j\lambda_j
\right].
\end{aligned} 
\end{equation}

We next bound the sum on the right-hand side. For every
realization, at most $q-1$ policies have $A_k>A_{(q)}$.
Since $c_k\leq1$ and $c_j\leq1$,
$ 
\sum_{k\neq j}c_k1\{A_k>A_{(q)}\}
\leq q-1\leq q-c_j.
$ 
Thus $A_{(q)}$ belongs to the set defining $\lambda_j$, and
$ 
\lambda_j\leq A_{(q)}
\text{ for every }j.
$ 
Also, the first $q$ policies in the ranking have total cost
at most $q$, so all of them are selected. Therefore,
$|\mathcal J_q(\mathbf X_1)|\geq q$, and
$ 
\sum_{j\notin\mathcal J_q(\mathbf X_1)}c_j\lambda_j
\leq
A_{(q)}
\sum_{j\notin\mathcal J_q(\mathbf X_1)}c_j \leq
A_{(q)}\bigl(J-|\mathcal J_q(\mathbf X_1)|\bigr)
\leq (J-q)A_{(q)}.
$ 
Combining this inequality with
Equation \eqref{eq:budget_cutoff_bound} proves the regret bound.

\paragraph{Maximization of expectation}
It remains to show that the expectation of $A_{(q)}$ is
maximized under $\boldsymbol\Delta=\mathbf0$.
For $d_j=|\Delta_j|/s_j$ and every $t\geq0$,
$ 
\mathbb P_{\Delta_j}
\left(\frac{|X_{1j}|}{s_j}\leq t\right)
=\Phi(t-d_j)-\Phi(-t-d_j).
$ 
The right-hand side is nonincreasing in $d_j$, since its
derivative is
$ 
-\phi(t-d_j)+\phi(t+d_j)\leq0.
$ 
Thus, $|X_{1j}|/s_j$ is stochastically smallest when
$\Delta_j=0$. Proposition \ref{prop:c_value_properties} shows
that $A_{\mathrm{value}}^\star(X_{1j};s_j)$ is weakly
decreasing in $|X_{1j}|/s_j$. Since $c_j>0$ is fixed, the
cost-adjusted value $A_j$ is therefore stochastically largest
under $\Delta_j=0$.

Independence across policies and the fact that the $q$-th
largest order statistic is increasing in every coordinate imply
$ 
\mathbb E_{\boldsymbol\Delta}[A_{(q)}]
\leq
\mathbb E_{\boldsymbol\Delta=\mathbf0}[A_{(q)}]
$ 
for every admissible $\boldsymbol\Delta$.

\subsection{Proof of Corollary \ref{prop:minimax}}
\label{proof:prop:minimax}

We use the notation $u_0$, $r_j^\star(c)$, and
$\underline c=\min_j c_j$ from the proof of
Theorem \ref{thm:optimality1}.
Equation \eqref{eq:budget_separation}, together with the
calculation of $r_j^\star(0)$ in the paragraph titled ``Relative optimality'' in the proof of Theorem \ref{thm:optimality1}, gives (where recall $s_j = a_j/\sqrt{n}$) 
\begin{equation}
\label{eq:linear_capacity_full_value}
\small 
\begin{aligned} 
R_\infty^\star
=
\frac1J\sum_{j=1}^J r_j^\star(0)
=
\frac{u_0\Phi(-u_0)}{\sqrt{2n}\,J}
\sum_{j=1}^J a_j.
\end{aligned} 
\end{equation}

\paragraph{Lower bound} We next establish the lower bound. We start from applying Lemma \ref{lem:boundaa} with multiplier
$\lambda/\sqrt n$ and using
Equation \eqref{eq:linear_capacity_full_value} to obtain 
\[
\small 
\begin{aligned}
R_q^\star-R_\infty^\star
\geq
\frac1J\sum_{j=1}^J
\left[
r_j^\star\!\left(\frac{\lambda c_j}{\sqrt n}\right)
-r_j^\star(0)
\right]
-\frac{\lambda q}{J\sqrt n}.
\end{aligned}
\]
We now give a lower bound on the right-hand side. To do so, first note that $r_j^\star$ minimizes over $\Delta_j$ so that it can be lower-bounded by one particular choice of prior over $\Delta_j$. We construct such priors as follows.  
For each policy $j$, consider the prior assigning equal
probabilities to $\Delta_j=\pm u_0s_j/\sqrt2$.
Its support lies in the parameter interval by the maintained
assumption. With both observations available for free, the
Bayes rule uses the sign of $X_{1j}+X_{2j}$. At either prior
support point, its probability of error is $\Phi(-u_0)$,
so its Bayes risk equals
$ 
\frac{u_0s_j}{\sqrt2}\Phi(-u_0)=r_j^\star(0).
$ 
Consequently, every one-policy rule
$\delta_j=(\delta_{1j},\delta_{2j})$ satisfies
\[
\small 
\begin{aligned} 
\frac12\sum_{\eta\in\{-1,1\}}
\mathrm{Reg}\!\left(
\delta_j;\frac{\eta u_0s_j}{\sqrt2},0
\right)
\geq r_j^\star(0).
\end{aligned} 
\]
Now fix $\lambda>0$ and consider the one-policy minimax optimal rule
in Lemma \ref{thm:dead_zone} at continuation cost
$\lambda c_j/\sqrt n$. Since $s_j=a_j/\sqrt n$, it continues
when
$ 
\frac{|X_{1j}|}{s_j}
<
\tau^\star\!\left(1,\frac{\lambda c_j}{a_j}\right).
$ 
Under either prior support point, $|X_{1j}|/s_j$ has the
distribution of $|Z|$, where
$Z\sim\mathcal N(u_0/\sqrt2,1)$.
The rule's worst-case penalized regret is at least its
average penalized regret under this prior. The latter is
its average unpenalized regret plus the continuation cost
times its continuation probability. Hence,
\[
\small 
\begin{aligned} 
r_j^\star\!\left(\frac{\lambda c_j}{\sqrt n}\right)
\geq \frac12\sum_{\eta\in\{-1,1\}}
\mathrm{Reg}\!\left(
\delta_j;\frac{\eta u_0s_j}{\sqrt2},0
\right) + \frac{\lambda c_j}{\sqrt n}
\mathbb P\!\left(
|Z|<
\tau^\star\!\left(1,\frac{\lambda c_j}{a_j}\right)
\right) \ge 
r_j^\star(0)
+
\frac{\lambda c_j}{\sqrt n}
\mathbb P\!\left(
|Z|<
\tau^\star\!\left(1,\frac{\lambda c_j}{a_j}\right)
\right).
\end{aligned} 
\]
Substituting into the preceding bound on
$R_q^\star-R_\infty^\star$ yields \\ 
$ 
R_q^\star-R_\infty^\star
\geq
\frac{\lambda}{J\sqrt n}
\left[
\sum_{j=1}^J c_j
\mathbb P\!\left(
|Z|<
\tau^\star\!\left(1,\frac{\lambda c_j}{a_j}\right)
\right)-q
\right].
$ 

By Lemma \ref{cor:1}(B), the continuation probabilities
converge to one as $\lambda\downarrow0$. The bracket
therefore converges to $\sum_jc_j-q>0$.
Choose a sufficiently small positive $\lambda$ and hold it
fixed as $n$ varies. Since the probabilities depend on
$\lambda$, $a_j$, and $c_j$, but not on $n$, the last bound
equals $C/\sqrt n$ for a constant $C\in(0,\infty)$ depending
only on the fixed quantities in the statement.
Every $\boldsymbol\delta\in\mathcal D_q$ has worst-case
regret at least $R_q^\star$, proving the claimed lower bound.

\paragraph{Upper bound}
The argument in the paragraph ``A bound for small budgets''
in the proof of Theorem \ref{thm:optimality1} gives
$ 
\sup_{\boldsymbol\Delta}
\overline{\mathrm{Reg}}
(\boldsymbol\delta^{A,q};\boldsymbol\Delta)
\leq
\frac{u_0\Phi(-u_0)}{J}
\left(
\sum_{j=1}^J s_j
+
\frac{q}{\underline c}\max_j s_j
\right).
$ 
Substituting $s_j=a_j/\sqrt n$ and subtracting
Equation \eqref{eq:linear_capacity_full_value}, we obtain \\ 
$\sup_{\boldsymbol\Delta}
\overline{\mathrm{Reg}}
(\boldsymbol\delta^{A,q};\boldsymbol\Delta)
-R_\infty^\star \leq
\frac{u_0\Phi(-u_0)}{J\sqrt n}
\left[
\left(1-\frac1{\sqrt2}\right)\sum_{j=1}^J a_j
+
\frac{q}{\underline c}\max_j a_j
\right]
=
\frac{C'}{\sqrt n},
$ 
where $C'\in(0,\infty)$ depends only on the fixed
quantities in the statement.

\subsection{Explicit characterization of the least favorable prior} \label{app:least_favorable}

Here we provide an explicit characterization of the least favorable prior in the relevant setting where $w = s_2/s_1 \ge 0.63$ for an example (more general cases for $w$ are possible to study and imply potentially non unique least favorable priors; we omit this case for brevity, while the proof of Theorem \ref{thm:multi_period_dead_zone} covers all such cases). 

\begin{prop}[Least-favorable priors in the two-period model]
\label{prop:least_favorable_two_periods}
Let Assumptions~\ref{ass:experiment} and~\ref{ass:2} hold, and
suppose $w\geq0.63$. Write $\gamma=c/s_1$ and let
\[
\mathcal M^\star
=\operatorname*{arg\,max}_{d\in[0,u_0]}
\frac{1}{s_1} \mathrm{Reg}(\delta^{\star},s_1 d,\gamma),
\]
with $\delta^\star$ as defined in Lemma \ref{thm:dead_zone},  where $u_0\simeq0.7518$ solves
$\Phi(-u_0)=u_0\phi(u_0)$.
Then $\mathcal M^\star=\{d^\star\}$ for some $d^\star>0$.
The least-favorable prior is unique and satisfies
\[
\pi^\star\bigl(\{s_1d^\star\}\bigr)
=\pi^\star\bigl(\{-s_1d^\star\}\bigr)
=\frac12.
\]
\end{prop}

\begin{proof}
We first establish uniqueness of the worst-case magnitude. We use the same notation as in the proof of Lemma \ref{cor:1}. 
Write $L(\tau,d)=d\,p(d,\tau)$ and
$a=\sqrt{1+w^{-2}}$. Since $w\geq0.63$, we have
$a^2u_0^2<2$. 
For finite $\tau\geq0$, equation~\eqref{eq:tau_cost_Ldd}
in the proof of part~(1) of Lemma~\ref{lem:tau_cost_curvature}
gives
\[
L_{dd}(\tau,d)
=
\int_{[0,1/w^2]}
\lambda(s)\phi(\lambda(s)d)
\bigl[(\lambda(s)d)^2-2\bigr]\,dH(s;\tau,w)
<0,
\qquad 0\leq d\leq u_0.
\]
Here $H$ is a probability measure and
$1\leq\lambda(s)\leq a$ by \eqref{eqn:lambda_s}.
The weak inequality in that lemma is strict here because
$(\lambda(s)d)^2\leq a^2u_0^2<2$ throughout the integration
domain. The identity also applies at $d=0$ by continuity.

Next, the calculation in the proof of part~(2) of
Lemma~\ref{lem:tau_cost_zeros}, obtained by differentiating
\eqref{eq:tau_cost_qderiv}, yields
\[
q_{dd}(d,\tau)
=
-\bigl[\tau-d\tanh(\tau d)\bigr]
 \bigl[\phi(\tau-d)+\phi(\tau+d)\bigr]\leq0.
\]
Indeed,
$\tau-d\tanh(\tau d)\geq\tau(1-d^2)\geq0$
for $0\leq d\leq u_0<1$.
The same identity holds at $\tau=0$, where $q_{dd}(d,0)=0$.
Therefore $R_{dd}=L_{dd}+\gamma q_{dd}<0$.

As noted near the end of part~(A) of the proof of
Lemma~\ref{cor:1}, the always-continuation rule $\tau=\infty$
has $p(d,\infty)=\Phi(-ad)$ and $q(d,\infty)=1$.
Substituting these expressions into \eqref{eqn:norm_regret}
and differentiating twice with respect to $d$ gives
\[
R(\infty,d,\gamma)=d\Phi(-ad)+\gamma,
\qquad
R_{dd}(\infty,d,\gamma)
=a\phi(ad)(a^2d^2-2)<0
\]
for $0\leq d\leq u_0$, since $a^2u_0^2<2$.
Together with the finite-threshold argument above, which
includes $\tau=0$, this establishes strict concavity of
$R(\tau^\star,\cdot,\gamma)$ on $[0,u_0]$ for every possible
value of $\tau^\star$. Continuity and compactness ensure
that a maximizer exists, and strict concavity makes it
unique. Thus $\mathcal M^\star=\{d^\star\}$.

The derivative calculation at $d=0$ in the proof of the
last assertion of Lemma~\ref{lem:tau_cost_curvature} gives
$R_d(\tau,0,\gamma)=1/2$ for finite $\tau$; the displayed
formula for $R(\infty,d,\gamma)$ gives the same derivative
at $\tau=\infty$. Hence $d^\star>0$.

Moreover, the argument in part~(A) of the proof of
Lemma~\ref{cor:1} gives strict decrease for $d>u_0$:
for every $\lambda\geq1$,
$ 
\frac{\partial}{\partial d}\{d\Phi(-\lambda d)\}
=\Phi(-\lambda d)-\lambda d\phi(\lambda d)<0,
\text{ for } d>u_0.
$ 
Integrating this inequality in \eqref{eqn:mixture_a} and
using $q_d\leq0$ yields $R_d(\tau,d,\gamma)<0$ for finite
$\tau$ and $d>u_0$. The same conclusion follows directly
from $R(\infty,d,\gamma)=d\Phi(-ad)+\gamma$.
Since $\bar\Delta/s_1\geq0.76>u_0$, the only worst-case
effects in the original parameter space are
$\pm s_1d^\star$.

Let $V$ denote minimax regret and let
$\boldsymbol\delta^\star$ denote the canonical minimax rule.
Existence of a least-favorable prior follows from the proof
of Theorem~\ref{thm:multi_period_dead_zone}. For any such prior $\pi$,
minimax duality gives
\[
\begin{aligned}
V
&=\inf_{\boldsymbol\delta}
  \int \mathrm{Reg}(\boldsymbol\delta;\Delta,c)\,\pi(d\Delta) \leq
  \int \mathrm{Reg}(\boldsymbol\delta^\star;\Delta,c)
       \,\pi(d\Delta)
\leq V.
\end{aligned}
\]
Consequently, $\boldsymbol\delta^\star$ is Bayes under $\pi$.
The nonnegative, continuous function
$V-\mathrm{Reg}(\boldsymbol\delta^\star;\Delta,c)$
also has zero integral under $\pi$. Therefore every
least-favorable prior is supported on
$\{-s_1d^\star,s_1d^\star\}$.

It remains to establish equal probabilities. Write $\alpha=\pi(\{s_1d^\star\})$.
Point masses have zero Bayes regret, whereas $V>0$ by the
nondegeneracy argument in the proof of Theorem \ref{thm:multi_period_dead_zone}, so $0<\alpha<1$.
Using the Gaussian likelihood calculation in its
\emph{Posterior policy decision} paragraph, the Bayes terminal
cutoff for $Y_h=\sum_{r=1}^h X_r/s_r^2$ is
 $ 
\frac{1}{2s_1d^\star}\log\frac{1-\alpha}{\alpha}.
$ 
Take $h=1$ if $\tau^\star=0$ and $h=2$ otherwise.
This period is reached with positive probability, and $Y_h$
conditional on reaching it has a strictly positive density on
$\mathbb R$. 

If $\alpha\neq1/2$, the Bayes terminal cutoff differs from
zero. The minimax rule and the Bayes terminal rule then
choose different actions whenever $Y_h$ lies strictly
between their two cutoffs. This event has positive
probability, and choosing the Bayes action on this event
would strictly reduce prior-average regret. This contradicts
the fact that the minimax rule is Bayes optimal under
the least-favorable prior $\pi$. Therefore $\alpha=1/2$,
establishing uniqueness.
\end{proof}

\section{Proofs for the Appendix}

\subsection{Proof of Theorem \ref{thm:asymmetric_dead_zone}}
\label{proof:asymmetric_dead_zone}
The proof of the theorem follows verbatim as the proof of 
Theorem~\ref{thm:multi_period_dead_zone}
(Appendix~\ref{proof:multi_period_dead_zone}), where we need to further adjust for the asymmetry component. We show this in the following lines. 

The existence and Bayes-characterization arguments in the proof
of Theorem~\ref{thm:multi_period_dead_zone}
(Appendix~\ref{proof:multi_period_dead_zone}) apply with $H=2$
and asymmetric regret: fixed $K>0$ preserves the boundedness
and continuity properties used there.
Thus a minimax strategy and a least-favorable prior $\pi^\star$
exist, and every minimax strategy is Bayes under $\pi^\star$.
Moreover,
$\mathrm{Reg}_{1,K}\geq\min\{1,K\}\mathrm{Reg}_{1,1}$,
so the positive-value argument in that proof applies.
A prior supported on only one side of zero has Bayes regret
zero; hence $\pi^\star$ assigns positive mass to both signs.

Retain $S_h$ and $Y_h=S_hT_h$ from that proof.
Using the posterior formula following \eqref{eqn:M_h}
with this prior, define
$ 
m_h(y)
=\mathbb E_{\pi^\star}
[\Delta_+-K\Delta_-\mid Y_h=y],
h=1,2.
$ Differentiating as in the posterior-policy
calculation gives
$ 
m_h'(y)
=\operatorname{Cov}_{\pi^\star}
(\Delta_+-K\Delta_-,\Delta\mid Y_h=y)>0,
$ 
because $\Delta\mapsto\Delta_+-K\Delta_-$ is strictly increasing
and the posterior is nondegenerate.
Exponential tilting toward the endpoints of the prior support
gives negative and positive limits as $y\to-\infty$ and
$y\to\infty$, respectively.
Thus $m_h$ has a unique finite zero $z_h$, and the Bayes action
upon stopping is $1\{Y_h\geq z_h\}$.

The likelihood-ratio calculation under the paragraph
``Posterior-predictive ordering'' in
Appendix~\ref{proof:multi_period_dead_zone} applies to this
possibly asymmetric prior. Hence $Y_2\mid Y_1=y$ is strictly
stochastically increasing in $y$.
Let $B(y)$ denote the reduction in posterior regret from
observing $X_2$, before paying its cost.
By the two-period version of \eqref{eqn:a7}, continuation is
optimal when $B(y)>c$.
The iterated-expectation calculation leading to
\eqref{eqn:a8}, applied to the asymmetric stopping losses, gives
\[
\small 
\begin{aligned} 
B(y)=
\begin{cases}
\mathbb E_{\pi^\star}[(m_2(Y_2))_+\mid Y_1=y],
&y\leq z_1,\\
\mathbb E_{\pi^\star}[(-m_2(Y_2))_+\mid Y_1=y],
&y\geq z_1.
\end{cases}
\end{aligned} 
\]
The strict predictive ordering therefore makes $B$ strictly
increasing below $z_1$ and strictly decreasing above $z_1$.
Gaussian positivity gives $B(y)>0$ at every finite $y$,
and dominated convergence gives continuity.
Also, $B(y)\to0$ in both tails: it is bounded above by the
posterior stopping regret, which vanishes by exponential tilting.

For $0<c<B(z_1)$, the smaller and larger solutions of
$B(y)=c$, divided by $S_1$, define
$\underline\tau^\star$ and $\overline\tau^\star$.
For $c=0$, set
$(\underline\tau^\star,\overline\tau^\star)=(-\infty,\infty)$;
for $c\geq B(z_1)$, set
$\underline\tau^\star=\overline\tau^\star=z_1/S_1$.
In every case, set $\tau_2^\star=z_2/S_2$.
Together with the stopping decisions above, these choices give
the stated rule, with the theorem's convention at the cutoffs.

Bayes actions are unique except at finitely many finite cutoffs.
The uniqueness and null-set argument in
Appendix~\ref{proof:multi_period_dead_zone} therefore shows that
every minimax strategy agrees with this deterministic rule
in the first period and on the continuation event in the second,
under every $\Delta\in[-\bar\Delta,\bar\Delta]$.
Thus the rule is minimax in $\mathcal D$ and has the claimed
almost-sure uniqueness.

Finally, Definition~\ref{defn:asymmetric_regret} gives
$R_K(\boldsymbol\tau;\Delta,c)$ as the regret of each threshold
rule. Since this class contains a minimax rule, minimizing
$\max_{\Delta\in[-\bar\Delta,\bar\Delta]}
R_K(\boldsymbol\tau;\Delta,c)$ over $\mathcal T_K$
gives the stated characterization.

\subsection{Proof of Proposition \ref{prop:p_value}}
\label{proof:p_value}

We will use some of the notation from the proof of Theorem \ref{thm:optimality1}. 
By scale equivariance, take $s=1$. Let $u_0\simeq0.7518$ denote the unique positive
solution to $\Phi(-u_0)=u_0\phi(u_0)$, so that
$\sup_{u\geq0}u\Phi(-u)=u_0\Phi(-u_0)$. Recall that $r_j^\star(c)$ in \eqref{eqn:r_starj} is the
one-policy minimax regret when collecting the follow-up
observation costs $c$.

The paragraph ``A bound for small budgets'' in the proof
of Theorem~\ref{thm:optimality1} establishes that, without
follow-up, the minimax rule implements policy $j$ if
$X_{1j}\geq0$. Since $\bar\Delta_j=s_j$, its worst-case
regret is
$ 
\max_{|\Delta_j|\leq s_j}
|\Delta_j|\Phi\left(-\frac{|\Delta_j|}{s_j}\right)
=s_j u_0\Phi(-u_0).
$ 
The paragraph ``Relative optimality'' in the same proof
establishes that, at zero continuation cost, the minimax
rule always collects the follow-up observation and
implements if $X_{1j}+X_{2j}\geq0$. The pooled mean has
standard error $s_j/\sqrt2$, so
$ 
r_j^\star(0)
=\max_{|\Delta_j|\leq s_j}
|\Delta_j|\Phi\left(-\frac{\sqrt2|\Delta_j|}{s_j}\right)
=\frac{s_j}{\sqrt2}u_0\Phi(-u_0).
$ 
\paragraph{Data-independent experimentation.}
Suppose $s_j=1$ for every $j$. For any $d\in[0,1]$, assign
independent probabilities $1/2$ to $\Delta_j=\pm d$.
In the proof of Lemma~\ref{lem:boundaa}, the argument for
the reverse inequality in
\eqref{eq:optimality_penalized_separation} establishes the
following posterior factorization: under a product prior,
the posterior distribution of $\Delta_j$ given all pilot
observations depends only on $X_{1j}$. If policy $j$
receives follow-up evidence, its posterior additionally
depends on $X_{2j}$, but not on other policies' observations.
The same argument applies here. Thus, observations from
other policies cannot change the posterior expected losses
of implementing or rejecting policy $j$, and cannot improve
its Bayes implementation decision.
The Bayes decisions are $1\{X_{1j}\geq0\}$ without
follow-up and $1\{X_{1j}+X_{2j}\geq0\}$ with follow-up.
At $\Delta_j=d$, their expected regrets are
$ 
d\,\mathbb P(X_{1j}<0\mid\Delta_j=d)
=d\Phi(-d), 
d\,\mathbb P(X_{1j}+X_{2j}<0\mid\Delta_j=d)
=d\Phi\left(-\frac{2d}{\sqrt2}\right)
=d\Phi(-\sqrt2d).
$ 
By symmetry, the same regrets hold at $\Delta_j=-d$,
so these are also their Bayes risks under the stated prior.
Since selection for which policies receive further experimentation is independent of the pilot
observations and at most $q$ policies are selected,
$ 
\sup_{\boldsymbol\Delta}
\overline{\mathrm{Reg}}(\boldsymbol\delta;\boldsymbol\Delta)
\geq
d\left\{
\left(1-\frac qJ\right)\Phi(-d)
+\frac qJ\Phi(-\sqrt2d)
\right\}.
$ 
Maximizing over $d$ and taking limits gives
\[
\liminf_{J\to\infty}
\sup_{\boldsymbol\Delta}
\overline{\mathrm{Reg}}(\boldsymbol\delta;\boldsymbol\Delta)
\geq
F_\rho
:=\max_{0\leq d\leq1}
d\{(1-\rho)\Phi(-d)+\rho\Phi(-\sqrt2d)\}.
\]
We next obtain an upper bound on $R_q^\star$ that we will
show is strictly below $F_\rho$. Since the top $A$-value
rule is feasible, it suffices to upper-bound its worst-case
regret.

Here all one-policy values $r_j^\star$ coincide.
For a fixed cutoff $\lambda$, the first inequality in
\eqref{eq:optimality_fixed_threshold} bounds the average
regret of the rule that continues when $A_j>\lambda$,
including an auxiliary cost $\lambda$ per continuation,
by $r_1^\star(\lambda)$.

The top $A$-value rule instead uses the random cutoff
$\widehat\lambda$ from the proof of
Theorem~\ref{thm:optimality1}. Since costs are one, it
selects exactly $q$ policies. Its average auxiliary cost
is therefore $\widehat\lambda q/J$, which is subtracted
in \eqref{eqn:hh1} to recover regret without these costs.
The paragraph ``Controlling the random cutoff'' in that
proof bounds the resulting fluctuation uniformly by
$C/\sqrt J$. Taking the supremum over possible cutoffs gives
\[
\small 
\begin{aligned} 
R_q^\star
\leq
\sup_{\boldsymbol\Delta}
\overline{\mathrm{Reg}}
(\boldsymbol\delta^{A,q};\boldsymbol\Delta)
\leq
\sup_{\lambda\geq0}
\left\{r_1^\star(\lambda)-\lambda\frac qJ\right\}
+\frac{C}{\sqrt J},
\end{aligned}
\]
where $C$ is independent of $J$.
Since continuation probabilities lie in $[0,1]$,
$r_1^\star$ is $1$-Lipschitz. Also,
$r_1^\star(\lambda)\leq u_0\Phi(-u_0)$, because immediate stopping
is feasible. Since the objective is positive at zero and
negative for $\lambda>u_0\Phi(-u_0) J/q$, the supremum can be
restricted to $[0,u_0\Phi(-u_0) J/q]$. As $q/J\to\rho>0$, these
intervals lie in a common compact interval, giving
$ 
\limsup_{J\to\infty}R_q^\star
\leq
\max_{\lambda\geq0}
\{r_1^\star(\lambda)-\rho\lambda\}.
$ To bound this maximum, consider the one-policy rule that
experiments with probability $\rho$, independently of the
pilot, and uses the corresponding sign decisions. Its
worst-case regret, including continuation cost $\lambda$,
is $F_\rho+\rho\lambda$.

The minimax argument in Appendix
\ref{proof:multi_period_dead_zone} allows comparison with
such randomized rules.\footnote{In Appendix~\ref{proof:multi_period_dead_zone}, the
paragraph ``Existence and worst-case prior'' enlarges
the minimax problem to include randomized rules.
The paragraph ``Threshold representation'' then shows
that a deterministic threshold rule attains this
minimax value. Applying this result with $H=2$,
$r_1^\star(\lambda)$ is therefore also the minimum
worst-case regret when independent randomization is
allowed. In particular, the rule that experiments with
probability $\rho$ gives
$ 
r_1^\star(\lambda)\leq F_\rho+\rho\lambda.
$ 
The same comparison applies to the modified randomized
rule constructed below.}

For $w=1$, letting $q(d,\tau)
=
\Phi(\tau-d)-\Phi(-\tau-d)$, anf $q_\tau$ its derivative with respect to $\tau$,  \eqref{eq:tau_cost_qderiv} gives
$q_\tau(d,t)=\phi(t-d)+\phi(t+d)$.
Conditional on $|X_1|=t$, the probabilities of
$X_1=-t$ and $X_1=t$ are respectively
$1/(1+e^{2dt})$ and $e^{2dt}/(1+e^{2dt})$.
Using the conditional error probability
$\Phi(-x-d)$ at $X_1=x$ from \eqref{eqn:p_gaussian} (where ($p(d,\tau)$ in Equation \eqref{eqn:p_gaussian} is as defined in Equation \eqref{eqn:pd}),
the conditional regrets of the pilot and pooled
sign decisions are therefore
$ 
\frac{d}{1+e^{2dt}}
\text{ and }
\frac{d\{\Phi(t-d)+e^{2dt}\Phi(-t-d)\}}
     {1+e^{2dt}},
$ 
respectively.\footnote{At $\Delta=d>0$, an incorrect decision incurs
regret $d$. The pilot sign is incorrect only when
$X_1=-t$, whose conditional probability is
$1/(1+e^{2dt})$. For the pooled sign, an error occurs
when $X_2<t$ if $X_1=-t$, and when $X_2<-t$ if $X_1=t$.
Since $X_2\sim\mathcal N(d,1)$ independently of $X_1$,
these error probabilities are $\Phi(t-d)$ and
$\Phi(-t-d)$. Weighting them by
$1/(1+e^{2dt})$ and $e^{2dt}/(1+e^{2dt})$, respectively,
and multiplying by $d$ gives the second expression.} Their difference is
$
d\frac{\Phi(d-t)-e^{2dt}\Phi(-d-t)}{1+e^{2dt}}
=-\frac{L_\tau(t,d)}{q_\tau(d,t)} = b(t,d),
$
where $b(t,d)$ is defined in \eqref{eq:tau_cost_b}.
Thus $b(t,d)$ is the conditional reduction in regret
from follow-up given $|X_1|=t$. By symmetry, the same
interpretation holds at $\Delta=-d$.
For $t,d>0$, \eqref{eq:tau_cost_h} gives $b(t,d)>0$.
Moreover, \eqref{eq:tau_cost_b} implies
\[
0\leq b(t,d)
\leq\frac{d}{1+e^{2dt}}
\leq d e^{-2dt}
\leq\frac{1}{2\mathrm e\,t},
\]
where the last inequality follows by maximizing
$d e^{-2dt}$ over $d\geq0$.
The same bounds hold at $d=0$, since $b(t,0)=0$.

For $\lambda>0$, choose $T>1/(\mathrm e\lambda)$ and
modify the randomized rule to use the pilot-sign decision
without follow-up whenever $|X_1|>T$.
At $|\Delta|=d\leq1$, cancelling continuation conditional
on $|X_1|=t>T$ increases expected implementation regret by
$b(t,d)\leq1/(2\mathrm e\,t)<\lambda/2$ and saves the
continuation cost $\lambda$.
Since the original rule continues with probability $\rho$,
the reduction in expected regret, including costs, is
$
\rho\,\mathbb E_\Delta\!\left[
\{\lambda-b(|X_1|,d)\}1\{|X_1|>T\}
\right]
\geq
\frac{\rho\lambda}{2}\,
\mathbb P_\Delta(|X_1|>T).
$
Moreover, \eqref{eq:tau_cost_qderiv} gives $q_d\leq0$ (where $q_d$ is the derivative of $q$ with respect to $d$), so
$
\mathbb P_\Delta(|X_1|>T)
=1-q(d,T)
\geq1-q(0,T)
=2\Phi(-T).
$
Thus the modification reduces expected regret, including
costs, by at least $\rho\lambda\Phi(-T)>0$, uniformly
over $|\Delta|\leq1$.

Consequently,
\[
\small 
\begin{aligned} 
r_1^\star(\lambda)-\rho\lambda
\leq F_\rho-\rho\lambda\Phi(-T)<F_\rho,
\qquad \lambda>0.
\end{aligned} 
\]
Since $u_0/\sqrt2\in[0,1]$, the definition of $F_\rho$ gives
$
F_\rho
\geq
\frac{u_0}{\sqrt2}
\left\{(1-\rho)\Phi(-u_0/\sqrt2)+\rho\Phi(-u_0)\right\} >
\frac{u_0}{\sqrt2}\Phi(-u_0)
=r_1^\star(0),
$
where strict inequality follows from $1-\rho>0$ and
$\Phi(-u_0/\sqrt2)>\Phi(-u_0)$.
Together with the inequality already proved for
$\lambda>0$, this shows
$r_1^\star(\lambda)-\rho\lambda<F_\rho$
for every $\lambda\geq0$.

The function $r_1^\star(\lambda)-\rho\lambda$ is continuous
and bounded above by
$u_0\Phi(-u_0)-\rho\lambda\to-\infty$.
Its maximum is therefore attained, so the strict
inequality also holds at its maximizing value of $\lambda$.
Combining the preceding lower bound for data-independent
allocations with the upper bound for $R_q^\star$ yields
$\liminf_{J\to\infty}
\left\{
\sup_{\boldsymbol\Delta}
\overline{\mathrm{Reg}}(\boldsymbol\delta;\boldsymbol\Delta)
-R_q^\star
\right\} \geq
F_\rho-\max_{\lambda\geq0}
\{r_1^\star(\lambda)-\rho\lambda\}>0.
$

\paragraph{$p$-value-based experimentation.}
Set $s_j=1$ for $j\leq q$ and $s_j=\varepsilon$ for
$j>q$, where $\varepsilon\in(0,1)$ will be chosen below.
Use the product prior from the proof of Corollary
\ref{prop:minimax}: independently, $\Delta_j=\pm d_0s_j$
with equal probabilities, where $d_0=u_0/\sqrt2$.
Then $T_j=|X_{1j}|/s_j$ are independent and identically
distributed as $|N(d_0,1)|$. All expectations below are
under this prior and the sampling distribution.

Write $e_j=1\{\delta_{1j}=\texttt{NA}\}$ for the selection
indicators in Definition~\ref{defn:p_value}.
The proof of the reverse inequality in
\eqref{eq:optimality_penalized_separation}, within
Lemma~\ref{lem:boundaa}, shows that observations from
other policies do not affect the posterior of $\Delta_j$.
Under the present prior, which assigns equal probabilities
to $\Delta_j=\pm d_0s_j$, the posterior odds of the
positive effect against the negative effect are
$
\exp\left(\frac{2d_0X_{1j}}{s_j}\right)
\quad\text{ and } 
\exp\left(\frac{2d_0(X_{1j}+X_{2j})}{s_j}\right)
$
after the pilot and after follow-up, respectively.
Since either incorrect decision incurs loss $d_0s_j$,
the Bayes decision implements whenever these odds are
at least one. Thus the Bayes implementation decisions
are $1\{X_{1j}\geq0\}$ without follow-up and
$1\{X_{1j}+X_{2j}\geq0\}$ with follow-up.

Conditional on the pilot observations, the expected regret
of the pilot-sign decision minus that of the pooled-sign
decision, averaging over the follow-up observation, is
$
s_jd_0
\frac{\Phi(d_0-T_j)-e^{2d_0T_j}\Phi(-d_0-T_j)}
     {1+e^{2d_0T_j}}
=s_jb(T_j,d_0),
$
by \eqref{eq:tau_cost_b} with $w=1$.
The factor $s_j$ appears because an incorrect decision
costs $|\Delta_j|=s_jd_0$, whereas $b(T_j,d_0)$ uses
the normalized loss $d_0$.
Thus, using the Bayes implementation decisions, the
additional conditional regret relative to collecting
follow-up evidence for policy $j$ is
$s_jb(T_j,d_0)(1-e_j)$: it is zero when $e_j=1$ and
$s_jb(T_j,d_0)$ when $e_j=0$.
Thus any implementation rule has
Bayes risk at least
$$
\small 
\begin{aligned} 
\frac1J\sum_{j=1}^J s_j
\left\{\frac{u_0 \Phi(-u_0)}{\sqrt2}
+\mathbb E[b(T_j,d_0)(1-e_j)]\right\}.
\end{aligned} 
$$
By permutation equivariance in Definition~\ref{defn:p_value}
and the identical distributions of the $T_j$, the
expectation in braces is the same for every $j$.
The last display therefore equals
\[
\small 
\begin{aligned} 
\left(\frac1J\sum_{j=1}^J s_j\right)
\left\{\frac{u_0 \Phi(-u_0)}{\sqrt2}
+\frac1J\mathbb E\sum_{j=1}^J b(T_j,d_0)(1-e_j)\right\}.
\end{aligned} 
\]

Equation \eqref{eq:tau_cost_btd} shows that $b(t,d_0)$
decreases in $t$. Choose $t_\rho>0$ such that
$\mathbb P(T_j\leq t_\rho)=(1+\rho)/2$.
Each unselected policy with $T_j\leq t_\rho$ contributes
at least $b(t_\rho,d_0)>0$, and at most $q$ of these
policies can be selected. Hence
$
\frac1J\mathbb E\sum_{j=1}^J b(T_j,d_0)(1-e_j)
\geq
b(t_\rho,d_0)\left(\frac{1+\rho}{2}-\frac qJ\right).
$
The right-hand side converges to
$B_\rho:=\frac{1-\rho}{2}b(t_\rho,d_0)>0$.
Since worst-case regret is at least Bayes risk,
\[
\small 
\begin{aligned} 
\liminf_{J\to\infty}
\sup_{\boldsymbol\Delta}
\overline{\mathrm{Reg}}(\boldsymbol\delta;\boldsymbol\Delta)
\geq
[\rho+(1-\rho)\varepsilon]
\left(\frac{u_0 \Phi(-u_0)}{\sqrt2}+B_\rho\right).
\end{aligned} 
\]

For comparison, experimenting on policies $1,\ldots,q$
and using the corresponding sign decisions is feasible
and gives
$
R_q^\star\leq
u_0\Phi(-u_0)\left\{
\frac{q}{\sqrt2J}
+\left(1-\frac qJ\right)\varepsilon
\right\}.
$\footnote{Select exactly policies $1,\ldots,q$, which
is feasible since $c_j=1$. Use
$1\{X_{1j}+X_{2j}\geq0\}$ for these policies and
$1\{X_{1j}\geq0\}$ for the others.
The paragraphs ``Relative optimality'' and
``A bound for small budgets'' in the proof of
Theorem~\ref{thm:optimality1} give their respective
worst-case regrets as $s_ju_0\Phi(-u_0)/\sqrt2$ and $s_j u_0 \Phi(-u_0)$.
Since $s_j=1$ for $j\leq q$ and $s_j=\varepsilon$
otherwise, averaging gives
$[q u_0\Phi(-u_0)/\sqrt2+(J-q)\varepsilon u_0 \Phi(-u_0)]/J$.
By definition, $R_q^\star$ is at most the worst-case
regret of this particular feasible strategy.}
Subtracting this bound yields
\[
\small 
\begin{aligned}
&\liminf_{J\to\infty}
\left\{
\sup_{\boldsymbol\Delta}
\overline{\mathrm{Reg}}(\boldsymbol\delta;\boldsymbol\Delta)
-R_q^\star
\right\} \geq
[\rho+(1-\rho)\varepsilon]B_\rho
-(1-\rho)\varepsilon u_0\Phi(-u_0)
\left(1-\frac1{\sqrt2}\right).
\end{aligned}
\]
The right-hand side tends to $\rho B_\rho>0$ as
$\varepsilon\downarrow0$. Choose a fixed
$\varepsilon_\rho\in(0,1)$ making it strictly positive.

Finally, setting $\lambda=0$ in
\eqref{eq:optimality_penalized_lower} of
Lemma~\ref{lem:boundaa} gives
$
R_q^\star
\geq\frac1J\sum_{j=1}^J r_j^\star(0)
=\frac{u_0 \Phi(-u_0)}{\sqrt2J}\sum_{j=1}^J s_j
>0.
$
The equality uses the zero-cost calculation in the
paragraph ``Relative optimality'' in the proof of
Theorem~\ref{thm:optimality1}; positivity follows from
$u_0 \Phi(-u_0)>0$ and $s_j>0$.

In both constructions, using only the pilot-sign decisions
is feasible and gives
$ 
R_q^\star
\leq \frac{u_0\Phi(-u_0)}{J}\sum_{j=1}^J s_j
\leq u_0\Phi(-u_0).
$ 
Since every feasible strategy has worst-case regret at
least $R_q^\star$, it follows that
$ 
\frac{
\sup_{\boldsymbol\Delta}
\overline{\mathrm{Reg}}(\boldsymbol\delta;\boldsymbol\Delta)
}{R_q^\star}
\geq
1+
\frac{
\sup_{\boldsymbol\Delta}
\overline{\mathrm{Reg}}(\boldsymbol\delta;\boldsymbol\Delta)
-R_q^\star
}{u_0\Phi(-u_0)}.
$ 
The two strictly positive additive lower bounds established
above therefore imply the claimed relative bounds, with
a common $\eta_\rho>0$.
Finally, restoring $s$ multiplies all regrets and minimax
values by $s$, leaving their ratios unchanged.

\subsection{Proof of Proposition \ref{prop:general_prior}}
\label{proof:prop:general_prior}

Fix a working prior $\Pi$ satisfying Definition~\ref{defn:bayes_rule}.
Let $u_0\simeq0.7518$ denote the unique positive
solution to $\Phi(-u_0)=u_0\phi(u_0)$, so that
$\sup_{u\geq0}u\Phi(-u)=u_0\Phi(-u_0)$ and $d_0=u_0/\sqrt2$,
write $\kappa=u_0\Phi(-u_0)$, and use $b(t,d)$ from
\eqref{eq:tau_cost_b} with $w=1$.

\paragraph{Follow-up values under the working prior.}
For one policy with standard error $s$, let $G_{\Pi,s}(x)$
be the conditional reduction in Bayes regret from collecting
follow-up evidence after $X_1=x$. We can write 
\[
\small 
\begin{aligned} 
G_{\Pi,s}(x)
=\mathbb E_\Pi\!\left[
\bigl(\mathbb E_\Pi[\Delta \mid X_1=x,X_2]\bigr)_+
\mid X_1=x\right]
-\bigl(\mathbb E_\Pi[\Delta \mid X_1=x]\bigr)_+.
\end{aligned} 
\]
For any follow-up decision $a:\mathbb R\to[0,1]$, the
Bayes decision without follow-up is at least as good as
implementing with constant probability $\mathbb E_0[a(X_2)]$.
The Gaussian total variation bound gives
$
\left|\mathbb E_\Delta[a(X_2)]-\mathbb E_0[a(X_2)]\right|
\leq 2\Phi\left(\frac{|\Delta|}{2s}\right)-1
\leq\min\left\{1,\frac{|\Delta|}{s\sqrt{2\pi}}\right\}.
$ Here $\mathbb E_\Delta$ denotes expectation when
$X_2\sim\mathcal N(\Delta,s^2)$, and $\mathbb E_0$
uses the same standard error with mean zero.
Comparing expected regrets and taking
the supremum over $a$ therefore gives
\begin{equation}
\small 
\begin{aligned} 
\label{eq:general_prior_gain_bound}
0\leq G_{\Pi,s}(x)
\leq
\mathbb E_\Pi\!\left[
|\Delta|\min\left\{1,\frac{|\Delta|}{s\sqrt{2\pi}}\right\}
\mid X_1=x\right].
\end{aligned} 
\end{equation}

At $x=st$, Bayes' rule weights $\Pi(d\Delta)$ by
$\exp\{t\Delta/s-\Delta^2/(2s^2)\}$ and normalizes by
its integral. For $|t|\leq M$, this weight is bounded by
$e^{M^2/2}$ and converges to one uniformly in $t$, for each
fixed $\Delta$. The normalizing integral therefore converges
uniformly to one. Dominated convergence in
\eqref{eq:general_prior_gain_bound}, using
$\int|\Delta|\,\Pi(d\Delta)<\infty$, gives
\begin{equation} 
\label{eq:general_prior_vanishing_gain}
\small 
\begin{aligned} 
\sup_{|t|\leq M}G_{\Pi,s}(st)\rightarrow0
\qquad (s\to\infty)
\end{aligned}
\end{equation} for every finite $M$.

For every fixed $s>0$, $G_{\Pi,s}(x)$ is continuous by
dominated convergence. It is also positive: the posterior
after any pilot retains mass on both signs, and sufficiently
large positive or negative follow-up observations make the
posterior mean positive or negative, respectively.
This follows by comparing a fixed positive-mass interval
on the corresponding side of zero with the opposite-sign
contribution to the Gaussian likelihood.
Both events have positive predictive probability, so strictness
of Jensen's inequality gives $G_{\Pi,s}(x)>0$. In particular,
\begin{equation}
\small 
\begin{aligned} 
\label{eq:general_prior_positive_gain}
g_0:=\min_{|x|\leq s_0}G_{\Pi,s_0}(x)>0.
\end{aligned} 
\end{equation}

The posterior-factorization argument used to prove
\eqref{eq:optimality_penalized_separation} in
Lemma~\ref{lem:boundaa} applies to $\Pi^{\otimes J}$.
Thus, with unit costs, a Bayes rule selects the $q$ largest
values $G_{\Pi,s_j}(X_{1j})$, since all gains are positive.
The working-prior predictive density is positive everywhere,
so this selection property holds almost surely under every
effect vector, regardless of how ties are resolved.

\paragraph{A lower bound on omitted follow-up benefits.}
To lower-bound worst-case regret, use a different product
prior $\pi$: independently, $\Delta_j=\pm d_0s$ with equal
probabilities for $j\leq q$, and $\Delta_j=0$ for $j>q$.
Its support is admissible because $d_0<1$. The rule still
selects policies according to the working prior $\Pi$;
$\pi$ is used only to lower-bound its regret.

Under $\pi$, the variables $T_j=|X_{1j}|/s_j$ are independent,
with distribution $|N(d_0,1)|$ for $j\leq q$ and
$|N(0,1)|$ for $j>q$.
The calculation in the proof of Proposition~\ref{prop:p_value}
shows that omitting follow-up for a policy $j\leq q$ adds
conditional Bayes regret $s b(T_j,d_0)$.
Equations \eqref{eq:tau_cost_h} and \eqref{eq:tau_cost_btd}
give positivity and monotonic decrease of $b(t,d_0)$ in $t$.

Set
$
\alpha_\rho=\rho(1-\rho)[2\Phi(1)-1]>0.
$
Choose $M_\rho>0$ such that
$\rho\,\mathbb P_\pi(T_1>M_\rho)\leq\alpha_\rho/2$, and put
$D_\rho=(\alpha_\rho/2)b(M_\rho,d_0)>0$.
These constants depend only on $\rho$, because the
distribution of $T_1$ does not depend on $J$, $s$, or $\Pi$.
By \eqref{eq:general_prior_vanishing_gain} and
\eqref{eq:general_prior_positive_gain}, choose a fixed
$s>s_0$ large enough that
\begin{equation}
\small 
\begin{aligned} 
\label{eq:general_prior_precision_choice}
\sup_{|t|\leq M_\rho}G_{\Pi,s}(st)<g_0,
\qquad
\kappa\frac{s_0}{s}(1-\rho)\leq\frac{D_\rho}{2}.
\end{aligned}
\end{equation}
Every policy $j>q$ with $T_j\leq1$ then ranks strictly
above every policy $j\leq q$ with $T_j\leq M_\rho$.

Write $e_j=1\{\delta_{1j}^{\Pi,q}=\texttt{NA}\}$ for the
selection indicators, and let
$N_J=\sum_{j=q+1}^J1\{T_j\leq1\}$.
At most $(q-N_J)_+$ first-group policies with
$T_j\leq M_\rho$ can be selected. Consequently,
\[
\small 
\begin{aligned} 
\sum_{j=1}^q b(T_j,d_0)(1-e_j)
\geq b(M_\rho,d_0)
\left[\min\{q,N_J\}-\sum_{j=1}^q1\{T_j>M_\rho\}\right].
\end{aligned} 
\]
Since $\min\{q,N_J\}\geq qN_J/J$ and
$\mathbb E_\pi[N_J]=(J-q)[2\Phi(1)-1]$, taking expectations
in the preceding inequality and using $q/J\to\rho$ gives
\begin{equation}
\small 
\label{eq:general_prior_omitted_gain}
\begin{aligned}
\liminf_{J\to\infty}\frac1J
\mathbb E_\pi\sum_{j=1}^q b(T_j,d_0)(1-e_j)
&\geq b(M_\rho,d_0)
\{\alpha_\rho-\rho\mathbb P_\pi(T_1>M_\rho)\}\geq D_\rho.
\end{aligned}
\end{equation}

\paragraph{Comparison with minimax regret.}
The two-point-prior calculation in the proof of
Corollary~\ref{prop:minimax} gives Bayes regret
$\kappa s/\sqrt2$ for each $j\leq q$ with free follow-up;
the remaining policies have zero loss under $\pi$.
As in the proof of Proposition~\ref{prop:p_value}, the
posterior factorization in Lemma~\ref{lem:boundaa} implies
that observations from other policies cannot replace an
omitted follow-up observation. Thus, allowing Bayes-optimal
implementation under $\pi$ gives the lower bound
$
\sup_{\boldsymbol\Delta}
\overline{\mathrm{Reg}}
(\boldsymbol\delta^{\Pi,q};\boldsymbol\Delta)
\geq\frac{q\kappa s}{\sqrt2J}
+\frac{s}{J}\mathbb E_\pi
\sum_{j=1}^q b(T_j,d_0)(1-e_j).
$
This bound applies regardless of the actual implementation
decisions used by $\boldsymbol\delta^{\Pi,q}$.

For comparison, select policies $1,\ldots,q$, use their
pooled signs, and use the pilot signs for the others.
The same feasible-rule comparison as in the proof of
Proposition~\ref{prop:p_value} gives
$
R_q^\star\leq
\kappa s\left[
\frac{q}{\sqrt2J}+\frac{s_0}{s}\left(1-\frac qJ\right)
\right].
$
Subtracting this bound and using
\eqref{eq:general_prior_omitted_gain} and
\eqref{eq:general_prior_precision_choice} yields
$
\liminf_{J\to\infty}
\frac{
\sup_{\boldsymbol\Delta}
\overline{\mathrm{Reg}}
(\boldsymbol\delta^{\Pi,q};\boldsymbol\Delta)-R_q^\star
}{s}
\geq D_\rho-\kappa\frac{s_0}{s}(1-\rho)
\geq\frac{D_\rho}{2}.
$
Finally, \eqref{eq:optimality_penalized_lower} at $\lambda=0$
and the zero-cost values computed in the proof of
Theorem~\ref{thm:optimality1} give
$R_q^\star\geq\frac{\kappa}{\sqrt2\,J}\sum_j s_j>0$.
The preceding upper bound also gives $R_q^\star\leq\kappa s$.
Therefore,
$
\liminf_{J\to\infty}
\frac{
\sup_{\boldsymbol\Delta}
\overline{\mathrm{Reg}}
(\boldsymbol\delta^{\Pi,q};\boldsymbol\Delta)
}{R_q^\star}
\geq1+\frac{D_\rho}{2\kappa}.
$
Taking $\eta_\rho=D_\rho/(2\kappa)$ proves the result;
only the choice of $s$ depends on $\Pi$ and $s_0$.

\subsection{Proof of Proposition \ref{prop:nonparametric_eb}}
\label{sec:empirical_bayes}

We break the proof into several steps. 

\paragraph{The fitted conditional prior.}
The standard errors $s_j=s(1+j/J^2)$ are pairwise distinct,
so the maximization in \eqref{eq:nonparametric_eb_fit}
separates across policies. Indeed, any admissible choices at
the finitely many observed standard errors extend to a Borel
kernel by setting $G_v=\delta_0$ elsewhere, where $\delta_t$
denotes the point mass at $t$.

For each $j$, set
\[
\small 
\begin{aligned}
\widehat\Delta_j=\max\{-s_j,\min\{X_{1j},s_j\}\}.
\end{aligned}
\]
The Gaussian likelihood has its unique maximum over
$[-s_j,s_j]$ at $\widehat\Delta_j$. Consequently, for every
probability distribution $G$ supported on this interval,
$
\int\varphi_{s_j}(X_{1j}-t)\,G(dt)
\leq\varphi_{s_j}(X_{1j}-\widehat\Delta_j),
$ 
with equality if and only if $G=\delta_{\widehat\Delta_j}$.
Thus every maximizer in \eqref{eq:nonparametric_eb_fit}
satisfies
\begin{equation}
\label{eq:nonparametric_eb_pointmass}
\small 
\begin{aligned} 
\widehat G_{s_j}=\delta_{\widehat\Delta_j},
\qquad j=1,\ldots,J.
\end{aligned} 
\end{equation}

\paragraph{Implementation decisions and regret.}
By Definition~\ref{defn:nonparametric_eb}, the fitted prior
is held fixed. Updating the point mass in
\eqref{eq:nonparametric_eb_pointmass} by a positive Gaussian
likelihood leaves it unchanged. The conditional-gain formula
defining $G_{\Pi,s}(x)$ in the proof of
Proposition~\ref{prop:general_prior} therefore reduces to
$(\widehat\Delta_j)_+-(\widehat\Delta_j)_+=0$.
Every feasible experimentation set is consequently optimal
under the fitted model.

Whatever set is selected, the Bayes implementation decision
is uniquely $1\{\widehat\Delta_j>0\}$ whenever
$\widehat\Delta_j\neq0$. Since $\widehat\Delta_j$ has the
sign of $X_{1j}$ and $X_{1j}=0$ has probability zero, the
eventual decision is $1\{X_{1j}\geq0\}$ almost surely under
every effect vector. Positive Gaussian follow-up densities
ensure that any exceptions under the fitted predictive model
are also null under the true sampling model. Hence
$ 
\overline{\mathrm{Reg}}
(\widehat{\boldsymbol\delta}^{\,\mathrm{UR},q};\boldsymbol\Delta)
=\frac1J\sum_{j=1}^J
|\Delta_j|\Phi\left(-\frac{|\Delta_j|}{s_j}\right).
$ 
Let $u_0\simeq0.7518$ denote the unique positive
solution to $\Phi(-u_0)=u_0\phi(u_0)$, and write
$\kappa=u_0\Phi(-u_0)$.
The pilot-risk calculation in the paragraph ``A bound for
small budgets'' of that proof, together with $\bar\Delta_j=s_j$,
gives the exact identity
\begin{equation}
\label{eq:nonparametric_eb_risk}
\sup_{\boldsymbol\Delta}
\overline{\mathrm{Reg}}
(\widehat{\boldsymbol\delta}^{\,\mathrm{UR},q};\boldsymbol\Delta)
=\frac{\kappa}{J}\sum_{j=1}^J s_j,
\end{equation}
since the maximum is attained at $|\Delta_j|=u_0s_j$
for every $j$.

\paragraph{Comparison with the constrained optimum.}
As in the proof of Proposition~\ref{prop:p_value}, consider
the feasible rule that selects policies $1,\ldots,q$, uses
$1\{X_{1j}+X_{2j}\geq0\}$ for these policies, and uses
$1\{X_{1j}\geq0\}$ for the others.
The paragraph ``Relative optimality'' in the proof of
Theorem~\ref{thm:optimality1} gives pooled-sign worst-case
regret $s_j\kappa/\sqrt2$. Therefore,
\begin{equation}
\label{eq:nonparametric_eb_benchmark}
\small 
\begin{aligned} 
R_q^\star\leq\frac{\kappa}{J}
\left\{
\frac1{\sqrt2}\sum_{j=1}^q s_j+
\sum_{j=q+1}^J s_j
\right\}.
\end{aligned} 
\end{equation}
Also, \eqref{eq:optimality_penalized_lower} in
Lemma~\ref{lem:boundaa}, at $\lambda=0$, gives
$ 
R_q^\star\geq\frac1J\sum_{j=1}^J r_j^\star(0)
=\frac{\kappa}{\sqrt2\,J}\sum_{j=1}^J s_j>0.
$ 
We may therefore divide \eqref{eq:nonparametric_eb_risk}
by the positive minimax value and apply
\eqref{eq:nonparametric_eb_benchmark} to obtain
$ 
\frac{
\sup_{\boldsymbol\Delta}
\overline{\mathrm{Reg}}
(\widehat{\boldsymbol\delta}^{\,\mathrm{UR},q};\boldsymbol\Delta)
}{R_q^\star}
\geq
\frac{\sum_{j=1}^J s_j}
{\frac1{\sqrt2}\sum_{j=1}^q s_j+\sum_{j=q+1}^J s_j}.
$ 
Finally, $\max_{j\leq J}|s_j-s|=s/J\to0$ and $q/J\to\rho$
imply
$ 
\frac1J\sum_{j=1}^J s_j\to s,
\frac1J\sum_{j=1}^q s_j\to s\rho.
$ 
Taking lower limits yields
$ 
\liminf_{J\to\infty}
\frac{
\sup_{\boldsymbol\Delta}
\overline{\mathrm{Reg}}
(\widehat{\boldsymbol\delta}^{\,\mathrm{UR},q};\boldsymbol\Delta)
}{R_q^\star}
\geq\frac1{1-\rho(1-1/\sqrt2)}.
$ 
Thus the result holds with
$\eta_\rho=[1-\rho(1-1/\sqrt2)]^{-1}-1>0$.

\subsection{Proof of Theorem~\ref{thm:non_gaussian_rule}}
\label{app:proof_non_gaussian_rule}

The proof mimics the structure of the proof of Lemma \ref{thm:dead_zone}. Some of the steps follow verbatim. 

Write $\Theta=[-\bar\Delta,\bar\Delta]$. Throughout, a prior on $\Theta$ is
used only to characterize the minimax solution.

\paragraph{Existence and a symmetric least-favorable prior} The existence and minimax representation argument in the proof of Theorem \ref{thm:dead_zone}. Indeed, the joint sampling distributions are dominated by \(h_1(x)h_2(y)\,dx\,dy\) and vary continuously in \(\Delta\) in total variation, by Scheffé’s lemma. The parameter space is compact, and regret is bounded and continuous in \(\Delta\), uniformly over randomized strategies. Hence a minimax strategy and a least-favorable prior exist.
Consequently, regret is bounded and continuous in $\Delta$, uniformly over
strategies. The minimax theorem therefore gives a minimax strategy, a
least-favorable prior, and
\[
 \small 
 \begin{aligned} V:=\inf_\delta\sup_{\Delta\in\Theta}
       \operatorname{Reg}(\delta;\Delta,c)
   =\max_{\pi\in\mathcal P(\Theta)} b(\pi),\qquad
 b(\pi):=\inf_\delta\int_\Theta
       \operatorname{Reg}(\delta;\Delta,c)\,\pi(d\Delta).
       \end{aligned} 
\]
Evenness of $\psi$ implies that $\eta$ is odd and
$f_j(-x;-\Delta)=f_j(x;\Delta)$. Reflecting observations and reversing the
terminal actions thus preserves regret after replacing $\Delta$ by
$-\Delta$. Since $b$ is concave, symmetrizing a least-favorable prior cannot
decrease its Bayes risk. We may therefore choose a symmetric
least-favorable prior $\pi^\star$.

This prior is nondegenerate. For any $d\in(0,\bar\Delta]$, put equal prior
mass on $-d$ and $d$. Even when both observations are available at zero
cost, its Bayes regret is
\[
\small 
\begin{aligned} 
 \frac d2\int_{\mathbb R^2}
 \min\{f_1(x;d)f_2(y;d),\,
        f_1(x;-d)f_2(y;-d)\}\,dx\,dy>0.
        \end{aligned} 
\]
This is a lower bound on its Bayes regret in the original experiment,
whereas a prior concentrated at zero has Bayes regret zero. Thus
$\pi^\star$ is not concentrated at zero and, by symmetry, assigns positive
probability to both signs. Finally, every minimax strategy is Bayes under
$\pi^\star$, because its integrated regret lies between $b(\pi^\star)=V$
and its worst-case regret $V$.

\paragraph{Posterior policy decisions.}
Let $S_1=a_1X_1$ and $S_2=a_1X_1+a_2X_2$.
For $n\in\{a_1,a_1+a_2\}$, define
\[
\small 
\begin{aligned} 
 M_n(s)=\int_\Theta e^{\eta(\Delta)s-n\psi(\eta(\Delta))}
                       \,\pi^\star(d\Delta),\qquad
 m_n(s)=\frac{\int_\Theta\Delta
           e^{\eta(\Delta)s-n\psi(\eta(\Delta))}\,\pi^\star(d\Delta)}
                  {M_n(s)}.
                  \end{aligned} 
\]
These are the posterior normalizer and posterior mean after accumulated
weight $n$ and evidence $s$. Symmetry makes $M_n$ even and $m_n$ odd.
Differentiating under the integral, which is justified by compactness,
gives
$ 
m_n'(s)=\operatorname{Cov}_{\pi^\star}
       \big(\Delta,\eta(\Delta)\mid n,s\big)>0.
$ 
The inequality is strict because $\eta$ is strictly increasing and the
posterior has the same nondegenerate support as $\pi^\star$.
Hence $\operatorname{sign}(m_n(s))=\operatorname{sign}(s)$.
The posterior regrets from rejection and implementation differ by $m_n(s)$,
so the unique stopping decision, except at $s=0$, is implementation iff
$s>0$. In particular, the terminal rule is
$\mathbf 1\{S_2\ge0\}$.

\paragraph{Ordering of follow-up evidence.}
Set $q_2(u)=a_2^{-1}h_2(u/a_2)$. The posterior-predictive density of
$S_2=t$ conditional on $S_1=s$ is
\begin{equation}
\small 
\begin{aligned} 
 K(s,t)=q_2(t-s)\frac{M_{a_1+a_2}(t)}{M_{a_1}(s)}.
 \label{eq:non_gaussian_predictive}
 \end{aligned} 
\end{equation}
For $s'>s$,
\[
\small 
\begin{aligned} 
 \frac{K(s',t)}{K(s,t)}
 =\frac{M_{a_1}(s)}{M_{a_1}(s')}
   \frac{q_2(t-s')}{q_2(t-s)}.
   \end{aligned} 
\]
Log-concavity of $q_2$ makes this ratio nondecreasing in $t$.
It is not constant: a constant ratio of two translates of a positive
log-concave density would make its log-density affine, contradicting
integrability on $\mathbb R$. Since these densities are positive everywhere,
the nonconstant monotone likelihood ratio implies strict stochastic
dominance: the conditional distribution function of $S_2$ is strictly
decreasing in $s$ at every finite argument. Also,
$K(-s,-t)=K(s,t)$.

\paragraph{The continuation region.}
Let $r_n(s)$ denote the minimum posterior regret from stopping, and define
the gain from costless follow-up evidence by
\[
\small 
\begin{aligned} 
 G(s)=r_{a_1}(s)
      -\mathbb E_{\pi^\star}\!\left[
           r_{a_1+a_2}(S_2)\mid S_1=s\right].
           \end{aligned} 
\]
The posterior mean is a martingale across the two observations. Therefore,
for $s\ge0$,
\begin{equation}
 G(s)=\mathbb E_{\pi^\star}\!\left[
       \big(-m_{a_1+a_2}(S_2)\big)_+\mid S_1=s\right].
 \label{eq:non_gaussian_gain}
\end{equation}
Indeed, $r_n(s)=\mathbb E[\Delta_+\mid n,s]-(m_n(s))_+$, so the
conditional expectation of $\Delta_+$ cancels, and
$\mathbb E[m_{a_1+a_2}(S_2)\mid S_1=s]=m_{a_1}(s)$.
The integrand in~\eqref{eq:non_gaussian_gain} is bounded, nonincreasing,
and strictly decreasing on $(-\infty,0)$. The strict stochastic ordering
just proved implies that $G$ is strictly decreasing on $[0,\infty)$.
Reflection symmetry makes $G$ even. It is positive at every finite $s$,
because $K(s,t)>0$ and $m_{a_1+a_2}(t)<0$ for $t<0$.
It is continuous: the densities in~\eqref{eq:non_gaussian_predictive}
are pointwise continuous in $s$, so Scheff\'e's lemma gives continuity
in total variation.

Finally, $G(s)\le r_{a_1}(s)=\mathbb E[\Delta_-\mid S_1=s]$ for
$s\ge0$, and the latter tends to zero as $s\to\infty$.
To see this, choose $d>0$ such that
$\pi^\star([d,\bar\Delta])>0$. The posterior numerator over negative
effects is bounded for $s\ge0$, whereas its denominator is bounded below
by a positive constant times $e^{\eta(d)s}$.
Consequently, the Bayes rule continues exactly when $G(S_1)>c$, apart
from ties. If $c=0$ it always continues. If $c\ge G(0)$ it stops immediately.
Otherwise there is a unique $b>0$ such that $G(b)=c$, and continuation is
equivalent to $|S_1|<b$. Setting $\tau^\star=b/a_1$, with the corresponding
conventions $\tau^\star=\infty$ and $\tau^\star=0$, gives the stated rule.

The uniqueness argument in the proof of Lemma~\ref{thm} applies here:
decision ties have probability zero, and positivity of the
prior-predictive densities implies that every minimax strategy
agrees almost surely with the displayed rule under every
$\Delta\in\Theta$.

Conditioning on $X_1$ gives the regret expression
in~\eqref{eq:non_gaussian_threshold} for $\Delta\ge0$;
reflection gives the same regret at $-\Delta$.
Since the threshold class contains an unrestricted minimax rule,
minimizing its worst-case regret over $\tau$ yields the
unrestricted minimax value.

\paragraph{Uniqueness and the threshold objective.}
Under $\pi^\star$, all decision ties occur on the first-period boundaries
or on $S_2=0$, which have probability zero. Thus the Bayes rule is unique
almost surely, including the terminal decision on continuation histories.
Every minimax strategy is Bayes under this same prior, so it agrees with
the displayed rule. The prior-predictive densities are everywhere positive;
their null sets are therefore null under every $\Delta\in\Theta$.
In particular, randomization does not improve on this deterministic rule.

For $\Delta\ge0$, an incorrect immediate decision has probability
$F_1(-\tau;\Delta)$. Conditional on $X_1=x$ and continuation, an
incorrect terminal decision has probability
$F_2(-a_1x/a_2;\Delta)$. Adding the continuation cost gives the expression
in~\eqref{eq:non_gaussian_threshold}; reflection gives the same regret
at $-\Delta$. Since an unrestricted minimax rule belongs to this threshold
class, minimizing its worst-case regret over $\tau$ yields precisely the
unrestricted minimax value. This proves the threshold characterization
and completes the proof.

\subsection{Proof of Theorem~\ref{thm:prior_a_value}}
\label{proof:prior_a_value}

\leavevmode

\paragraph{Bayesian threshold.}
First note that the arguments concerning posterior decisions, predictive ordering,
and backward induction in the proof of
Theorem~\ref{thm:multi_period_dead_zone} apply with $H=2$ and
$\pi$ replacing $\pi^\star$. They require only symmetry,
nondegeneracy, and bounded support, all guaranteed by
Assumption~\ref{ass:prior_information}. As a result, we can apply similar arguments also under a fixed prior $\pi$ instead of the least favorable prior used in the proof of Theorem~\ref{thm:multi_period_dead_zone}. Specifically, write $a_j=s_j^{-2}$,
$m_1(x)=\mathbb E_\pi[\Delta\mid X_1=x]$, and
$m_2(x,y)=\mathbb E_\pi[\Delta\mid X_1=x,X_2=y]$.
Define the posterior gain from a costless follow-up by
\begin{equation}
\small 
\begin{aligned} 
 G_\pi(x)
 =\mathbb E_\pi[
      \{m_2(x,X_2)\}_+\mid X_1=x]
      -\{m_1(x)\}_+.
 \label{eq:prior_posterior_gain}
 \end{aligned} 
\end{equation}
Following verbatim  the arguments in Theorem~\ref{thm:multi_period_dead_zone} under the prior $\pi$ imply $G_\pi$ is continuous, even,
and strictly decreasing on $[0,\infty)$. Backward induction
gives an unrestricted Bayes rule that continues exactly when
$G_\pi(x)>c$, stopping at equality, with the policy
recommendations in Lemma~\ref{thm:dead_zone}. The posterior-mean martingale identity and
bounded support imply, for $x\geq0$,
\[
\small 
\begin{aligned}
0<G_\pi(x)
&=\mathbb E_\pi[
       \{-m_2(x,X_2)\}_+\mid X_1=x] \leq \bar\Delta\,
 \Phi\!\left(
   \frac{a_2\bar\Delta-a_1x}{\sqrt{a_2}}
 \right)
 \rightarrow0
 \qquad\text{as }x\to\infty.
\end{aligned}
\]
Consequently, the continuation threshold is
\[
\small 
\begin{aligned} 
 \tau_\pi(c)=
 \begin{cases}
  \infty,&c=0,\\
  s_1^{-1}
  \bigl(G_\pi|_{[0,\infty)}\bigr)^{-1}(c),
       &0<c<G_\pi(0),\\
  0,&c\geq G_\pi(0).
 \end{cases}
 \end{aligned} 
\]
Since $\pi$ is fixed as $c$ varies, this representation gives
all the continuity, monotonicity, and endpoint properties
in~(A). Strict monotonicity of $G_\pi$ and positivity of the
prior-predictive density imply that this threshold is the
unique minimizer defining $\tau_\pi(c)$.

\paragraph{Switching value.}
For $i\in\{M,\pi\}$, let
$\alpha_i(c)=2\Phi(-\tau_i(c))$.
Lemmas~\ref{cor:1} and~\ref{lem:tau_cost_unique} establish
that $\tau_M$ is continuous and nonincreasing, equals infinity
at zero cost, and converges to zero as $c \rightarrow \infty$.
Part~(A) establishes the same properties for $\tau_\pi$,
with $\tau_\pi(c)=0$ for $c\geq G_\pi(0)$.
Consequently, both $\alpha_M$ and $\alpha_\pi$ are continuous
and nondecreasing, start at zero, and equal one for all
sufficiently large costs. Equation~\eqref{eq:prior_weighted_threshold} gives
\[
 \alpha_\lambda(c):=2\Phi(-\tau_\lambda(c))
       =(1-\lambda)\alpha_M(c)+\lambda\alpha_\pi(c).
\]
The same properties hold for $\alpha_\lambda$.
For a finite observation $x$, write $p_x=2\Phi(-|x|/s_1)>0$.
The set $\{c\geq0:p_x\leq\alpha_\lambda(c)\}$ is a nonempty closed
upper interval. Its lower endpoint is $A_{\lambda,\pi}(x;s_1,w)$.
Therefore
\[
 \frac{|x|}{s_1}<\tau_\lambda(c)
 \quad\Longleftrightarrow\quad
 p_x>\alpha_\lambda(c)
 \quad\Longleftrightarrow\quad
 c<A_{\lambda,\pi}(x;s_1,w).
\]
This proves the switching characterization, including the endpoints
$\lambda=0,1$.

To establish the bounds in part~(B), fix $x$ and write
$p_x=2\Phi(-|x|/s_1)$. Recall that the combined rule continues
when $p_x>\alpha_\lambda(c)$.
If $c$ is below both endpoint switching values, both endpoint
rules continue. Hence $p_x$ exceeds both critical levels and
therefore also their weighted average, so the combined rule
continues. Conversely, if $c$ is at or above both endpoint
switching values, both endpoint rules stop. Then $p_x$ is
weakly below both critical levels and their weighted average,
so the combined rule also stops. Its switching cost must
therefore lie between the two endpoint switching costs:
$ 
\min\{A_{0,\pi}(x),A_{1,\pi}(x)\}
\leq A_{\lambda,\pi}(x)
\leq
\max\{A_{0,\pi}(x),A_{1,\pi}(x)\}.
$ 

\paragraph{Convexity of regret.} 
For $u\in[0,1]$, define
\[
 \widetilde R(u;\Delta,c)
 =\mathrm{Reg}\!\left(
   \Phi^{-1}\!\left(\frac{1+u}{2}\right);
   \Delta,c
 \right).
\]
The pointwise convexity calculation in the proof of
Lemma~\ref{lem:tau_cost_unique}, specifically
Equation~\eqref{eq:tau_cost_convexityidentity} and the
subsequent change of variables, shows that
$\widetilde R(\cdot;\Delta,c)$ is convex for every
$\Delta\neq0$, including its continuous extensions to
$u=0,1$. Here we use the normalization
$d=|\Delta|/s_1$, $\gamma=c/s_1$, and rescale regret by $s_1$.
At $\Delta=0$, $\widetilde R(u;0,c)=cu$ is affine.
Thus convexity holds for every $\Delta\in\Theta$ and $c\geq0$.

\paragraph{Regret bounds and the weight.}
Fix $c$ and put $u_M=2\Phi(\tau_M(c))-1$ and
$u_\pi=2\Phi(\tau_\pi(c))-1$.
The interpolation satisfies
$u_\lambda=2\Phi(\tau_\lambda(c))-1
  =(1-\lambda)u_M+\lambda u_\pi$.
Convexity therefore gives, for every $\Delta\in\Theta$,
\begin{equation}
 \mathrm{Reg}(\delta_c^\lambda;\Delta,c)
 \leq(1-\lambda)\mathrm{Reg}(\delta_c^0;\Delta,c)
        +\lambda\mathrm{Reg}(\delta_c^1;\Delta,c).
 \label{eq:prior_pointwise_bound}
\end{equation}
Integrating under $\pi$ gives
$B_\lambda(c)\leq(1-\lambda)B_0(c)+\lambda B_1(c)$.
Taking the supremum over $\Delta$ gives
$M_\lambda(c)\leq(1-\lambda)M_0(c)+\lambda M_1(c)$.

Finally, at this fixed cost, the functions
\[
 u\longmapsto\sup_{\Delta\in\Theta}\widetilde R(u;\Delta,c),
 \qquad
 u\longmapsto\int_\Theta\widetilde R(u;\Delta,c)\,\pi(d\Delta)
\]
are convex and minimized at $u_M$ and $u_\pi$, respectively.
Their restrictions to $u_\lambda=(1-\lambda)u_M+\lambda u_\pi$
are convex functions of $\lambda$, minimized at $0$ and $1$,
respectively. A convex function minimized at the left endpoint is
nondecreasing, and one minimized at the right endpoint is
nonincreasing. This proves the final assertion.

\section{Additional Applications}

\begin{table}[!ht]
\centering
\caption{Optimal thresholds and normalized costs for unconstrained $\Delta$.}
\label{tab:optimal_thresholds}
\begin{tabular}{ccc}
\toprule
\multicolumn{3}{l}{
    \textit{Panel A: Optimal threshold $\tau^\star$ for selected normalized costs}
} \\
\addlinespace[2pt]
$\displaystyle \frac{c}{s_1}$
&
\begin{tabular}{c}
Perfect follow-up \\[-2pt]
$s_2=0$
\end{tabular}
&
\begin{tabular}{c}
Equal precision \\[-2pt]
$s_2=s_1$
\end{tabular}
\\
\midrule
$0$     & $\infty$ & $\infty$ \\
$0.025$ & $5.564$  & $0.880$  \\
$0.050$ & $2.757$  & $0.564$  \\
$0.100$ & $1.295$  & $0.265$  \\
$0.150$ & $0.761$  & $0.109$  \\
$0.200$ & $0.469$  & $0.010$  \\
$0.250$ & $0.280$  & $0$      \\
$0.300$ & $0.145$  & $0$      \\
$0.400$ & $0$      & $0$      \\
\midrule
\multicolumn{3}{l}{
    \textit{Panel B: Normalized cost corresponding to $\tau^\star=1.96$}
} \\
\addlinespace[2pt]
Follow-up design
&
Target threshold
&
$\displaystyle \frac{c}{s_1}$
\\
\midrule
Perfect follow-up, $s_2=0$
&
$1.96$
&
$0.06918$
\\
Equal precision, $s_2=s_1$
&
$1.96$
&
$0.0015$
\\
\bottomrule
\end{tabular}

\begin{minipage}{0.92\textwidth}
\vspace{0.15cm}
\footnotesize
\textit{Notes:} Panel A reports the minimax-optimal 
threshold $\tau^\star$ for as a function of $c/s_1$. Panel B reports the cost that generates
$\tau^\star=1.96$, corresponding to a two-sided nominal significance
level of $5\%$.  
\end{minipage}
\end{table}

\subsection{Cash-transfer program: additional analysis}

\begin{figure}[!ht]
    \centering
    \includegraphics[scale = 0.5]{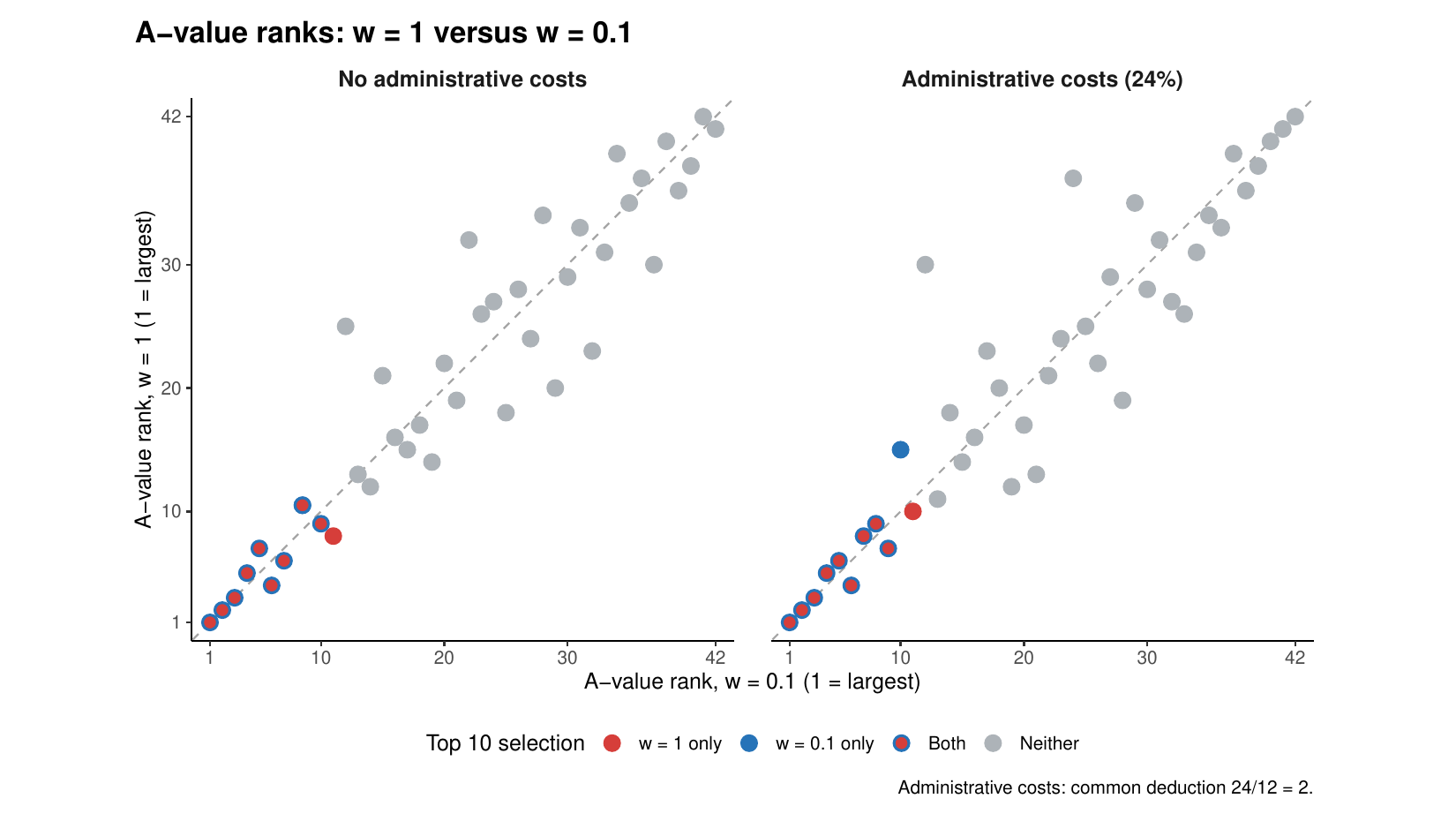}
    \caption{\footnotesize The figure illustrates that $A$-value ranking are strongly correlated as we vary the precision of the follow up study from $s_2/s_1 = 1$ to $s_2/s_1 = 0.1$. }
    \label{fig:different_precision}
\end{figure}

\begin{figure}[!ht]
\centering
\includegraphics[scale = 0.5]{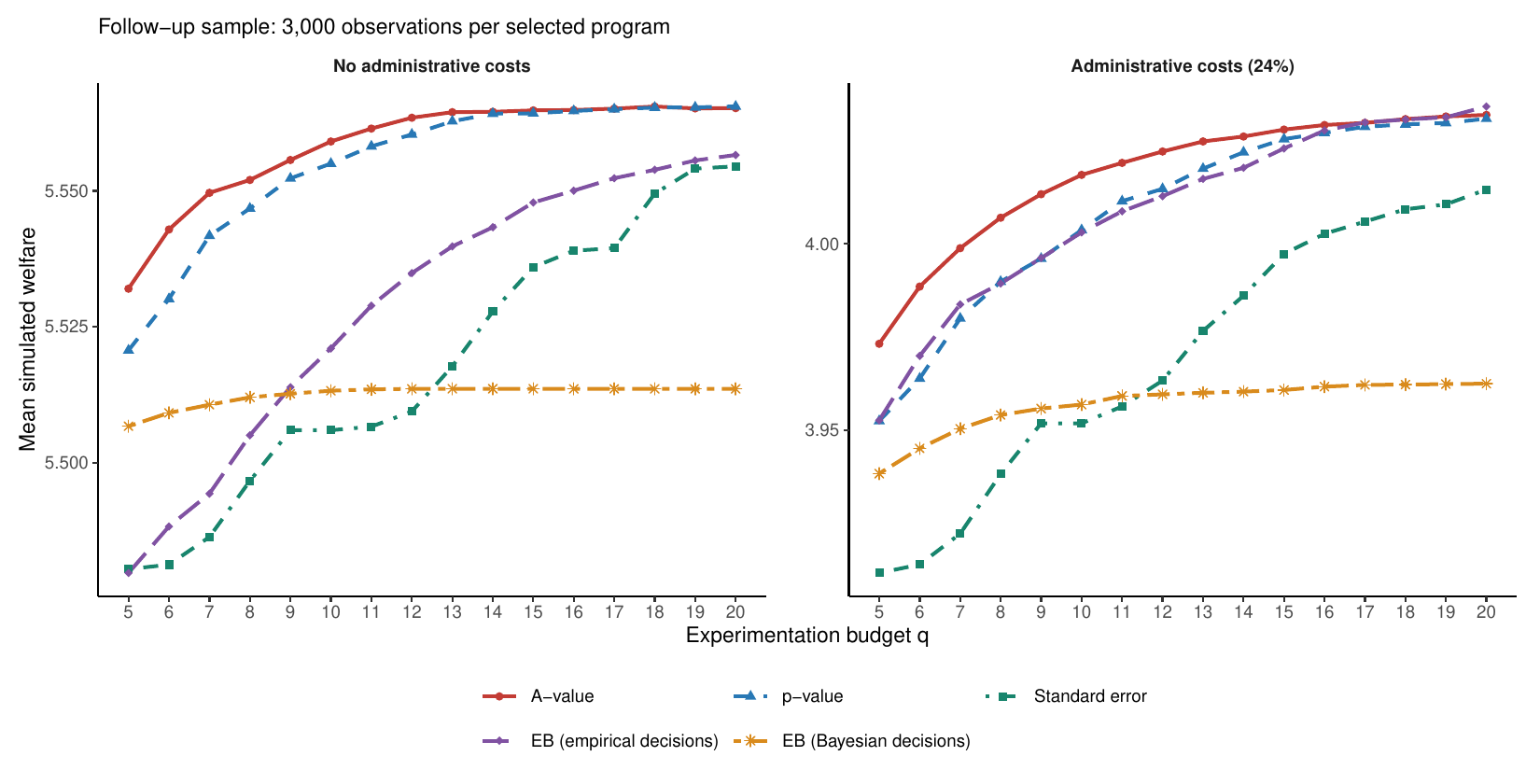}
  \caption{\footnotesize Calibrated simulations with 3,000
participants in each follow-up study. Initial standard errors
are $s_j$, and follow-up standard errors are calibrated as
$s_{2j}=s_j\sqrt{n_j/3000}$, where $n_j$ is the original total
sample size. For selected programs, empirical implementation
decisions use $1\{w_j^2X_{1jr}+X_{2jr}\geq0\}$, with
$w_j=s_{2j}/s_j$. Bayesian implementation uses the posterior
net mean calculated with the corresponding signal precisions.
Remaining simulation settings are as in
Figure \ref{fig:placeholder}.}\label{fig:different_precision2}
\end{figure}

Table~\ref{tab:uct_scale_cost_sensitivity} varies prospective transfer
allocations from \$5 million to \$20 million and the cost of each
follow-up study from \$250,000 to \$1 million, holding the welfare
horizon and follow-up precision fixed. Across the calibrations in
Table~\ref{tab:uct_scale_cost_sensitivity}, further experimentation
is warranted for between two and seven programs without
administrative costs and between two and six with administrative
costs. At a \$10 million transfer allocation and a \$500,000
follow-up cost, three programs qualify under each specification.

\begin{table}[!ht]
\centering
\caption{Cost-based cutoff for future experimentation.}
\label{tab:uct_scale_cost_sensitivity}
\small
\setlength{\tabcolsep}{4pt}
\begin{tabular}{@{}rrrrrrr@{}}
\toprule
\multicolumn{3}{c}{Calibration}
& \multicolumn{2}{c}{No administrative costs}
& \multicolumn{2}{c}{Administrative costs (24\%)} \\
\cmidrule(lr){1-3}\cmidrule(lr){4-5}\cmidrule(lr){6-7}
$NT$ & $C_{\mathrm{exp}}$ & $A$-value
& Studies & \shortstack{Minimum\\monetary $A$}
& Studies & \shortstack{Minimum\\monetary $A$} \\
(\$m) & (\$m) & cutoff & selected & (\$m) & selected & (\$m) \\
\midrule
5  & 0.25 & 0.4167 & 3 & 0.725 & 3 & 0.677 \\
5  & 0.50 & 0.8333 & 3 & 0.725 & 3 & 0.677 \\
5  & 1.00 & 1.6667 & 2 & 1.707 & 2 & 1.299 \\
\addlinespace
10 & 0.25 & 0.2083 & 4 & 0.289 & 6 & 0.259 \\
10 & 0.50 & 0.4167 & 3 & 1.450 & 3 & 1.355 \\
10 & 1.00 & 0.8333 & 3 & 1.450 & 3 & 1.355 \\
\addlinespace
20 & 0.25 & 0.1042 & 7 & 0.301 & 6 & 0.518 \\
20 & 0.50 & 0.2083 & 4 & 0.578 & 6 & 0.518 \\
20 & 1.00 & 0.4167 & 3 & 2.899 & 3 & 2.710 \\
\bottomrule
\end{tabular}
\par\vspace{0.5em}
\begin{minipage}{\linewidth}
\footnotesize\raggedright
\textit{Notes:} $NT$ denotes prospective transfers under implementation;
$C_{\mathrm{exp}}$ is the illustrative cost of one independent follow-up
study. We fix $H=12$ months and $w=1$. For each program, its evaluated
transfer $T_j$ can be retained by setting $N_j=NT/T_j$.
Follow-up is warranted when $A_j>100C_{\mathrm{exp}}/(NTM)$, equivalently
when its monetary $A$-value $NTHA_j/100$ exceeds $C_{\mathrm{exp}}$.
``Studies selected'' counts qualifying programs among all 42;
``Minimum monetary $A$'' is the smallest monetary value among them.
The administrative-cost specification uses
$X_{1j}^{\mathrm{admin}}=\widehat b_j-24/12=\widehat b_j-2$
for every program.
Selections and monetary equivalents are calculated using unrounded
program-level $A$-values and cutoffs.
Monetary amounts are in millions of 2010 USD PPP.
\end{minipage}
\end{table}

\subsection{An example for effect heterogeneity in a medical study}
\label{sec:application_ist}

We next illustrate the use of $A$-values to study effect heterogeneity in medical studies. For illustrative purposes we focus on the efficacy of aspirin on six-month survival after stroke. We use publicly available individual-level data from the International Stroke Trial, which enrolled 19,435 patients with clinically diagnosed acute ischaemic stroke across 36 countries \citep{sandercock2011}. We use the corrected second release of the dataset and compare assignment to aspirin (300 mg daily for up to 14 days) with assignment to avoid aspirin.\footnote{Link: \url{https://doi.org/10.7488/ds/104}.
We use the file \texttt{IST\_corrected.csv}.}  We have in total 19,285 observations after excluding missing values.
 We partition patients into eight mutually exclusive groups by crossing four baseline age bands with gender. These are 60 years or younger, 61-70, 71-80, and older than 80. For each group $j$, $X_{1j}$ is the difference in observed survival rates between the randomized aspirin and no-aspirin arms. Both are expressed in percentage points, with positive effects indicating increased survival.

\paragraph{Precision and $A$-values} Figure~\ref{fig:a_value_medical} reports the estimated survival effects and
$95\%$ confidence intervals. All eight intervals include zero,
but the value of additional evidence varies substantially across groups. We compute the $A$-value for a follow-up trial with fixed 1,000 participants, and the same design as the current study; here the choice of 1,000 participants approximately matches the average trial size for cardiovascular diseases in \cite{sertkaya2024costs} ($\sim$ 1,100 participants), and allows us to compare it to its corresponding average cost. It follows that the corresponding $A$-value for group $j$ reads as 
$
A_j = A_{\mathrm{value}}\Big(X_{1j};s_j, \sqrt{n_j/1000}\Big)
$
where $n_j$ is the sample size in the current data and $w_j = \sqrt{n_j/1000}$ is appropriately rescaled.

\paragraph{Monetary calibration} To calibrate these values to experimentation costs, suppose each
group's treatment effect applies to $N=10,000$/year once the drug is approved over a ten year horizon, abstracting from discounting.\footnote{We use a ten year horizong to match typical calibrations, e.g., \cite{bennette2016development}.} Here we calibrate this common population size to the counts in \citet{xian2022evaluation}.\footnote{Total
697,857 ischemic stroke patients prescribed aspirin during October 2012-June 2020. Annualizing over 7.75 years
and averaging across our eight age-by-sex groups gives
$697,857/(7.75\times8)=11,256 \approx 10,000$.} (See Table \ref{tab:ist1_population_cost_sensitivity} for sensitivity analysis.)
We value mortality-risk reductions using a common value of a
statistical life of $V=\$14.2$ million \citep{usdot2026vsl}.
Maintaining the conditions of Theorem~\ref{thm:a1}, and valuing each
dollar of experimental expenditure as one dollar of forgone welfare  (abstracting from costs of delay),
the minimax-regret rule recommends an independent
follow-up with 1,000 participants if and only if its total cost  for $H = 10$
\[
\small 
\begin{aligned} 
C_{\mathrm{exp},j}<\frac{NV H}{100}A_j
=1.42\times10^{10} A_j. 
\end{aligned} 
\]
The factor $1/100$ converts percentage points into probabilities.

\paragraph{Funding follow-up trials and making recommendations} We will compare $A$-values across different groups to the average cost threshold for cardiovascular diseases in \cite{sertkaya2024costs} to illustrate for which study it is worth to collect additional data (\$42 million for 1,000 participants in 2025 US\$ adjusting their figure in 2018 for inflation).
The five largest monetary $A$-values exceed this threshold. The rule therefore
recommends an immediate recommendation only for male above sevety and female between 60 and 70, and further experimentation on the remaining groups. These immediate decisions
arise even though the corresponding $95\%$ confidence intervals
include zero as shown in Figure \ref{fig:a_value_medical}, illustrating the importance of experimentation costs.

\begin{figure}[!ht]
    \centering
    \includegraphics[scale = 0.5]{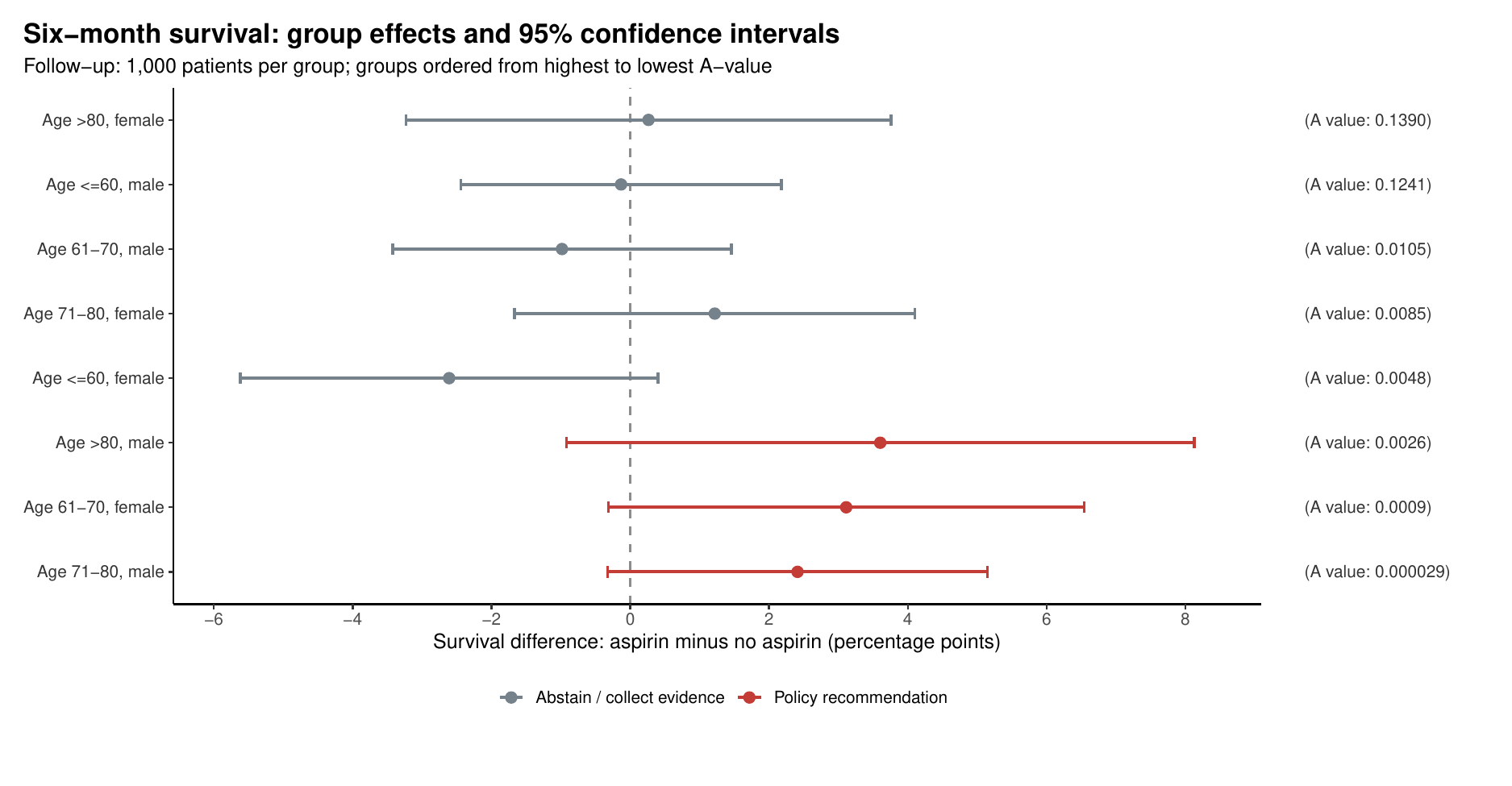}
    \caption{\footnotesize Six-month survival effects and $A$-values in the
    International Stroke Trial (IST-1). Points show the estimated
    aspirin-minus-no-aspirin difference in survival rates, in percentage
    points, for eight age-by-sex groups; horizontal bars indicate
    pointwise 95\% normal confidence intervals. Positive effects favor
    aspirin, and the dashed vertical line denotes zero effect.
    Groups are ordered by decreasing $A$-value, reported
    in parentheses. $A$-values assume independent
    follow-up evidence with 1,000 participants ($w=\sqrt{n_j/1,000}$). Decisions are made for six of the eight groups with $N=10,000$ future patients per
    group and experimentation costs of \$42 million  per 1,000 patient trial using data from \cite{sertkaya2024costs} annualized in 2025 US\$.}
    \label{fig:a_value_medical}
\end{figure}

\begin{table}[!ht]\centering\small
\caption{A-values and recommendations with 1,000-patient follow-up trials}
\label{tab:ist1_a_values_monetary_N1000}
\setlength{\tabcolsep}{5pt}
\scalebox{0.9}{\begin{tabular}{lrrrrrc}\toprule
Group & $n_j$ & $w_j$ & $A_j$ (pp) & Monetary $A_j$ & Calibrated cost & Recommend \\
 & & & & (\$ million) & (\$ million per trial) & aspirin? \\
\midrule
Women, age $>80$ & 2,967 & 1.722 & 0.1390 & 1973.484 & 42 & Abstain \\
Men, age $\leq60$ & 2,089 & 1.445 & 0.1241 & 1762.146 & 42 & Abstain \\
Men, age 61--70 & 2,920 & 1.709 & 0.0105 & 149.463 & 42 & Abstain \\
Women, age 71--80 & 3,288 & 1.813 & 0.0085 & 120.936 & 42 & Abstain \\
Women, age $\leq60$ & 1,016 & 1.008 & 0.0048 & 68.325 & 42 & Abstain \\
Men, age $>80$ & 1,795 & 1.340 & 0.0026 & 36.800 & 42 & Yes \\
Women, age 61--70 & 1,695 & 1.302 & 0.0009 & 13.353 & 42 & Yes \\
Men, age 71--80 & 3,515 & 1.875 & 0.000029 & 0.418 & 42 & Yes \\
\bottomrule\end{tabular}}\par\vspace{0.6em}

\begin{minipage}{0.95\linewidth}\footnotesize
$n_j$ is the complete-outcome sample across both randomized arms
of the original experiment. Each proposed follow-up has
$1{,}000$ patients in total. We evaluate $A_j$ at
$w_j=s_{2j}/s_{1j}=\sqrt{n_j/1{,}000}$.
Monetary values equal $NVA_j/100$, using full-precision $A_j$,
$N=100{,}000$ future beneficiaries per group and
$V=\$14.2$ million.
The common trial cost is calibrated to \$42 million per
1,000 patients using data from \cite{sertkaya2024costs},
inflation-adjusted to 2025 US\$.
Abstention means that the monetary A-value exceeds that cost;
otherwise recommendations follow the sign of the original
estimated effect.
\end{minipage}
\end{table}

\begin{table}[!ht]
\centering
\caption{Optimal experimentation as a function of costs in the medical application.}
\label{tab:ist1_population_cost_sensitivity}
\small
\setlength{\tabcolsep}{4pt}
\renewcommand{\arraystretch}{1.1}
\scalebox{0.8}{\begin{tabular}{rrrrrrr}
\toprule
$N$ & $C_{\mathrm{exp}}$ & $V$ & $n_{\mathrm{exp}}$
& \shortstack{Maximum\\break-even cost}
& \shortstack{Groups worth\\experimenting on}
& \shortstack{Minimum break-even cost\\among selected} \\
(patients) & (\$ million) & (\$ million) & (patients)
& (\$ million) & & (\$ million) \\
\midrule
10,000 & 10 & 14.2 & 1,000 & 197.348 & 4 & 12.094 \\
10,000 & 20 & 14.2 & 1,000 & 197.348 & 2 & 176.215 \\
10,000 & 30 & 14.2 & 1,000 & 197.348 & 2 & 176.215 \\
10,000 & 40 & 14.2 & 1,000 & 197.348 & 2 & 176.215 \\
\addlinespace
20,000 & 10 & 14.2 & 1,000 & 394.697 & 5 & 13.665 \\
20,000 & 20 & 14.2 & 1,000 & 394.697 & 4 & 24.187 \\
20,000 & 30 & 14.2 & 1,000 & 394.697 & 2 & 352.429 \\
20,000 & 40 & 14.2 & 1,000 & 394.697 & 2 & 352.429 \\
\addlinespace
30,000 & 10 & 14.2 & 1,000 & 592.045 & 6 & 11.040 \\
30,000 & 20 & 14.2 & 1,000 & 592.045 & 5 & 20.498 \\
30,000 & 30 & 14.2 & 1,000 & 592.045 & 4 & 36.281 \\
30,000 & 40 & 14.2 & 1,000 & 592.045 & 3 & 44.839 \\
\addlinespace
100,000 & 10 & 14.2 & 1,000 & 1973.484 & 7 & 13.353 \\
100,000 & 20 & 14.2 & 1,000 & 1973.484 & 6 & 36.800 \\
100,000 & 30 & 14.2 & 1,000 & 1973.484 & 6 & 36.800 \\
100,000 & 40 & 14.2 & 1,000 & 1973.484 & 5 & 68.325 \\
\bottomrule
\end{tabular}
}
\par\vspace{0.2cm}
\begin{minipage}{\linewidth}
\footnotesize
\textit{Notes:} $N$ is the future population affected by the decision
in each age-by-sex group. Each independent follow-up trial has
$n_{\mathrm{exp}}=1{,}000$ participants across both randomized arms;
$C_{\mathrm{exp}}$ is the cost of one such trial.
The value of a statistical life is fixed at $V=\$14.2$ million
\citep{usdot2026vsl}.
For group $j$, the break-even experimentation cost is
$C_j^\star(N)=NVA_j/100$, where $A_j$ is expressed in percentage
points and evaluated for the 1,000-participant follow-up.
Experimentation is worthwhile when $C_{\mathrm{exp}}<C_j^\star(N)$.
The maximum break-even cost is $\max_j C_j^\star(N)$:
at least one group warrants follow-up below this cost.
The final column reports the smallest $C_j^\star(N)$ among groups
satisfying the criterion.
\end{minipage}
\end{table}

\paragraph{Results for different cost calibrations} Table~\ref{tab:ist1_population_cost_sensitivity} examines the
sensitivity of our recommendations to the monetary calibration.
We vary the population affected by the treatment decision from
10,000 to 100,000 patients per group and experimentation costs from
\$10 million to \$40 million, holding the follow-up sample size
at 1,000 participants and the value of a statistical life fixed.
Experimentation remains worthwhile for women older than 80 and
men aged 60 or younger across most specifications; both groups
qualify at every considered cost. With 100,000 affected patients and
a \$10 million trial cost, the rule recommends experimentation
for the seven groups with the largest $A$-values, for six groups for cost between 10 and 30 million, and five groups for costs equal to 40 million. These comparisons
illustrate the flexibility of our framework to accommodate costs comparisons that can be pre-specified.

 \begin{rem}[Incorporating prior or model assumptions] Note that when available, it is possible to further incorporate prior or model information and adjust the $A$-value accordingly, an extension formally discussed in Appendix \ref{sec:prior_information}.  \qed 
\end{rem}

\end{document}